%% file: DoFMIMOCells_2013.tex
\documentclass[journal]{IEEEtranNew}
\usepackage[ top=0.7in, bottom=0.72in, left=0.65in, right=0.65in]{geometry}
\usepackage{graphicx}
\usepackage{psfrag}
\usepackage{pstool}
\usepackage{textcomp}
\usepackage{cite}
\usepackage{amssymb}
\usepackage{amsopn}
\usepackage{mathtools}

\usepackage{color}
\usepackage{hyperref}
\usepackage{url}
\usepackage{modamsthm}
\usepackage{aliascnt}
\usepackage{cancel}
\usepackage{multirow}
\usepackage[font=small]{caption}
\usepackage{subcaption}
\usepackage{flushend}
\makeatletter 
\renewcommand\subsubsection{\@startsection{subsubsection}{3}{\z@}%
                       {-18\p@ \@plus -4\p@ \@minus -4\p@}%
                       {4\p@ \@plus 2\p@ \@minus 2\p@}%
                       {\normalfont\normalsize\it\unboldmath
                        \rightskip=\z@ \@plus 8em\pretolerance=10000 }}

\newcommand*{\shifttext}[2]{%
  \settowidth{\@tempdima}{#2}%
  \makebox[\@tempdima]{\hspace*{#1}#2}%
}
\makeatother

\newaliascnt{thmcounter}{section}
\newaliascnt{defcounter}{section}
\newaliascnt{egcounter}{section}
\newaliascnt{conjcounter}{section}

\newtheorem{theorem}{Theorem}[thmcounter]
\newtheorem{lemma}{Lemma}[thmcounter]

\newtheorem{corollary}{Corollary}[thmcounter]

\newtheorem{definition}{Definition}[defcounter]

\usepackage{latexsym}
\usepackage{slashed}
\usepackage{centernot}

\newcommand{\bt}{\mathbf} 
\newcommand{\ms}{\hspace{0.5mm}}
\allowdisplaybreaks

\begin{document}

\title{Degrees of Freedom of MIMO Cellular Networks: Decomposition and Linear Beamforming Design}

\author{Gokul~Sridharan
        and~Wei~Yu
\thanks{The authors are with the The Edward S. Rogers Sr. Department of Electrical and Computer Engineering, University of Toronto, Toronto, ON M5S4G4, Canada e-mail: (gsridharan@comm.utoronto.ca, weiyu@comm.utoronto.ca).}
\thanks{This work was supported by the Natural Science and Engineering Research
Council (NSERC) of Canada. The material in this paper has been presented
in part at Canadian Workshop Inf. Theory, Jun. 2013, IEEE Int. Symp. Inf. Theory, Jul. 2013, and IEEE Global Commun. Conf., Dec. 2013. Manuscript submitted to IEEE Transactions on Information Theory on December 9, 2013.}
}

\maketitle

\markboth{}%
{}

\begin{abstract}
This paper investigates the symmetric degrees of freedom (DoF) of multiple-input multiple-output (MIMO) cellular networks with $G$ cells and $K$ users per cell, having $N$ antennas at each base station (BS) and $M$ antennas at each user. In particular, we investigate achievability techniques based on either decomposition with asymptotic interference alignment or linear beamforming schemes, and show that there are distinct regimes of $(G,K,M,N)$ where one outperforms the other. We first note that both one-sided and two-sided decomposition with asymptotic interference alignment achieve the same degrees of freedom. We then establish specific antenna configurations under which the DoF achieved using decomposition based schemes is optimal by deriving a set of outer bounds on the symmetric DoF. Using these results we completely characterize the optimal DoF of any $G$-cell network with each user having a single antenna. For linear beamforming schemes, we first focus on small networks and propose a structured approach to linear beamforming based on a notion called packing ratios. Packing ratio describes the interference footprint or shadow cast by a set of transmit beamformers and enables us to identify the underlying structures for aligning interference. Such a structured beamforming design can be shown to achieve the optimal spatially normalized DoF (sDoF) of two-cell two-user/cell network and the two-cell three-user/cell network. For larger networks, we develop an unstructured approach to linear interference alignment, where transmit beamformers are designed to satisfy conditions for interference alignment without explicitly identifying the underlying structures for interference alignment. The main numerical insight of this paper is that such an approach appears to be capable of achieving the optimal sDoF for MIMO cellular networks in regimes where linear beamforming dominates asymptotic decomposition, and a significant portion of sDoF elsewhere. Remarkably, polynomial identity test appears to play a key role in identifying the boundary of the achievable sDoF region in the former case.
\end{abstract}

\IEEEpeerreviewmaketitle

\input{intro.tex}

\input{systemmodel.tex}  

\section{Decomposition Based Schemes: Achievable DoF and Conditions for Optimality}  
\label{allAboutDecomp}
\input{decomposition.tex}

\input{outerbounds.tex}

\input{optimality.tex}

\input{keytakeaways.tex}
\section{Linear Beamforming: Structured Design}
\input{optimal-sDoF.tex}

\input{packingratios.tex}


\input{structureagnosticapproach.tex}

\section{Conclusion}
In this paper we investigate the DoF of MIMO cellular networks. In particular we establish the achievable DoF through the decomposition based approach and linear beamforming schemes. Through a new set of outer bounds, we establish conditions for optimality of the decomposition based approach. Through these outer bounds it is apparent that the optimal DoF of a general $G$-cell, $K$-users/cell network exhibits two distinct regimes, one where decomposition based approach dominates over linear beamforming and vice-versa. With regard to linear beamforming, we develop a structured approach to linear beamforming that is DoF-optimal in small networks such as the two-cell two-users/cell network and the two-cell three-users/cell network. We also develop an unstructured approach to linear beamforming that is applicable to general MIMO cellular networks, and through numerical experiments, show that such an approach is capable of achieving the optimal-sDoF for a wide class of MIMO cellular networks. 

Although the structured design of linear beamformers takes a disciplined approach to constructing beamformers, the wide applicability of the unstructured approach and its apparent ability to achieve the optimal sDoF in regimes where the sDoF curve exhibits a piecewise-linear behavior renders it highly attractive. The remarkable effectiveness of the unstructured approach warrants a deeper investigation on the role of randomization and that of the polynomial identity test in designing aligned beamformers.

\appendices



\section{DoF Outer Bound for the Two-Cell Three-Users/Cell Network When $\frac{5}{9}  \leq \gamma  < \frac{3}{4}$ }
\label{genieouterboundB}
\input{genieouterboundB.tex}

\section{Achievability of the Optimal sDoF for the Two-Cell Two-Users/Cell Network and the Two-Cell Three-Users/Cell Network}
\label{finerdetails}
\input{finerdetails.tex}

\bibliographystyle{IEEEtran}
\bibliography{IEEEabrv,ref_file}

\end{document}

%% file: intro.tex
\section{Introduction}
\label{section_intro}



Cellular networks are fundamentally limited by inter-cell interference. In this context, degrees of freedom (DoF) has emerged as a useful yet tractable metric in quantifying the extent to which interference can be mitigated through transmit optimization in time/frequency/spatial domains. In this work we study the DoF of multiple-input multiple-output (MIMO) cellular networks with $G$ cells and $K$ users/cell having $N$ antennas at each base station (BS) and $M$ antennas at each user---denoted in this paper as a $(G,K,M,N)$ network.

The study of DoF starts with the work on the two-user MIMO interference channel \cite{fakhereddin}. In \cite{maddah,jafarshamai}, the authors investigate the DoF of the $2\times 2$ $X$ network for which 
\textit{linear beamforming} based interference alignment is used to establish the optimal DoF.
This is followed by the landmark paper of \cite{cadambejafar}, where it is shown that the $K$-user single-input single-output (SISO) interference channel has $K/2$ DoF. The crucial contribution of \cite{cadambejafar} is an \textit{asymptotic scheme} for interference alignment over multiple symbol extensions in time or frequency for establishing the optimal DoF. This scheme requires channels to be time/frequency varying and crucially relies on the commutativity of diagonal channel matrices obtained from symbol extensions in time or frequency. The asymptotic scheme has been extended to MIMO cellular networks \cite{suh2} and MIMO $X$ networks \cite{cadambeXchannel}. We note that instead of relying on infinite symbol extensions over time or frequency varying channels, a signal space alignment scheme based on rational dimensions developed in \cite{motahari} achieves the same DoF using the scheme in \cite{cadambejafar}, but over constant channels. Since these early results, the asymptotic schemes of \cite{cadambejafar,motahari} and the linear beamforming schemes have emerged as the leading techniques for establishing the optimal DoF of various networks. 



In this work, we study the DoF achieved using the asymptotic schemes of \cite{cadambejafar,motahari} and the linear beamforming schemes along with conditions for their optimality in the context of MIMO cellular networks. Optimizing either scheme for general MIMO cellular networks is not straightforward. While the asymptotic schemes require the multi-antenna nodes in a MIMO network to be decomposed into independent single-antenna nodes, linear beamforming schemes require significant customization for each MIMO cellular network. This paper is motivated by the work of \cite{wangjafarisit2012}, which shows that for the $K$-user MIMO interference channel the two techniques have distinct regimes where one outperforms the other and that both play a critical role in establishing the optimal DoF. We observe that the same insight also applies to MIMO cellular networks, but the characterization of the optimal DoF is more complicated because of the presence of multiple users per cell. This paper makes progress on this front by studying the optimality of decomposition based schemes for a general $(G,K,M,N)$ network, and by developing two contrasting approaches to linear beamforming that emerge from two different perspectives on interference alignment. In a parallel and independent investigation, Liu and Yang \cite{liu-yang-arxiv} develop a new set of outer bounds on the DoF of MIMO cellular networks and a structured approach to characterize the optimal DoF under linear beamforming. While some of the results of this paper overlap with that of \cite{liu-yang-arxiv}, the approach taken in this paper for establishing these results is considerably different, and in some cases conceptually simpler than that of \cite{liu-yang-arxiv}.

\subsection{Literature Review}
\subsubsection{Decomposition Based Schemes }
The asymptotic scheme developed in \cite{cadambejafar} for the SISO $K$-user interference channel can be extended to other MIMO networks, including the $X$-network \cite{cadambeXchannel,huasunjafarisit2012}, and cellular networks \cite{suh,suh2} having the same number of antennas at each node. Since the original scheme in \cite{cadambejafar} relies on the commutativity of channel matrices, applying this scheme to MIMO networks requires decomposing multi-antenna nodes into multiple single-antenna nodes. Two-sided decomposition involves decomposing both transmitters and receivers into single-antenna nodes, while one-sided decomposition involves decomposing either the transmitters or the receivers. Once a network has been decomposed, the scheme in \cite{cadambejafar} can be applied to get an inner bound on the DoF of the original network.

Two-sided decomposition is first used to prove that the $K$-user interference channel with $M$ antennas at each node has $KM/2$ DoF \cite{cadambejafar}. This shows that the network is two-side decomposable, i.e., no DoF are lost by decomposing multi-antenna nodes into single antenna nodes. Two-sided decomposition is also known to achieve the optimal DoF of MIMO cellular networks with the same number of antennas at each node \cite{suh2}. In particular, it is shown that a $(G,K,N,N)$ network has $KN/(K+1)$ DoF/cell. However, for $X$-networks with $A$ transmitters and $B$ receivers having $N$ antennas at each node, two-sided decomposition is shown to be suboptimal and that one-sided decomposition achieves the optimal DoF of $ABN/(A+B-1)$ \cite{huasunjafarisit2012}. In \cite{gou,ghasemi}, the DoF of the $K$-user interference channel with $M$ antennas at the transmitters and $N$ antennas at the receivers is studied and the optimal DoF is established for some $M$ and $N$ (e.g., when $M$ and $N$ are 
such that $\frac{\max (M,N)}{\min (M,N)}$ is an integer) using the rational dimensions framework developed in \cite{motahari}. In \cite{wangjafarisit2012}, it is shown that decomposition based schemes achieve the optimal DoF of the $K$-user interference channel whenever $\tfrac{K-2}{K^2-3K+1} \leq \tfrac{M}{N} \leq 1$ for $K \geq 4$.  

\subsubsection{Linear Beamforming }
Linear beamforming techniques that do not require decomposition of multi-antenna nodes play a crucial role in establishing the optimal DoF of MIMO networks with different number of antennas at the transmitters and receivers. In particular, the work of Wang et al. \cite{chenweiwang} highlights the importance of linear beamforming techniques in achieving the optimal DoF of the MIMO three-user interference channel. In \cite{chenweiwang}, the achievability of the optimal DoF is established through a linear beamforming technique based on a notion called subspace alignment chains. A more detailed characterization of the DoF of the MIMO $K$-user interference channel is provided in \cite{wangjafarisit2012} where antenna configuration (values of $M$ and $N$) is shown to play an important role in determining whether the asymptotic schemes or linear beamforming schemes achieve the optimal DoF.

The study of the design and feasibility of linear beamforming for interference alignment without symbol extensions has received significant attention \cite{yetis,razaviyayn,tingtingliu,moewin,oscargonzalez,zhuangfeasibility}.
Designing transmit and receive beamformers for linear interference alignment is equivalent to solving a system of bilinear equations and a widely used necessary condition to check for the feasibility of linear interference alignment is to verify if the total number of variables exceeds the total number of constraints in the system of equations. If a system has more number of variables than constraints, it is called a proper system. Otherwise, it is called an improper system \cite{yetis}. In particular, when $d$ DoF/user are desired in a $(G,K,M,N)$ network, the system is said to be proper if $M+N\geq(GK+1)d$ and improper otherwise \cite{zhuangfeasibility}. While it is known that almost all improper systems are infeasible \cite{razaviyayn,tingtingliu}, feasibility of proper systems is still an area of active research. In \cite{razaviyayn,tingtingliu,moewin} a set of sufficient conditions for feasibility are established, while numerical tests to check for feasibility are provided in \cite{oscargonzalez}. 

While the optimality of linear beamforming for the $K$-user MIMO interference channel has been well studied, the role of linear beamforming in MIMO cellular networks having different number of antennas at the transmitters and receivers has not received significant attention. Partial characterization of the optimal DoF achieved using linear beamforming for two-cell networks are available in \cite{parklee,shin,ayoughi,leek}, while \cite{love} establishes a set of outer bounds on the DoF for the general $(G,K,M,N)$ network. Linear beamforming techniques to satisfy the conditions for interference alignment without symbol extensions are presented in \cite{love,leek,tingtingliuWCNC,yanjunma}.    

Characterizing linear beamforming strategies that achieve the optimal DoF for larger networks is challenging primarily because multiple subspaces can interact and overlap in complicated ways. Thus far in the literature, identifying the underlying structure of interference alignment for each given network  (e.g. subspace alignment chains for the three-user MIMO interference channel) has been a prerequisite for (a) developing counting arguments that expose the limitations of linear beamforming strategies, and (b) developing DoF optimal linear beamforming strategies. Concurrent to this work, significant recent progress has been made in \cite{liu-yang-arxiv} on characterizing the DoF of MIMO cellular networks. By identifying a genie chain structure, the optimality of linear beamforming is established for certain regimes of antenna configuration. In contrast to \cite{liu-yang-arxiv}, the current paper on one hand establishes a simpler structure called packing ratios for smaller networks, yet on the other hand, through numerical observation, establishes that even an unstructured approach can achieve the optimal DoF for a wide range of MIMO cellular networks, thus significantly alleviating the challenge in identifying structures in DoF-optimal beamformer design for larger networks.

\subsection{Main Contributions}
This paper aims to understand the DoF of MIMO cellular networks using both decomposition based schemes and linear beamforming. On the use of decomposition, we first note that both the asymptotic scheme of \cite{gou} for the MIMO interference channel and the asymptotic scheme of \cite{cadambeXchannel} for the $X$-network can be applied to MIMO cellular networks. Extending the scheme in \cite{gou} to MIMO cellular networks requires one-sided decomposition on the user side (multi-antenna users are decomposed to multiple single antenna users), while extending the scheme in \cite{cadambeXchannel} requires two-sided decomposition. More importantly, both approaches achieve the same degrees of freedom. In this paper, we develop a set of outer bounds on the DoF of MIMO cellular networks and use these bounds to establish conditions under which decomposition based approaches are optimal. The outer bounds that we develop are based on an outer bound for MIMO $X$-networks established in \cite{cadambeXchannel}. In particular we establish that for any $(G,K,M,N)$ network, $\max \big (\tfrac{M}{K\eta+1},\ \tfrac{N\eta}{K\eta+1} \big )$ is an outer bound on the DoF/user, where $\eta \in \left \{\frac{p}{q}: p \in \{ 1,2,\hdots , G-1\}, q \in \{ 1,2,\hdots , (G-p)K \} \right \}$.


In order to study linear beamforming strategies for MIMO cellular networks, similar in spirit to \cite{chenweiwang}, we allow for spatial extensions of a
given network and study the spatially-normalized DoF (sDoF). Spatial
extensions are analogous to time/frequency extensions where spatial
dimensions are added to the system through addition of antennas at the
transmitters and receivers. Unlike time or frequency extensions where
the resulting channels are block diagonal, spatial extensions assume generic channels with no additional structure---making them significantly easier to study without the peculiarities
associated with additional structure. Using the notion of sDoF, we first develop a structured approach to linear beamforming that is particularly useful in two-cell MIMO cellular networks. We then focus on an unstructured approach to linear beamforming that can be applied to a broad class of MIMO cellular networks.

\textit{ Structured approach to linear beamforming}: This paper develops linear beamforming strategies that achieve the optimal sDoF of two-cell MIMO cellular networks with two or three users per cell. We characterize the optimal sDoF/user for all values of $M$ and $N$ and show that the optimal sDoF is a piecewise-linear function, with  either $M$ or $N$ being the bottleneck. We introduce the notion of \textit{packing ratio} that describes the interference footprint or shadow cast by a set of uplink transmit beamformers and exposes the underlying structure of interference alignment. Specifically, the packing ratio of a given set of beamformers is the ratio between the number of beamformers in the set and the number of dimensions these beamformers occupy at an interfering base-station (BS).

Packing ratios are useful in determining the extent to which interference can be aligned at an interfering BS. For example, for the two-cell, three-user/cell MIMO cellular network, when $M/N \leq 2/3$, the best possible packing ratio is $2\!:\!1$, i.e., a set of two beamformers corresponding to two users aligns onto a single dimension at the interfering BS. This suggests that if we have sufficiently many such sets of beamformers, no more than $2N/3$ DoF/cell are possible. This in fact turns out to be a tight upper bound whenever $\tfrac{5}{9} \leq \tfrac{M}{N} \leq \tfrac{2}{3}$. Through the notion of packing ratios, it is easier to visualize the achievability of the optimal sDoF using linear beamforming and the exact cause for the alternating behavior of the optimal sDoF where either $M$ or $N$ is the bottleneck becomes apparent. In particular, we establish the sDoF of two-cell networks with two or three users/cell.

\textit{Unstructured approach to linear beamforming:} In order to circumvent the bottleneck of identifying the underlying structure of interference alignment and to establish results for a broad set of networks, this paper proposes a structure agnostic approach to designing linear beamformers for interference alignment. In such an approach, depending on the DoF demand placed on a given MIMO cellular network, we first identify the total number of dimensions that are available for interference at each BS. We then design transmit beamformers in the uplink by first constructing a requisite number of random linear vector equations that the interfering data streams at each BS are required to satisfy so as to not exceed the limit on the total number of dimensions occupied by interference. This system of linear equations is then solved to obtain a set of aligned transmit beamformers. 

The crucial element in such an approach is the fact that we construct linear vector equations with random coefficients. This is a significant departure from typical approaches to constructing aligned beamformers where the linear equations that identify the alignment conditions emerge from notions such as subspace alignment chains or packing ratios and are predefined with deterministic coefficients. The flexibility in choosing random coefficients allows us to use this technique for interference alignment in networks of any size, without having to explicitly infer the underlying structure.

Such an approach is also discussed in a limited context in \cite{tingtingliuWCNC} where it is used to design aligned transmit beamformers when only 1 DoF/user is desired. We significantly expand the scope of such an approach by proposing the use of a polynomial identity test to resolve certain linear independence conditions that need to be satisfied when more than 1 DoF/user are desired. In our work we  outline the key steps to designing aligned transmit beamformers using this approach and take a closer look at the DoFs that can be achieved. We then proceed to numerically examine the optimality of the DoF achieved through such a scheme. Numerical evidence suggests that for any given $(G,K,M,N)$ network, the unstructured approach to linear beamforming achieves the optimal sDoF whenever $M$ and $N$ are such that the decomposition inner bound $\big (\tfrac{MN}{KM+N} \big)$ lies below the proper-improper boundary $\big (\tfrac{M+N}{GK+1} \big)$. Remarkably, the polynomial identity test plays a key role in identifying the optimal sDoF in this regime. 

\subsection{Paper Organization}
The presentation in this paper is divided into two main parts. The first part, presented in Section \ref{allAboutDecomp},
 discusses the achievable DoF using decomposition based approaches, establishes outer bounds on the DoF of MIMO cellular networks, and identifies the conditions under which such an approach is DoF optimal. In the second part, we present a structured and an unstructured approach to linear beamforming design for MIMO cellular networks. In particular, in Section \ref{SAP}, we establish the optimal sDoF of the two-cell MIMO network with two or three users per cell through a linear beamforming strategy based on packing ratios. Section \ref{USAP} introduces the unstructured approach to interference alignment and explores the scope and limitations of such a technique in achieving the optimal sDoF of any $(G,K,M,N)$ network.
\subsection{Notation}
We represent all column vectors in bold lower-case letters and all
matrices in bold upper-case letters. The conjugate transpose and Euclidean norm of vector $\bt v$
are denoted as $\bt v^H$ and $\|\bt v \|$, respectively. Calligraphic letters
(e.g., $\mathcal{Q}$) are used to denote sets. The column span of the columns of a matrix $\bt M$ is denoted as $\mathrm{span}(\bt M)$.

%% file: systemmodel.tex
\section{System Model}

    Consider a network with $G$ interfering cells with $K$ users in each cell, as shown in Fig. \ref{toyfigure}. Each user is assumed to have $M$ antennas and each BS is assumed to have $N$ antennas. The index pair $(j,l)$ is used to denote the $l$th user in the $j$th cell. The channel from user $(j,l)$ to the $i$th BS is denoted as the $N\times M$ matrix $\bt H_{(jl,i)}$. We assume all channels to be generic and time varying. In the uplink, user $(j,l)$ is assumed to transmit the $M\times 1$ signal vector $\bt x_{jl}(t)$ in time slot $t$. The transmitted signal satisfies the average power constraint, $\frac{1}{T}\sum_{t=1}^T \|\bt x_{ij}(t)\|^2\leq \rho$. The resulting received signal at the $i$th BS can be written as
    \begin{align}
	\bt y_i= \sum_{j=1}^{G}\sum_{l=1}^K\bt H_{(jl,i)}\bt x_{jl}+\bt n_i,
    \end{align}
    where $\bt y_i$ is an $N \times 1$ vector and $\bt n_i$ is the $N\times 1$ vector representing circular symmetric additive white Gaussian noise $ \sim \mathcal {CN} (\bt 0,\bt I)$.
    The received signal is defined similarly for the downlink.
    
    Suppose the transmit signal vector is formed through a $M \times d$ linear transmit beamforming matrix $\bt V_{jl}$ and received using a $N \times d$ receive beamforming matrix $\bt U_{jl}$, where $d$ represents the number of transmitted data streams per user, then the received signal can be written as

    \begin{equation}
	\bt y_i= \sum_{j=1}^{G}\sum_{l=1}^K\bt H_{(jl,i)}\bt{V}_{jl}\bt{s}_{jl}+\bt n_i,
	\end{equation}
    where $\bt{s}_j$ is the $d\times 1$ symbol vector transmitted by user $(j,l)$. We denote the space occupied by interference at the $i$th BS as the column span of a matrix $\bt R_i$ formed using the column vectors from the set $\{\bt{H}_{(jl,i)}\bt{v}_{jlk}: j\in \{1,2,\hdots ,G\},\ l \in \{1,2,\hdots ,K \},\ k \in \{1,2,\hdots ,d\},\ j\neq i \}$, where we use the notation $\bt v_{jlk}$ to denote the $k$th beamformer associated with user $(j,l)$.
    
 To recover the signals transmitted by user $(i,l)$, the signal received by the $i$th BS is processed using the receive beamformer $\bt{U}_{il}$ and the received signal after this step can be written as
    \begin{align}
	\bt{U}_{il}^H\bt{y}_i&=\sum_{j=1}^{G}\sum_{l=1}^K \bt{U}_{il}^H \bt H_{(jl,i)}\bt{V}_{jl}\bt{s}_{jl}+\bt{U}_{il}^H\bt n_i.
	\label{eq_uhv}
    \end{align}

The information theoretic quantity of interest is the degrees of freedom. In particular, the total degrees of freedom of a network is defined as
\begin{align}
 \limsup_{\rho \rightarrow \infty}\left [\sup_{\{R_{ij}(\rho) \}\in \mathcal{C}(\rho)} \frac{\big ( R_{11}(\rho)+ R_{12}(\rho)+ \hdots + R_{GK}(\rho) \big )}{\log (\rho)} \right ] \nonumber 
\end{align}
where $\rho$ is the signal-to-noise (SNR) ratio, $\{R_{ij}(\rho) \}$ is an achievable rate tuple for a given SNR where $R_{ij}$ denotes the rate to user $(i,j)$, and $\mathcal{C}(\rho)$ is the capacity region for a given SNR. As is evident, the sum-DoF of a network is the pre-log factor at which sum-capacity scales as transmit power is increased to infinity. Informally, it is the total number of interference free directions that can be created in a network. Due to the symmetry in the network under consideration, maximizing the sum-DoF is equivalent to maximizing the DoF/user or DoF/cell. The maximum DoF/user that can be achieved in a network is also referred to as the symmetric DoF of a network. This paper focuses on characterizing the optimal symmetric DoF of MIMO cellular networks.
 \begin{figure}
\begin{center}
  \includegraphics[width=3.5in]{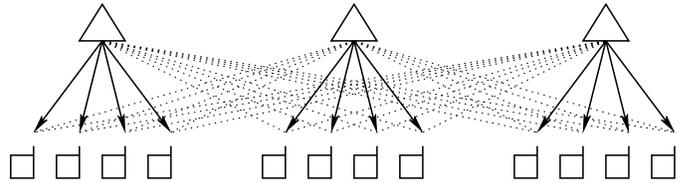}
\caption{Figure representing a cellular network having three mutually interfering cells with four users per cell.}
\label{toyfigure}
\end{center}
 \end{figure}

%% file: decomposition.tex
In this section we discuss the DoF/user that can be achieved in a MIMO cellular network using the asymptotic scheme presented in \cite{cadambejafar} and establish the conditions under which such an approach is DoF optimal.

\subsection{Achievable DoF using decomposition based schemes}
    \label{decomposition}
	 Applying the asymptotic scheme in \cite{cadambejafar} to a MIMO network requires us to decompose either the transmitters or the receivers, or both, into independent single-antenna nodes. When using the asymptotic scheme on the decomposed network, the DoF achieved per user in the original network is simply the sum of the DoFs achieved over the individual single-antenna nodes. 

	One-sided decomposition of a $(G,K,M,N)$ cellular network on the user side reduces the network to a $G$-cell cellular network with $KM$ single antenna users per cell. Since user-side decomposition of both the MIMO interference channel and the MIMO cellular network results in a MISO cellular network, the results of \cite{gou,ghasemi} naturally extend to MIMO cellular networks. Two-sided decomposition of a $(G,K,M,N)$ cellular network results in $GN$ single-antenna BSs and $KM$ single-antenna users, which form a $GN \times GKM$ $X$-network with a slightly different message requirement than in a traditional $X$-network since each single-antenna user is interested in a message from only $N$ of the $GN$ single-antenna BSs. The asymptotic alignment scheme developed in \cite{cadambeXchannel} for $X$-networks can also be applied to this $GN \times GKM$ $X$-network. It turns out that one-sided decomposition and two-sided decomposition achieve the same DoF in a $(G,K,M,N)$ network. Using the results in \cite{gou, ghasemi, cadambeXchannel}, the achievable DoF for general MIMO cellular networks using decomposition based schemes is stated in the following theorem.
	\begin{theorem}
	For the $(G,K,M,N)$ cellular network, using one-sided decomposition on the user side or two-sided decomposition, $\frac{KMN}{KM+N}$ DoF/cell are achievable when $(G-1)KM\geq N$. 
	\label{onesided_theorem}
	\end{theorem}
	Note that when $(G-1)K < N$, there is no scope for interference alignment and random transmit beamforming in the uplink turns out to be the DoF optimal strategy. The proof of this theorem follows from a straightforward application of the results in \cite{gou, ghasemi, cadambeXchannel} and is omitted here. This theorem generalizes the result established in \cite{suh2}, where it is shown that SISO cellular networks with $K$-users/cell have $K/K+1$ DoF/cell. By duality of linear interference alignment, this result applies to both uplink and downlink. While we consider decomposing multi-antenna users into single-antenna users for one-sided decomposition here, we can alternately consider decomposing the multi-antenna BSs. It can however be shown that the achievable DoF remains unchanged. Designing the achievable scheme is similar to \cite{huasunjafarisit2012}, where separation between signal and interference is no longer implicitly assured.

%% file: outerbounds.tex
\subsection{Outer Bounds on the DoF of MIMO Cellular Networks}
    \label{outerbounds}  
	We derive a new set of outer bounds on the DoF of MIMO cellular networks that are based on a result in \cite{cadambeXchannel}, where MIMO $X$-networks with $A$ transmitters
and $B$ receivers are considered. By focusing on the set of messages originating from or intended for a transmitter-receiver pair and splitting the total messages in the network into $AB$
sets, \cite{cadambeXchannel} derives a bound on the total DoF of this set of messages. Let $d_{i,j}$ represent the DoF between the $i$th transmitter and the $j$th receiver. The
following lemma presents the outer bound obtained in this manner.
	
	\begin{lemma}[$\!$\cite{cadambeXchannel}$\, $] In a wireless $X$-network with $A$ transmitters and $B$ receivers, the DoF of all messages originating at the $a$th transmitter and the DoF of
all the messages intended for the $b$th receiver are bounded by
	\begin{align}
	\sum_{i=1}^B d_{a,i}+\sum_{j=1}^A d_{j,b}-d_{a,b} \leq \max (M,N),
	\end{align}
	where $M$ is the number of antennas at the $a$th transmitter and $N$ is the number of antennas at the $b$th receiver. By symmetry, this bound also holds when the direction of communication is reversed.
	\label{dofx_lemma}
	\end{lemma}

	Before we proceed to establish outer bounds on the DoF of a MIMO cellular network, we define the set $\mathcal{Q}$ as
\begin{equation}
 \mathcal Q=\left \{\frac{p}{q}: p \in \{ 1,2,\hdots , G-1\}, q \in \{ 1,2,\hdots , (G-p)K \} \right \}.
\end{equation}
 The following theorem presents an outer bound on the DoF.
	\begin{theorem}
	If a $(G,K,M,N)$ network satisfies $M/N\leq p/q$, for some $p/q \in \mathcal{Q}$, then $Np/(Kp+q)$ is an outer bound on the DoF/user of that network. Further, if $M/N\geq p/q$, for some $p/q \in \mathcal{Q}$, then $Mq/(Kp+q)$ is an outer bound on the DoF/user of that network.
	\label{1byq_bound}
	\end{theorem}
	\begin{proof}
	To prove this theorem, we first note that a cellular network can be regarded as an $X$-network with some messages set to zero. Further, Lemma \ref{dofx_lemma} is applicable even when
some messages are set to zero. Now, suppose $\frac{M}{N}\leq \frac{p}{q}$ for some $\frac{p}{q} \in \mathcal{Q}$, then consider a set of $p$ cells and allow the set of BSs in these $p$ cells to cooperate fully. Let $\mathcal{B}$ denote the set of indices corresponding to the $p$ chosen cells. From the remaining $G-p$ cells, we pick $q$ users and denote the set of indices corresponding to these users as $\mathcal{U}_{\bar{\mathcal{B}}}$ and allow them to cooperate fully.

Applying Lemma \ref{dofx_lemma} to the set of BSs $\mathcal{B}$ and the set of users $\mathcal{U}_{\bar{\mathcal{B}}}$, we get

	\begin{align}
	\sum_{i\in{\mathcal{B}}} \sum_{j=1}^K d_{ij,i}+\sum_{(g,h)\in \mathcal{U}_{\bar{\mathcal{B}}} }d_{gh,g} \leq \max(pN,qM).
	\end{align}

By summing over similar bounds for all the ${G}\choose{p}$ sets of $p$ BSs and the corresponding ${(G-p)K}\choose{q}$ sets of $q$ users for each set of $p$ BSs, we obtain


	\begin{align}
	 \left [\frac{K}{q}+\frac{1}{p}  \right ] \sum_{i=1}^G \sum_{j=1}^K d_{ij,i} \leq & \frac{GK}{pq} \max(pN,qM) \nonumber \\
	 \Rightarrow \sum_{i=1}^G \sum_{j=1}^K d_{ij,i} \leq&  \frac{GK}{Kp+q} \max(pN,qM)=pN.
	\end{align}



	Thus, the total DoF in the network is bounded by $\frac{GKNp}{Kp+q}$. Hence, DoF/user $\leq \frac{Np}{Kp+q}$ whenever $p/q \in \mathcal{Q}$. The outer bound is established in a similar manner when $\frac{M}{N}\geq\frac{p}{q}$. Note that whenever $\frac{M}{N}=\frac{p}{q}$,$\frac{Np}{Kp+q}=\frac{Mq}{Kp+q}=\frac{MN}{KM+N}$. 
	\end{proof}
		
	In \cite{love}, outer bounds on the DoF for MIMO cellular network are derived which are also based on the idea of creating multiple message sets \cite{cadambeXchannel}. The DoF/user of a $(G,K,M,N)$ network is shown to be bounded by
	\begin{align}
	\text{DoF/user} \leq \min \left ( M,\tfrac{N}{K},\tfrac{\max[KM,(G-1)N]}{K+G-1}, \tfrac{\max[N,(G-1)M]}{K+G-1} \right ).
	\label{love_bounds}
	\end{align}
	While it is difficult to compare this set of bounds and the bounds in Theorem \ref{1byq_bound} over all parameter values, we can show that under certain settings the bounds obtained in Theorem \ref{1byq_bound} are tighter. For example, since $p/q \in \mathcal{Q}$, let us fix $p/q=1/K$, then set $M/N=p/q=1/K$. Further, let us assume that $(G-1)<K$.
Under such conditions, (\ref{love_bounds}) bounds the DoF/user by $\tfrac{MK}{K+G-1}$ while Theorem \ref{1byq_bound} states that DoF/user $\leq \tfrac{M}{2}$. Since we have assumed $K > G-1$, it is
easy to see that the latter bound is tighter.

%% file: optimality.tex
\subsection{Optimality of the DoF Achieved Using Decomposition}
\label{optimality}
    Using the results in preivious two sections, we establish conditions for the optimality of one-sided and two-sided decomposition of MIMO cellular networks in
the following theorem.
    \begin{theorem}
    The optimal DoF for any $(G,K,M,N)$ network with $\frac{M}{N} \in \mathcal{Q}$ is
$\frac{MN}{KM+N}$ DoF/user. The optimal DoF is achieved by either one-sided or two-sided decomposition with asymptotic interference alignment.
    \label{optimal_theorem}
    \end{theorem}

    This result follows immediately from Theorems \ref{onesided_theorem} and \ref{1byq_bound}. We observe that this result is analogous to the results in \cite{gou,ghasemi} where it is shown that the $G$-user interference channel has $\frac{MN}{M+N}$ DoF/user whenever $\eta=\frac{\max (M,N)}{\min (M,N)}$ is an integer and $G > \eta$. It is easy to see that the results of \cite{gou,ghasemi} can be easily recovered from the above theorem by setting $K=1$ and letting $G$ represent the number of users in the interference channel.

    The result in Theorem \ref{optimal_theorem} has important consequences for cellular networks with single-antenna users. The following corollary describes the optimal DoF/user of any cellular network with single antenna users that satisfies $(G-1)K \geq N$.
    \begin{corollary}
    The optimal DoF of a $(G,K,M=1,N)$ network with $(G-1)K \geq N$, is $\frac{N}{K+N}$ DoF/user.
    \label{miso_theorem}
    \end{corollary}

 \begin{figure}
\begin{center}
  \includegraphics[width=3.2in]{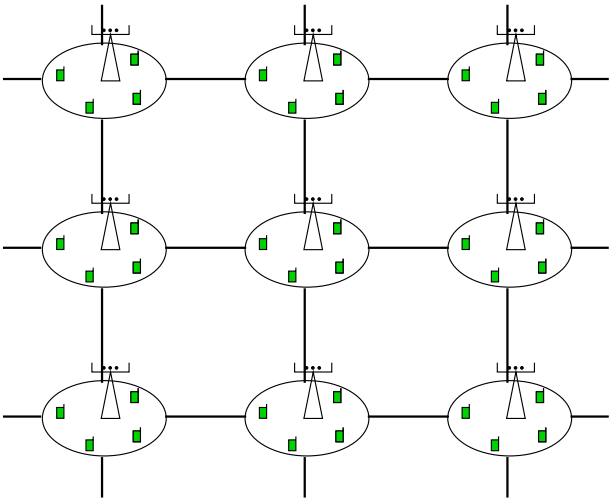}
\caption{Figure showing the 2-D Wyner model of a cellular network. Two cells are connected to each other if they mutually interfere. Cells that are not directly connected to each other are assumed to see no interference from each other. Note that each user in a given cell sees interference from the four adjacent BSs.}
\label{2dwyner}
\end{center}
 \end{figure}

    For example, this corollary states that a three-cell network having four single-antenna users per cell and four antennas at each BS has $1/2$ DoF/user. Using this corollary and the DoF achieved using zero-forcing beamforming, the optimal DoF of cellular networks with single-antenna users can be completely characterized and is stated in the following theorem.

    \begin{theorem}
    The DoF of a G-cell cellular network with $K$ single-antenna users per cell and $N$ antennas at each BS is given by
    \begin{align}
    \text{DoF/user}=\begin{cases}
                     \frac{N}{N+K} & N < (G-1)K \\
		     \frac{N}{GK} & (G-1)K \leq N <GK \\
		     1  &  N \geq GK 
		    \end{cases}.
    \end{align} The optimal DoF is achieved through zero-forcing beamforming when $N \geq (G-1)K$ and through asymptotic interference alignment when $N < (G-1) K$.
    \end{theorem}

    Another interesting consequence of Theorem \ref{optimal_theorem} for two-cell cellular networks is stated in the following corollary.
	
    \begin{corollary}
    For a $(G=2,K,M,N)$ cellular network with $K=\frac{N}{M}$, time sharing across cells is optimal and the optimal DoF/user is $\frac{N}{2K}$.
    \end{corollary}
    \begin{proof}
    Using Theorem \ref{optimal_theorem}, the optimal DoF/user of this network is $\frac{N}{2K}$. Since the $K$-user MAC/BC with $\frac{M}{N}=\frac{1}{K}$ has $\frac{N}{K}$ DoF/user, accounting for time sharing between the two cells gives us the required result.
    \end{proof}

    This result recovers and generalizes a similar result obtained in \cite{parklee} for two-cell MISO cellular networks, which shows that in dense cellular networks where $K=N/M$, when two closely located cells cause significant interference to each other, simply time sharing between the two mutually interfering BSs is a DoF-optimal way to manage interference in the network. This result can be further extended to the 2-D Wyner model for MIMO cellular networks and is stated in the following corollary.

\begin{figure*}[t]
        \centering
        \begin{subfigure}[b]{0.3\textwidth}
		\psfrag{xlabel}[tc][tc][0.8]{$M/N$}
		\psfrag{0.05}[tc][tc][0.65]{}
		\psfrag{0.15}[tc][tc][0.65]{}
		\psfrag{0.25}[tc][tc][0.65]{}
		\psfrag{0.35}[tc][tc][0.65]{}
		\psfrag{0.45}[tc][tc][0.65]{}
		\psfrag{0}[tr][tr][0.65]{}
		\psfrag{0.1}[tr][tr][0.8]{0.1}
		\psfrag{0.2}[tr][tr][0.8]{0.2}
		\psfrag{0.3}[tr][tr][0.8]{0.3}
		\psfrag{0.4}[tr][tr][0.8]{0.4}
		\psfrag{0.5}[tr][tr][0.8]{0.5}
		\psfrag{1}[tr][tr][0.8]{1}
		\psfrag{2}[tr][tr][0.8]{2}
		\psfrag{3}[tr][tr][0.8]{3}
		\psfrag{1.5}[tr][tr][0.65]{}
		\psfrag{2.5}[tr][tr][0.65]{}
		\psfrag{ylabel}[Bc][Bl][0.8][180]{ $\begin{matrix} \text{DoF/user/N} \\ \vspace{0.3cm} \end{matrix}$}
		\psfrag{BSs=2, Users=2}[bc][bc][0.9]{2 cells and 2 users/cell}
		\includegraphics[width=2in,height=1.6in]{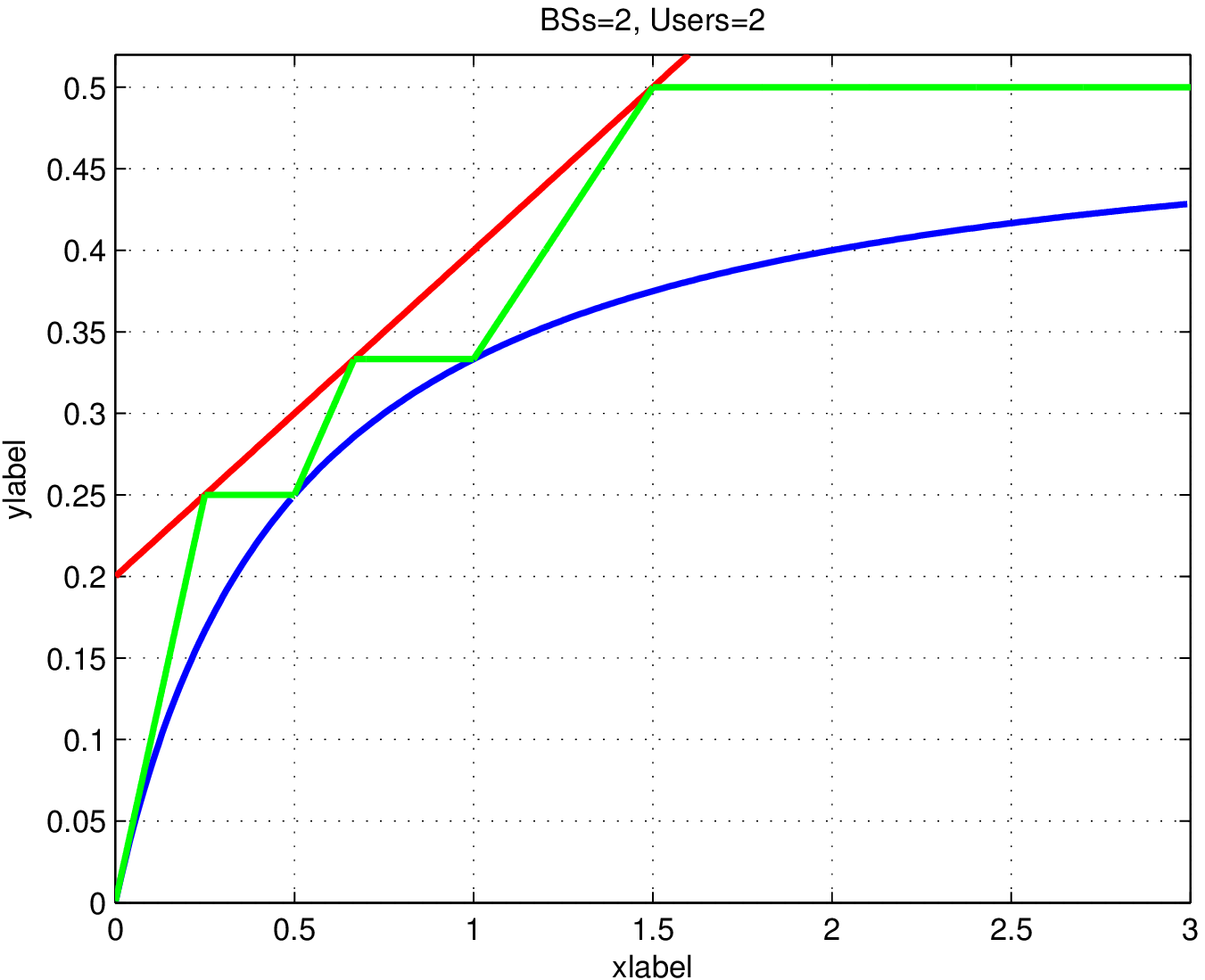}
		\caption{}
		\label{fig:B2U2}

        \end{subfigure}%
        ~ 
        \begin{subfigure}[b]{0.3\textwidth}
		\psfrag{xlabel}[tc][tc][0.8]{$M/N$}
		\psfrag{0.05}[tc][tc][0.65]{}
		\psfrag{0.15}[tc][tc][0.65]{}
		\psfrag{0.25}[tc][tc][0.65]{}
		\psfrag{0.35}[tc][tc][0.65]{}
		\psfrag{0.45}[tc][tc][0.65]{}
		\psfrag{0}[tc][tc][0.65]{}
		\psfrag{0.1}[tr][tr][0.8]{0.1}
		\psfrag{0.2}[tr][tr][0.8]{0.2}
		\psfrag{0.3}[tr][tr][0.8]{0.3}
		\psfrag{0.4}[tr][tr][0.8]{0.4}
		\psfrag{0.5}[tr][tr][0.65]{}
		\psfrag{1}[tr][tr][0.8]{1}
		\psfrag{2}[tr][tr][0.8]{2}
		\psfrag{3}[tr][tr][0.8]{3}
		\psfrag{1.5}[tr][tr][0.65]{}
		\psfrag{2.5}[tr][tr][0.65]{}
		\psfrag{ylabel}[cc][cl][0.8][180]{ $\begin{matrix} \text{DoF/user/N} \\ \vspace{0.3cm} \end{matrix}$}
		\psfrag{BSs=2, Users=3}[bc][bc][0.9]{2 cells and 3 users/cell}
		\includegraphics[width=2in,height=1.6in]{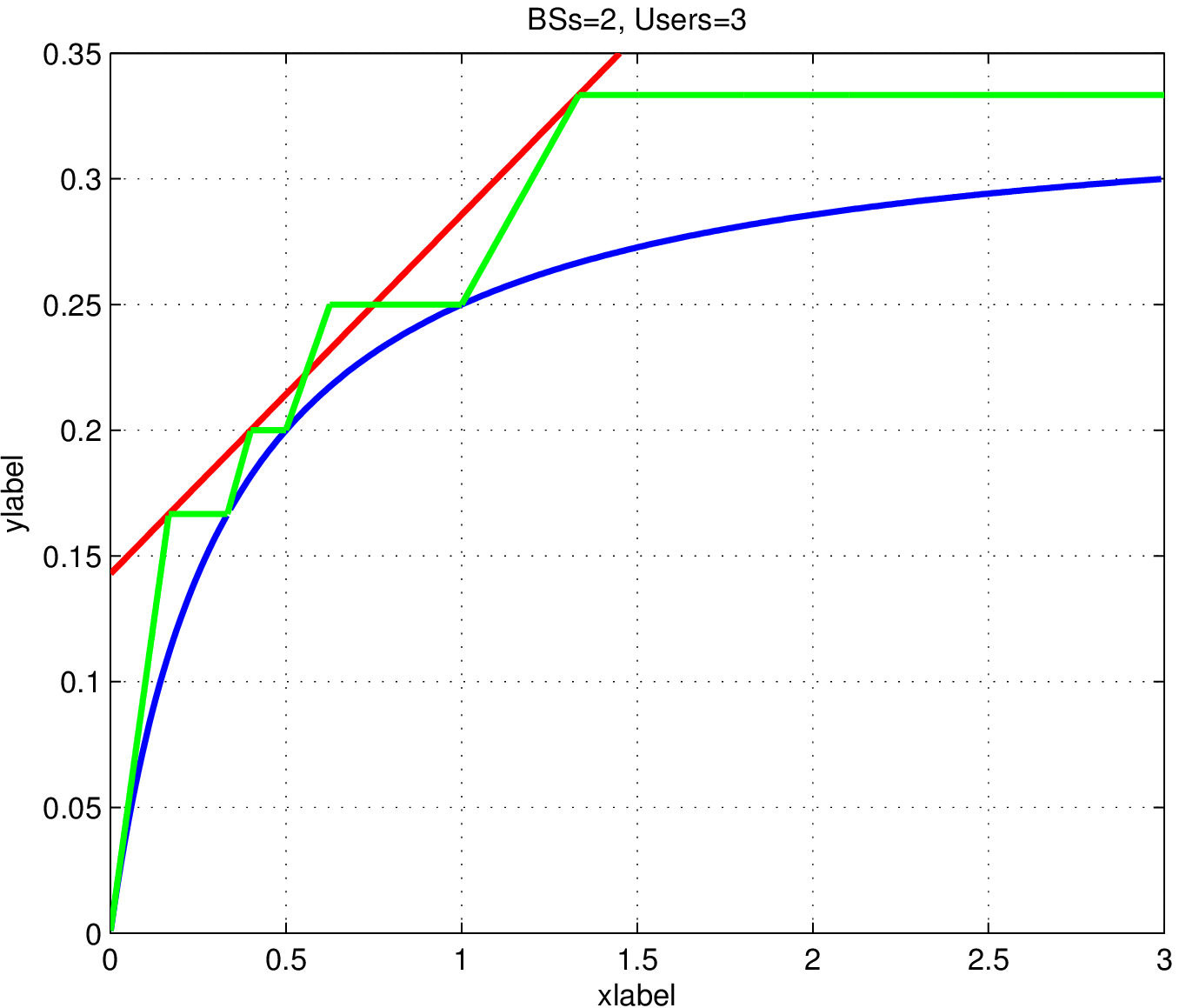}
                \caption{}
                \label{fig:B2U3}
        \end{subfigure}
        ~ 
        \begin{subfigure}[b]{0.3\textwidth}
                \psfrag{xlabel}[tc][tc][0.8]{$M/N$}
		\psfrag{0.05}[tc][tc][0.65]{}
		\psfrag{0.15}[tc][tc][0.65]{}
		\psfrag{0.25}[tc][tc][0.65]{}
		\psfrag{0.35}[tc][tc][0.65]{}
		\psfrag{0.45}[tc][tc][0.65]{}
		\psfrag{0}[tr][tr][0.65]{}
		\psfrag{0.1}[tr][tr][0.8]{0.1}
		\psfrag{0.2}[tr][tr][0.8]{0.2}
		\psfrag{0.3}[tr][tr][0.8]{0.3}
		\psfrag{0.4}[tr][tr][0.8]{0.4}
		\psfrag{0.5}[tr][tr][0.65]{}
		\psfrag{1}[tr][tr][0.8]{1}
		\psfrag{2}[tr][tr][0.8]{2}
		\psfrag{3}[tr][tr][0.8]{3}
		\psfrag{1.5}[tr][tr][0.65]{}
		\psfrag{2.5}[tr][tr][0.65]{}
		\psfrag{ylabel}[Bc][Bl][0.8][180]{ $\begin{matrix} \text{DoF/user/N} \\ \vspace{0.3cm} \end{matrix}$}
		\psfrag{BSs=2, Users=4}[bc][bc][0.9]{2 cells and 4 users/cell}
		\includegraphics[width=2in,height=1.6in]{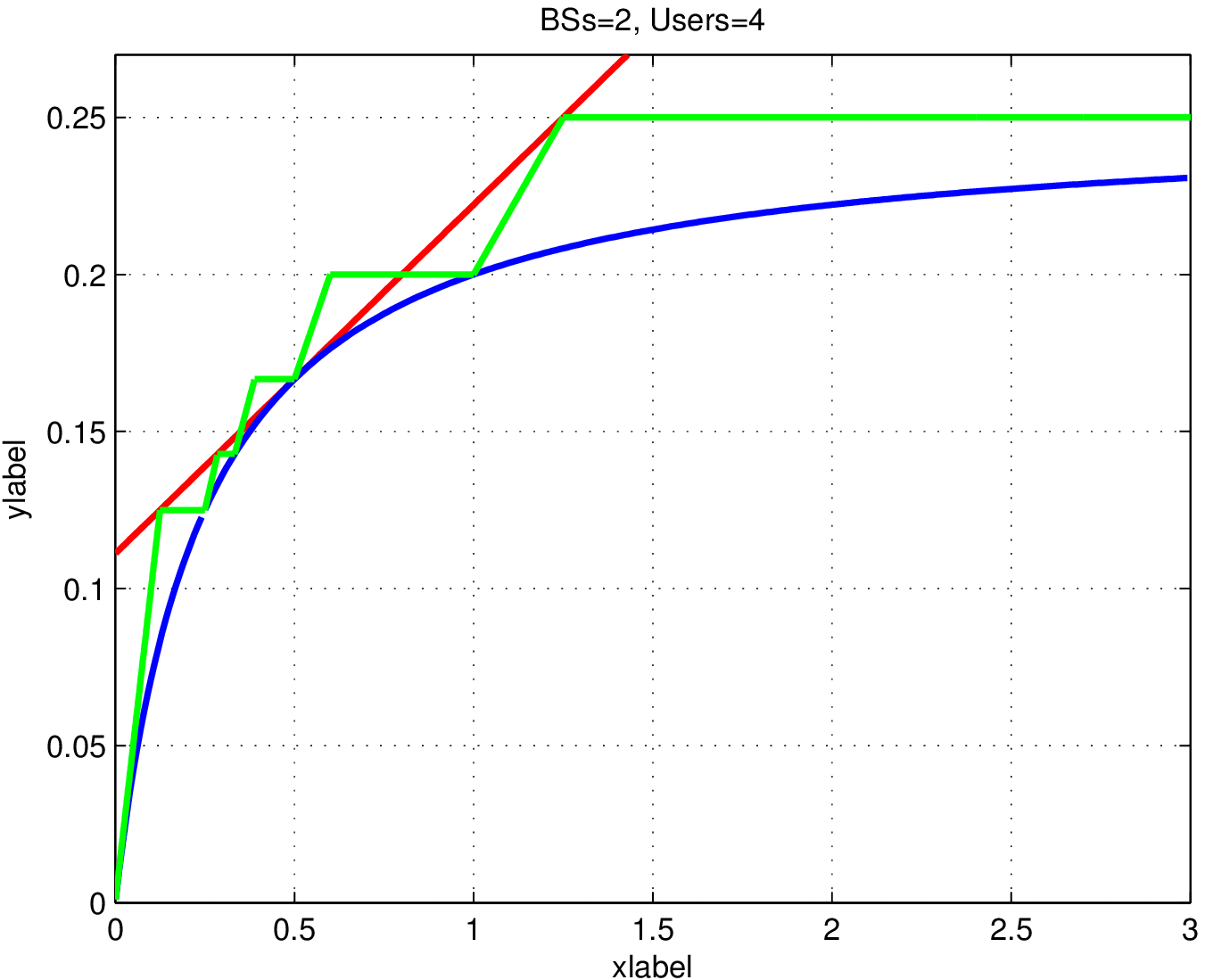}
                \caption{}
                \label{fig:B2U4}
        \end{subfigure}
	~
	\begin{subfigure}[b]{0.3\textwidth}
		\psfrag{xlabel}[tc][tc][0.8]{$M/N$}
		\psfrag{0.02}[tc][tc][0.65]{}
		\psfrag{0.04}[tc][tc][0.65]{}
		\psfrag{0.06}[tc][tc][0.65]{}
		\psfrag{0.08}[tc][tc][0.65]{}
		\psfrag{0.12}[tc][tc][0.65]{}
		\psfrag{0.14}[tc][tc][0.65]{}
		\psfrag{0.16}[tc][tc][0.65]{}
		\psfrag{0.18}[tc][tc][0.65]{}
		\psfrag{0}[tc][tc][0.65]{}
		\psfrag{0.1}[tr][tr][0.8]{0.1}
		\psfrag{0.2}[tr][tr][0.8]{0.2}
		\psfrag{0.3}[tr][tr][0.8]{0.3}
		\psfrag{0.4}[tr][tr][0.8]{0.4}
		\psfrag{0.5}[tr][tr][0.65]{}
		\psfrag{1}[tr][tr][0.8]{1}
		\psfrag{2}[tr][tr][0.8]{2}
		\psfrag{3}[tr][tr][0.8]{3}
		\psfrag{1.5}[tr][tr][0.65]{}
		\psfrag{2.5}[tr][tr][0.65]{}
		\psfrag{ylabel}[][][0.8][180]{ $\begin{matrix} \text{DoF/user/N} \\ \vspace{0.3cm} \end{matrix}$}
		\psfrag{BSs=2, Users=5}[bc][bc][0.9]{2 cells and 5 users/cell}
		\includegraphics[width=2in,height=1.6in]{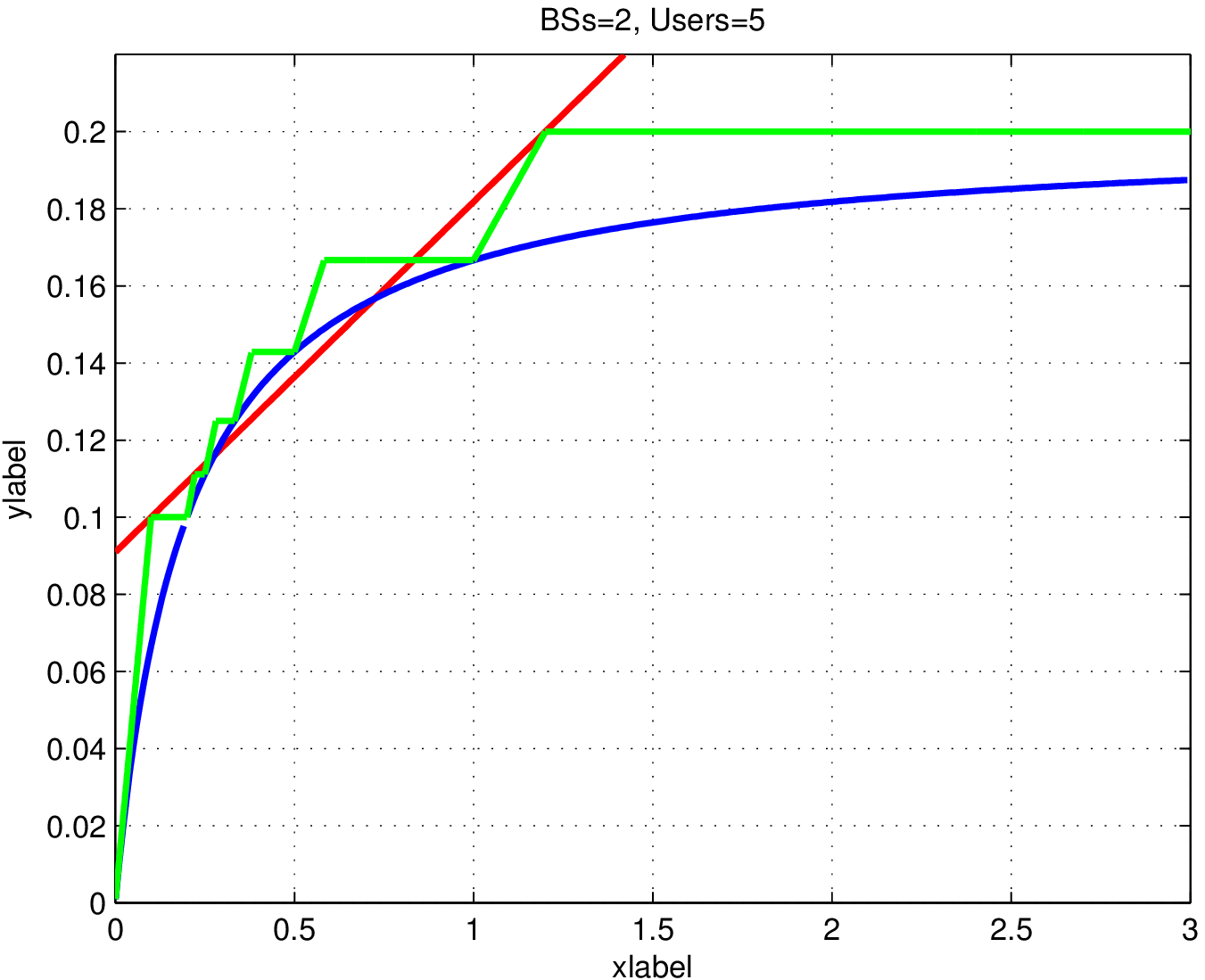}
                \caption{}
                \label{fig:B2U5}
        \end{subfigure}
	~
	\begin{subfigure}[b]{0.3\textwidth}
		\psfrag{xlabel}[tc][tc][0.8]{\color{black}$M/N$}
		\psfrag{0.05}[tr][tr][0.8]{0.05}
		\psfrag{0.01}[tc][tc][0.65]{}
		\psfrag{0.02}[tc][tc][0.65]{}
		\psfrag{0.03}[tc][tc][0.65]{}
		\psfrag{0.04}[tc][tc][0.65]{}
		\psfrag{0.06}[tc][tc][0.65]{}
		\psfrag{0.07}[tc][tc][0.65]{}
		\psfrag{0.08}[tc][tc][0.65]{}
		\psfrag{0.09}[tc][tc][0.65]{}
		\psfrag{0}[tr][tr][0.65]{}
		\psfrag{0.1}[tr][tr][0.8]{0.1}
		\psfrag{0.5}[tr][tr][0.65]{}
		\psfrag{1}[tr][tr][0.8]{1}
		\psfrag{2}[tr][tr][0.8]{2}
		\psfrag{3}[tr][tr][0.8]{3}
		\psfrag{1.5}[tr][tr][0.65]{}
		\psfrag{2.5}[tr][tr][0.65]{}
		\psfrag{ylabel}[Bc][Bl][0.8][180]{ \color{black} $\begin{matrix} \text{DoF/user/N} \\ \vspace{0.3cm} \end{matrix}$}
		\psfrag{BSs=2, Users=10}[bc][bc][0.9]{\color{black}2 cells and 10 users/cell}
		\includegraphics[width=2in,height=1.6in]{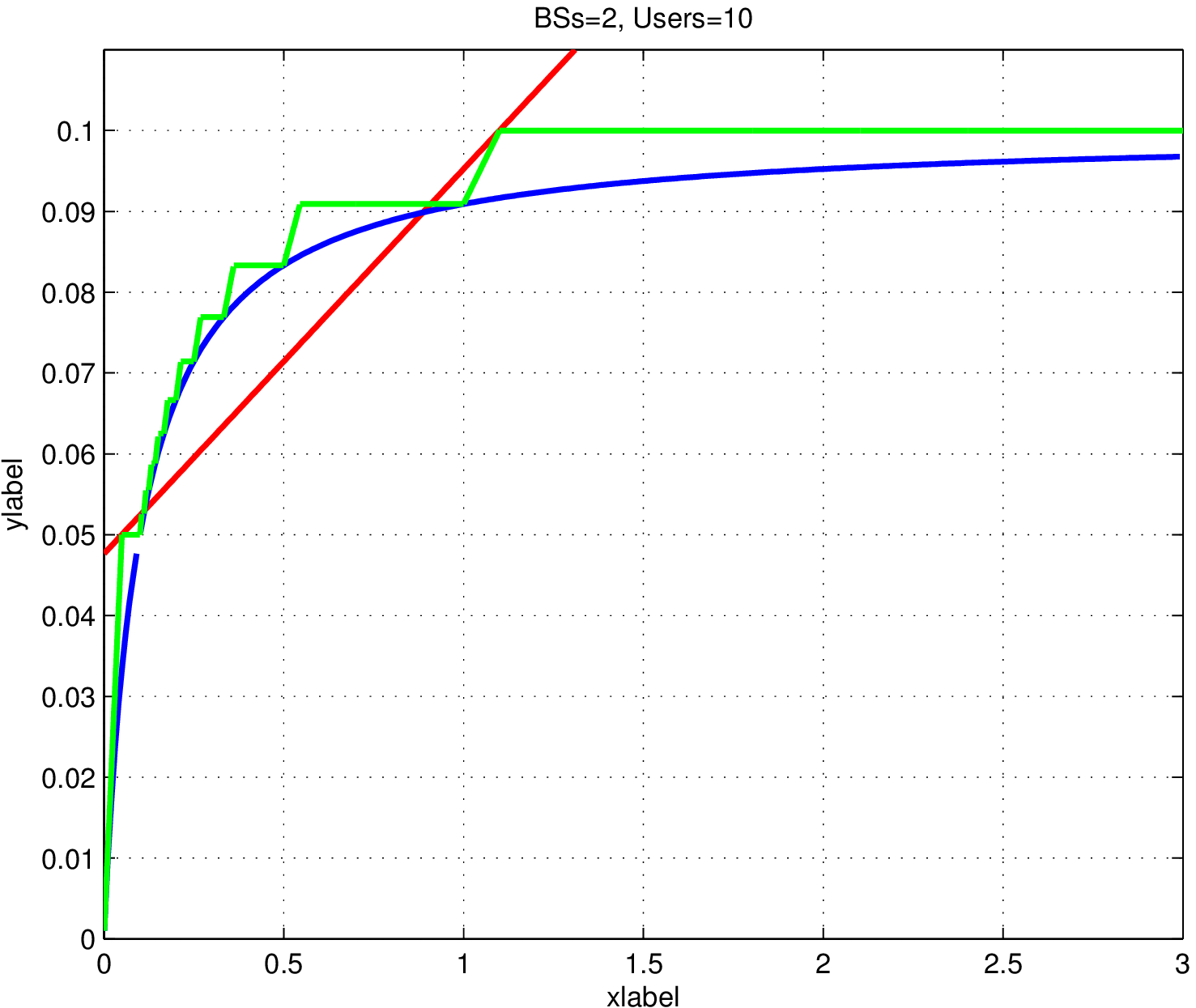}
		\caption{}
		\label{fig:B2U10}

        \end{subfigure}%
	\caption{The proper-improper boundary (red), decomposition inner bound (blue), and the DoF outer bounds (green) for a set of two-cell networks with different number of users per cell. Note the increasing dominance of the decomposition based inner bound as the network size increases.}\label{fig:5plots}
\end{figure*}

\begin{corollary}
Consider a two-dimensional square grid of BSs with $K$ users/cell, $M$ antennas/user, and $N$ antennas/BS, such that each BS interferes only with the four neighboring BSs as shown in Fig. \ref{2dwyner}. When $KM=N$, time sharing between adjacent cells so as to completely avoid interference is a DoF optimal strategy and achieves $N/2K$ DoF/user. 
 \end{corollary}

%% file: keytakeaways.tex
\subsection{Insights on the Optimal DoF of MIMO Cellular Networks}
\label{keytakeaways}

When the achievable DoF using decomposition, the outer bounds on the DoF, and the proper-improper boundary are viewed together, an insightful (albeit incomplete) picture of the optimal DoF of MIMO cellular
networks emerges. Fig.~\ref{fig:5plots} plots the normalized DoF/user (DoF/user/N) achieved by the decomposition based approach as a function of the ratio $M/N \ (\gamma)$ along with the outer
bounds derived in Theorem \ref{1byq_bound} for a set of two-cell networks with different number of users/cell. We also plot the proper-improper boundary $(M+N \lessgtr(GK+1)d)$ that acts as an upper
bound on the DoF that can be achieved using linear beamforming (improper systems are almost surely infeasible). Although Fig.~\ref{fig:5plots} only considers two-cell networks, several important insights on general MIMO cellular networks can be inferred and are listed below.\\

\noindent(a) \textit{Two distinct regimes:} Depending on the network parameters $G$, $K$, $M$ and $N$, there are two distinct regimes where decomposition based schemes outperform linear beamforming and vice versa.\\
\noindent (b) \textit{Optimality of decomposition based schemes for large networks:} For large networks, the decomposition based approach is capable of achieving higher DoF than linear beamforming and the range of $\gamma$ over which the decomposition based approach dominates over linear beamforming increases with network size. The outer bounds on the DoF suggest that when the decomposition based inner bound lies above the proper-improper boundary, the inner bound could well be optimal. Fig. \ref{fig:5plots}(e) is particularly illustrative of this observation.\\
\noindent(c) \textit{Importance of linear beamforming for small networks:} For small networks (e.g. two-cell, two-users/cell; two-cell, three-users/cell), the decomposition based inner bound lies below the proper-improper boundary, suggesting that linear beamforming schemes can outperform decomposition based schemes. In the next section, we study the DoF of the two smallest cellular networks and design
a linear beamforming strategy that achieves the optimal DoF of these two networks. In the subsequent section a general technique to design linear beamformers for any cellular network is presented.\\
\noindent (d) \textit{Inadequacy of existing outer bounds:} The outer bounds listed in Theorem \ref{1byq_bound} are not exhaustive, i.e., in some cases, tighter bounds are necessary to establish the optimal DoF. This observation is drawn from Fig.~\ref{fig:5plots}(b), where it is seen that some part of the outer bound lies above both the proper-improper boundary and the decomposition based inner bound 
suggesting that tighter outer bounds may be possible. In the next section, we indeed derive a tighter outer bound for specific two-cell three-users/cell networks.

Motivated by the above observations, we now turn to linear beamforming schemes for MIMO cellular networks.

%% file: optimal-sDoF.tex
\label{SAP}

    Consider a $(G,K,M,N)$ network with the goal of serving each user with $d$ data streams to each user. Using (\ref{eq_uhv}), when no symbol extensions are allowed, the linear beamformers $\bt V_{ij}$ and $\bt U_{ij}$ need to satisfy the following two conditions for linear interference alignment \cite{yetis}:

     \begin{align}
	\label{IAreq1}
	\bt{U}_{ij}^H \bt{H}_{lm,i} \bt{V}_{lm}&=\bt{0} \ \forall \ (i,j)\neq (l,m)\\
	\label{IAreq2}
	\mathrm{rank}(\bt{U}_{ij}^H \bt{H}_{ij,i} \bt{V}_{ij})&=d \ \forall \ (i,j).
     \end{align}
     For a given system, it is not always possible to satisfy the conditions in (\ref{IAreq1}) and (\ref{IAreq2}) and a preliminary check on feasibility is to make sure that the given system is proper
\cite{yetis,zhuangfeasibility}. As mentioned earlier, a $(G,K,M,N)$ network with $d$ DoF/user is said to be proper if $M+N\geq(GK+1)d$ and improper otherwise
\cite{zhuangfeasibility}. While not all proper systems are feasible, improper systems have been shown to be almost surely infeasible \cite{razaviyayn,tingtingliu}. For proper-feasible systems, solving the system of bilinear equations (\ref{IAreq1}) typically requires the use of iterative algorithms such as those developed in \cite{gomadam,peters2009,dimakis,gokulICASSP}. In certain cases where $\max{M,N}\geq GKd$, it is possible to solve the system of bilinear equations by randomly choosing either the receive beamformers $\{\bt U_{ij}\}$ or the transmit beamformers $\{ \bt V_{ij} \}$ and then solving the resulting linear system of equations.

Assuming the channels to be generic allows us to restate the conditions in (\ref{IAreq1}) and (\ref{IAreq2}) in a manner that is more useful in developing DoF optimal linear beamforming schemes. Since direct channels do not play a role in (\ref{IAreq1}), the condition in (\ref{IAreq2}) is automatically satisfied whenever $\bt{U}_{ij}$ and $\bt{V}_{ij}$ have rank $d$ and whenever the channels are generic \cite{yetis}. As a further consequence of channels being generic, satisfying (\ref{IAreq1}) is equivalent to the condition that the set of uplink transmit beamformers $\{\bt{V}_{ij} \}$ is such that there are at least $d$ interference-free dimensions at each receiver before any linear processing. In essence, generic channels ensure that at each BS, the intersection between useful signal subspace (span([$\bt H_{i1,i}\bt V_{i1}, \bt H_{i2,i}\bt V_{i2},\hdots,\bt H_{iK,i}\bt V_{iK}$]) and interference subspace (span($\bt R_i$)) is almost surely zero dimensional, provided that the rank$(\bt{R}_i) \leq (N -Kd)  \ \forall i$. Thus the requirements for interference alignment can be
alternately stated as
    \begin{align}
	\label{altIAreq1}
	\mathrm{rank}(\bt{R}_i) \leq N-Kd \ \forall \ i ,\\
	\label{altIAreq2}
	\mathrm{rank}(\bt{V}_{jl})=d \ \forall \ j,l .
    \end{align}
The rank constraint in (\ref{altIAreq1}) essentially requires the $(G-1)Kd$ column vectors of $\bt R_i$ to satisfy $L=GKd-N$ distinct linear vector equations. Given a set of transmit precoders $\{\bt{V}_{jl} \}$ that satisfy the above conditions, designing the receive filters is then straightforward.

This alternate perspective on interference alignment lends itself to counting arguments that account for the number of dimensions at each BS occupied by signal or interference. These counting arguments in turn lead to the development of DoF-optimal linear beamforming strategies such as the subspace alignment chains for the 3-user interference channel \cite{chenweiwang}.

In this section we take a structured approach to constructing the $L$ distinct linear vector equations that need to satisfy (\ref{altIAreq1}) and (\ref{altIAreq2}). Such an approach is DoF-optimal for small networks such as the two-cell two-user/cell and the two-cell, three-user/cell networks.

\subsection{Main Results}

We consider two of the simplest cellular networks, namely the two-cell two-user/cell and the two-cell, three-user/cell networks, and establish a linear beamforming strategy that achieves the optimal symmetric DoF. In particular, we establish the spatially-normalized DoF of these two networks for all values of the ratio $\gamma= M/N$. The spatially-normalized DoF of a network is defined as follows \cite{chenweiwang}.
    \begin{definition}
    Denoting the DoF/user of a $(G,K,M,N)$ cellular network as $\text{DoF}(M,N)$, the spatially-normalized DoF/user is defined as
    \begin{equation}
    \text{sDoF}(M,N)=\max_{q\in \mathcal{Z}^+} \frac{\text{DoF}(qM,qN)}{q}.
    \end{equation}
    \end{definition}
Analogous to frequency and time domain symbol extensions, the definition above allows us to permit extensions in space, i.e., adding antennas at the transmitters and receivers while maintaining the ratio $M/N$ to be a constant. Unlike time or frequency extensions where the resulting channels are block diagonal, spatial extensions assume generic channels with no additional structure. The lack of any structure in the channel obtained through space extensions makes it significantly easier to analyze the network.

\begin{figure*}[ht]
      \begin{center}
	\hspace{-0.5cm}
	\psfrag{ya}[Br][Bc][1]{$0$}
	\psfrag{yb}[Br][Br][0.8]{$1/4$}
	\psfrag{yc}[Br][Br][0.8]{$1/3$}
	\psfrag{yd}[Br][Br][0.8]{$1/2$}
	\psfrag{x1}[tc][tc][1]{$0$}
	\psfrag{x2}[tc][tc][1]{$\frac{1}{4}$}
	\psfrag{x3}[tc][tc][1]{$\frac{1}{2}$}
	\psfrag{x4}[tc][tc][1]{$\frac{2}{3}$}
	\psfrag{x5}[tc][cc][1]{$1$}
	\psfrag{x6}[tc][tc][1]{$\frac{3}{2}$}
	\psfrag{2/3}[tc][tc][1]{$\frac{2}{3}$}
	\psfrag{3/4}[tc][tc][1]{$\frac{3}{4}$}
	\psfrag{4/3}[tc][tc][1]{$\frac{4}{3}$}
	\psfrag{1M}[cc][cc][1]{$\frac{N}{2}$}
	\psfrag{2M}[cc][cc][1]{$\frac{M}{3}$}
	\psfrag{3M}[cc][cc][1]{$\frac{N}{3}$}
	\psfrag{4M}[cc][cc][1]{$\frac{M}{2}$}
	\psfrag{5M}[cc][cc][1]{$\frac{N}{4}$}
	\psfrag{6M}[cc][cc][0.8]{$M$}
	\psfrag{pl}{$\frac{\gamma+1}{5}$}
	\psfrag{t1}[Bl][Bl][1]{$\frac{\gamma+1}{5}$}
	\psfrag{t3}[Bl][Bl][1]{proper-improper boundary}
	\psfrag{t2}[Bl][Bl][1]{$\frac{\gamma}{2\gamma+1}$}
	\psfrag{title}[Bl][Bl][1]{}
	\psfrag{xlabel}[tc][tc][1]{$\gamma \ (M/N)$}
	\psfrag{ylabel}[Bc][Bc][0.9][180]{Normalized sDoF/user}
	\psfrag{pip}[Bl][Bl][0.9]{}
	\psfrag{pip2}[Bl][Bl][0.9]{}
	\psfrag{proper-improper boundary}[Bl][Bl][0.8]{Proper-improper boundary }
	\psfrag{optimal sDoF}[Bl][Bl][0.75]{Optimal sDoF}
	\psfrag{decomposition based inner bound}[Bl][Bl][0.78]{Decomposition based inner bound}
	\psfrag{BSs=2, Users=2}[Bl][Bl][0.9]{}
	\includegraphics[width=5.4in]{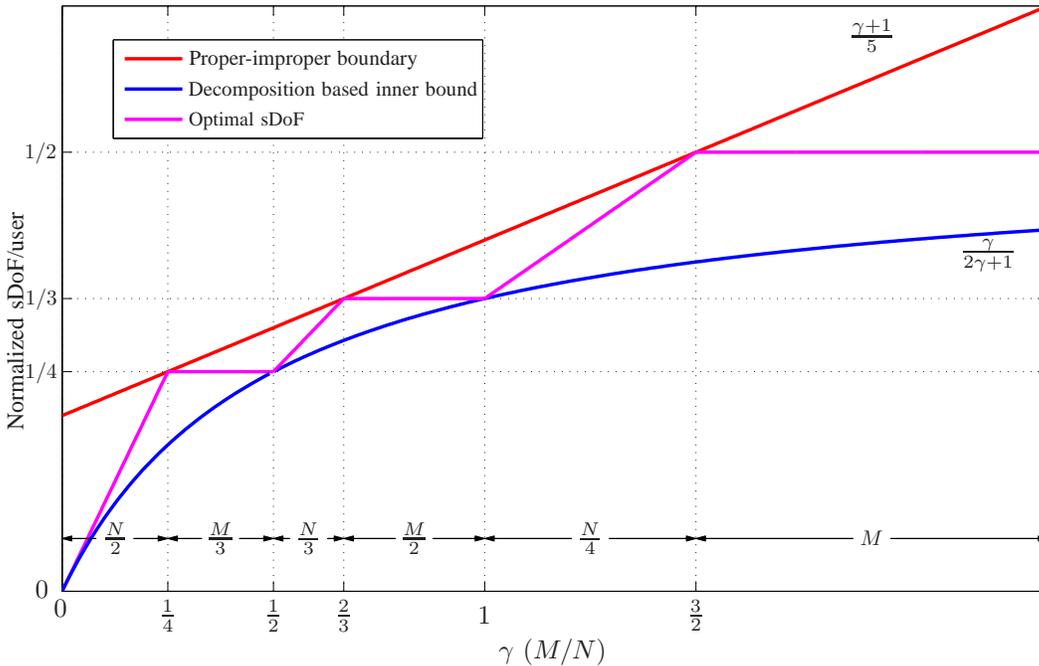}
	\caption{The sDoF/user (normalized by $N$) of a 2-cell, 3-user/cell MIMO cellular network as a function of $\gamma$. }
	\label{fig_22}
      \end{center}
    \end{figure*}

    \begin{figure*}[ht]
      \begin{center}
	\hspace{-0.5cm}
\psfrag{0y}[Br][Bc][1]{$0$}
	\psfrag{1y/5}[Br][Br][0.8]{$1/5$}
	\psfrag{2y/9}[Br][Br][0.8]{$2/9$}
	\psfrag{1y/4}[Br][Br][0.8]{$1/4$}
	\psfrag{1y/6}[Br][Br][0.8]{$1/6$}
	\psfrag{1y/3}[Br][Br][0.8]{$1/3$}
\psfrag{0x}[tc][cc][1]{$0$}
	\psfrag{1}[tc][cc][1]{$1$}
	\psfrag{1/3}[tc][tc][1]{$\frac{1}{3}$}
	\psfrag{1/2}[tc][tc][1]{$\frac{1}{2}$}
	\psfrag{1/6}[tc][tc][1]{$\frac{1}{6}$}
	\psfrag{5/9}[tc][tc][1]{$\frac{5}{9}$}
	\psfrag{2/5}[tc][tc][1]{$\frac{2}{5}$}
	\psfrag{2/3}[tc][tc][1]{$\frac{2}{3}$}
	\psfrag{3/4}[tc][tc][1]{$\frac{3}{4}$}
	\psfrag{4/3}[tc][tc][1]{$\frac{4}{3}$}
\psfrag{2}[Br][Br][0.8]{}
\psfrag{1M}[cc][cc][0.8]{$M$}
\psfrag{2M}[cc][cc][0.8]{$\frac{N}{6}$}
\psfrag{3M}[cc][cc][0.8]{$\frac{M}{2}$}
\psfrag{4M}[cc][cc][0.8]{$\frac{N}{5}$}
\psfrag{5M}[cc][cc][0.8]{$\frac{2M}{5}$}
\psfrag{6M}[cc][cc][0.8]{$\frac{2N}{9}$}
\psfrag{7M}[cc][cc][0.8]{$\frac{M}{3}$}
\psfrag{8M}[cc][cc][0.8]{$\frac{N}{4}$}
\psfrag{9M}[cc][cc][0.8]{$\frac{M}{4}$}
\psfrag{10M}[cc][cc][0.8]{$\frac{N}{3}$}

	\psfrag{pl}{$\frac{\gamma+1}{5}$}
	\psfrag{t1}[Bl][Bl][1]{$\frac{\gamma+1}{7}$}
	\psfrag{t3}[Bl][Bl][1]{proper-improper boundary}
	\psfrag{t2}[Bl][Bl][1]{$\frac{\gamma}{3\gamma+1}$}
	\psfrag{title}[Bl][Bl][1]{}
	\psfrag{xlabel}[][][1]{\raisebox{0cm}{$\gamma \ (M/N)$}}
	\psfrag{ylabel}[Bc][Bc][0.85][180]{Normalized sDoF/user}
	\psfrag{pip}[Bl][Bl][0.9]{}
	\psfrag{pip2}[Bl][Bl][0.9]{}
	\psfrag{proper-improper boundary}[Bl][Bl][0.9]{Proper-improper boundary }
	\psfrag{optimal sDoF}[Bl][Bl][0.9]{Optimal sDoF}
	\psfrag{decomposition based inner bound}[Bl][Bl][0.9]{Decomposition based inner bound}
	\psfrag{abcdefghijkl}[Bl][Bl][0.9]{}
	\includegraphics[width=7in]{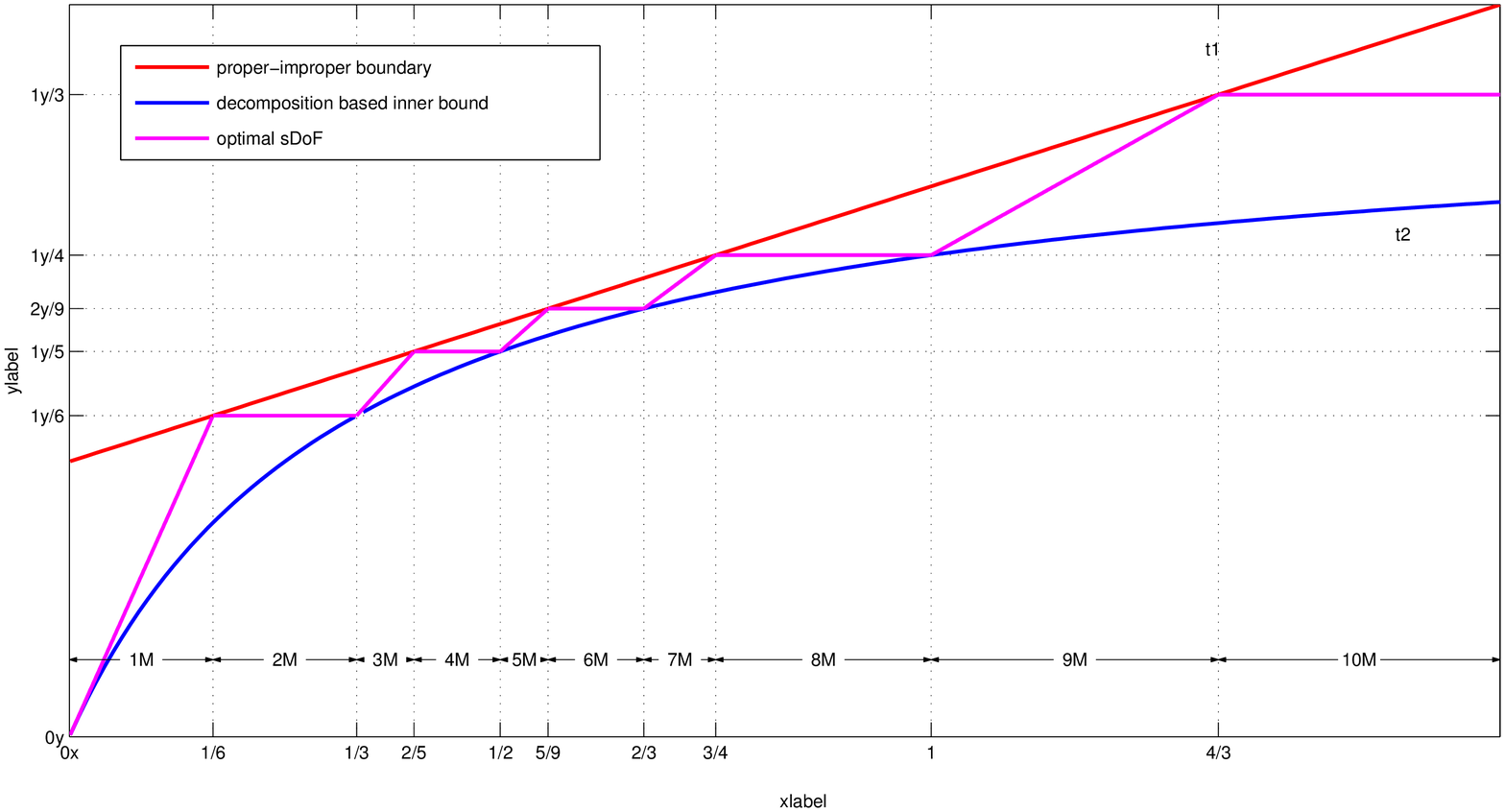}
	\caption{The sDoF/user (normalized by $N$) of a 2-cell, 3-user/cell MIMO cellular network as a function of $\gamma$. }
	\label{fig_23}
      \end{center}
    \end{figure*}
We now present the main results concerning the sDoF of the two cellular networks under consideration. 

Let the function $f_{(\omega,K)} (\cdot )$ be defined as 
\begin{align}
 f_{(\omega,K)}(M,N)=\max \left(\frac{N\omega}{K\omega+1},\frac{M}{K\omega+1} \right),
\end{align}
where $\omega \geq 0$ and $K\in \mathbb{Z}^+$. Further, define the function $D_{(2,2)}(\cdot)$ to be
\begin{align}
D_{(2,2)}(M,N)=&\min \big (N,KM,f_{(\frac{1}{2},2)}(M,N),f_{(1,2)}(M,N) \big ),
\end{align}
 and the function $D_{(2,3)}(\cdot)$ to be
\begin{align}
D_{(2,3)}(M,N)=&\min \big (N,KM,f_{(\frac{1}{3},3)}(M,N),f_{(\frac{1}{2},3)}(M,N), \nonumber \\
&\hspace{1.5cm}f_{(\frac{2}{3},3)}(M,N),f_{(1,3)}(M,N) \big ) .
\end{align}

    The following theorem characterizes an outer bound on the DoF/user of the two-cell two-user/cell network and the two-cell three-user/cell network.
    \begin{theorem}
     The DoF/user of a two-cell, K-user/cell MIMO cellular network with $K \in \{2,3\}$, having $M$ antennas per user and $N$ antennas per BS is bounded above by $D_{(2,K)}(M,N)$, i.e., 
     \begin{equation}
      \text{DoF/user}\leq D_{(2,K)}(M,N).
     \end{equation}
    \label{outerbound_theorem}
    \vspace{-3mm}
    \end{theorem}
    Note that since this outer bound is linear in either $M$ or $N$, this bound is invariant to spatial normalization and hence is also a bound on sDoF and not just DoF. The outer bounds for the two-cell, two-user/cell case follows directly from either the bounds established in Section \ref{outerbounds} (for $1/4 \leq \gamma \leq 3/2$) or through DoF bounds on the multiple-access/broadcast channel (MAC/BC) obtained by letting the two cells cooperate (for $ \gamma \leq 1/4)$ and $\gamma \geq 3/2$). In the case of the two-cell, three-user/cell network, the bounds when $\gamma \leq 1/6$ or $\gamma \geq 4/3$ follow from DoF bounds on the MAC/BC obtained by letting the two cells cooperate, while the bounds when $1/6 \leq \gamma \leq 5/9$ and $3/4 \leq \gamma \leq 4/3$ follow from the bounds established in Section \ref{outerbounds}. When $5/9 \leq \gamma \leq 3/4$, we derive a new set of genie-aided outer bounds on the DoF. Our approach to deriving these new bounds is similar to the 
approach taken in \cite{chenweiwang} and the exact details of this derivation are presented in Appendix \ref{genieouterboundB}.  

    The outer bound presented in the previous theorem turns out to be tight. The main theorem of this section is a characterization of the sDoF/user of the two-cell, two-or-three-user/cell MIMO cellular network. The proof of achievability is deferred to the next section.

    \renewcommand{\thesection}{\arabic{section}}
    \begin{theorem}
     The spatially-normalized DoF of a 2-cell, $K$-user/cell cellular network with $K \in \{2,3\}$, having $M$ antennas per user and $N$ antennas per BS is given by
    \begin{align}
    \text{sDoF/user}=D_{2,K}(M,N).
    \end{align}
    \label{sDoF_theorem}
     \vspace{-3mm}
    \end{theorem}
    \renewcommand{\thesection}{\Roman{section}}
    This result states that when spatial-extensions are allowed, the outer bound presented in Theorem \ref{outerbound_theorem} is tight. The achievability part of the result in Theorem \ref{sDoF_theorem} is based on a linear beamforming strategy developed using the notion of packing ratios. We elaborate further on this scheme in the next subsection.
    
    Figs. \ref{fig_22} and \ref{fig_23} capture the main results presented in the above theorems and plot sDoF/user normalized by $N$ as a function of $\gamma$. It can be seen in both the figures that, just as in the 3-user interference channel \cite{chenweiwang}, there is an alternating behavior in the sDoF with either $M$ or $N$ being the bottleneck for a given $\gamma$.

    The figures also plot the boundary separating proper systems from improper systems. It is seen from the two figures that not all proper systems are feasible. For example, for the two-cell three-users/cell case, networks with $\gamma$ $\in$ $\{1/6,2/5,5/9,3/4,4/3\}$ are the only ones on the proper-improper boundary that are feasible.
    
    For the two-cell two-users/cell network, we can see from Fig. \ref{fig_22} that when $\gamma$ $\in$ $\{1/4,2/3,3/2\}$, neither $M$ nor $N$ has any redundant dimensions, and decreasing either of them affects the sDoF. On the other hand, when $M/N$ $\in$ $\{1/2,1\}$, both $M$ and $N$ have redundant dimensions, and some dimensions from either $M$ or $N$ can be sacrificed without losing any sDoF. For all other cases, only one of $M$ or $N$ is a bottleneck. Similar observations can also be made for the 2-cell 3-users/cell network from Fig. \ref{fig_23}.

    Figs. \ref{fig_22} and \ref{fig_23} also plot the achievable DoF using the decomposition based approach. Interestingly, the only cases where the decomposition based inner bound achieves the optimal sDoF is when both $M$ and $N$ have redundant dimensions i.e., $\gamma$ $\in$ $\{1/2,1\}$ in the case of the two-cell, two-user/cell network and when $\gamma$ $\in$ $\{1/3,1/2,2/3,1\}$ in the case of the two-cell, three-user/cell network.

%% file: packingratios.tex
\subsubsection{Packing Ratios}
\label{packingratios}

We now present the linear transmit beamforming strategy that achieves the optimal sDoF of the two networks under consideration. We consider achievability only in the uplink as duality of interference alignment through linear beamforming ensures achievability in the downlink as well. We start by introducing a new notion called the \textit{packing ratio} to describe a collection of transmit beamforming vectors. 

\begin{definition}
Consider the uplink of a two-cell network and let $\mathcal{S}$ be a collection of transmit beamformers used by users belonging to the same cell. If the number of dimensions occupied by the signals transmitted using this set of beamformers at the interfering BS is denoted by $d$, then the packing ratio $\eta$ of this set of beamformers is given by $|\mathcal{S}|\!:\!d$.
\end{definition}

As an example, consider a two-cell, three-users/cell cellular network with 2 antennas at each user and 3 antennas at each BS. Suppose we design two beamformers $\bt v$ and $\bt w$ for two different users in the same cell so that $\bt H_{11,2}\bt v=\bt H_{12,2}\bt w$, then the set of vectors $\mathcal{S}=\{\bt v, \bt w \}$ is said to have a packing ratio of $2\!:\!1$. As another example, for the same network, consider the case when $M>N$. Since users can now zero-force all antennas at the interfering BS, we can have a set $\mathcal{S}$ of beamformers with packing ratio $|\mathcal{S}|:0$.

When designing beamformers for the two-cell network, it is clear that choosing sets of beamformers having a high packing ratio is desirable as this reduces the number of dimensions occupied by interference at the interfering BS. The existence of beamformers satisfying a certain packing ratio is closely related to the ratio $\gamma$ ($M/N$). For example, it is easily seen that when $\gamma<\frac{2}{3}$, it is not possible to construct beamformers having a packing ratio of $3\!:\!1$. Further even when beamformers satisfying a certain packing ratio exist, there may not be sufficient sets of them to completely use all the available dimensions at a BS. In such a scenario, we need to consider designing beamformers with the next best packing ratio.

Using the notion of packing ratios, we now describe the achievability of the optimal sDoF of the two-cell three-users/cell cellular network. We first define the set $\mathcal{P}_{23}=\{1\!:\!0,\ 3\!:\!1,\ 2\!:\!1,\ 3\!:\!2,\ 1\!:\!1 \}$ to be the set of fundamental packing ratios for the two-cell, three-users/cell cellular network. For any given $\gamma$, our strategy is to first construct the sets of beamformers that have the highest possible packing ratio from the set $\mathcal{P}_{23}$. If such beamformers do not completely utilize all the available dimensions at the two BSs, we further construct beamformers having the next best packing ratio in $\mathcal{P}_{23}$ until all the dimensions at the two BSs are either occupied by signal or interference. Our proposed strategy is essentially a greedy strategy to minimize the dimensions occupied by interference. Greedy strategies for aligning interference, including the notion of subspace alignment chains developed in \cite{chenweiwang} where an alignment chain is terminated until no more interference can be aligned, seem capable of achieving the optimal sDoF. The strategy we develop is illustrated in the following example.

Consider the case  when $2/3< \gamma< 3/4$. Since $M<N$, no transmit zero-forcing is possible. Further, each user can access only $M$ of the $N$ dimensions at the interfering BS. Since we assume all channels to be generic, and $2M >N$, the subspaces accessible to any two users overlap in $2M-N$ dimensions. This $2M-N$ dimensional space overlaps with the $M$ dimensions accessible to the third user in $3M-2N$ dimensions. Note that such a space exists as we have assumed $2/3< \gamma $. Thus, we can construct $3M-2N$ sets of three beamformers (one for each user) that occupy just one dimension at the interfering BS and thus have a packing ratio of $3\!:\!1$. Assuming that the same strategy is adopted for users in both cells, at any BS, signal vectors occupy a total of $3(3M-2N)$ dimensions while interference occupies $3M-2N$ dimensions. Thus a total of $4(3M-2N)$ dimensions are occupied by signal and interference. Since $4(3M-2N)<N$ whenever $4M <3N$, we see that such vectors do not completely utilize all the $N$ dimensions at a BS. 

In order to utilize the remaining $9N-12M$ dimensions, we additionally construct beamformers with the next highest packing ratio ($2\!:\!1$). Let $M'=M-(3M-2N)=2N-2M$ denote the unused dimensions at each user. At the interfering BS, each pair of users has $2M'-N$ dimensions that can be accessed by both users. Note that since $2M'-N=2(2N-2M)-N=3N-4M > 0$, such an overlap exists almost surely. For a fixed pair of users in each cell, we choose $(3N-4M)$ sets of two beamformers (one for each user in the pair) whose interference aligns onto a single dimension, so that each set has a packing ratio of $2\!:\!1$. After choosing beamformers in this manner, we see that signal and interference span all $N$ dimensions at each of the two BSs. Through this process, each BS receives $3(3M-2N)+2(3N-4M)$ signaling vectors while interfering signals occupy $(3M-2N)+(3N-4M)$ dimensions. We have thus shown that $3(3M-2N)+2(3N-4M)=M$ DoF/cell are achievable. To ensure that $M/3$ DoF/user are achieved, we can either cycle through different pairs of users when designing the second set of beamformers, or we can simply pick $(3N-4M)/3$ sets of beamformers for every possible pair of users in a cell. If $(3N-4M)/3$ is not an integer, we simply scale $N$ and $M$ by a factor of 3 to make it an integer. We can afford the flexibility to scale $M$ and $N$ because we are only characterizing the sDoF of the network.

\begin{table*}[t]
\renewcommand{\arraystretch}{1.5}
\caption{ The sets of beamformers and their corresponding packing ratios used to prove achievability of the optimal sDoF of the two-cell two-user/cell network for different values of $\gamma$. }
\begin{center}
\normalsize
\begin{tabular}{|c|c|c|c|c|c|c}
  \cline{1-6}
  \multirow{2}{*}{$\gamma$ $(M/N)$} & \multicolumn{4}{|c|}{Set of beamformers}  & \multirow{2}{2.7cm}{\centering  DoF/cell  (No. of signal-vectors per cell)\scriptsize} & \\
  \cline{2-5}
  & Packing ratio & No. of sets & Packing ratio & No. of sets  &  \\
  \cline{1-6}
  $0 < \gamma < \frac{1}{4}$ & $1\!:\!1$ & $2M$ & -- & -- & $2M$ &  \\
  \cline{1-6}
  $\frac{1}{4}\leq \gamma \leq \frac{1}{2}$ & $1\!:\!1$ & $\frac{N}{2}$ & -- & -- & $\frac{N}{2}$ & \\
  \cline{1-6}
  $\frac{1}{2}< \gamma < \frac{2}{3}$ & $2\!:\!1$ & $2M-N$  & $1\!:\!1$ & $\frac{4N-6M}{2}$ & $M$ & \\
  \cline{1-6}
  $\frac{2}{3} \leq \gamma \leq 1$ & $2\!:\!1$ & $2M-N$ & -- & -- & $\frac{2N}{3}$ & \\
  \cline{1-6}
  $1< \gamma < \frac{3}{2}$ & $1\!:\!0$ & $2(M-N)$ & $2\!:\!1$ & $\frac{3N-2M}{3}$ & $\frac{2M}{3}$ &  \\
  \cline{1-6}
  $\frac{3}{2}\leq \gamma$ & $1\!:\!0$ & $N$ & -- & -- & $N$ &\\
  \cline{1-6}
  \end{tabular}
\end{center}
\label{table22}
\end{table*}

\begin{table*}[t]
\renewcommand{\arraystretch}{1.5}
\caption{ The sets of beamformers and their corresponding packing ratios used to prove achievability of the optimal sDoF of the two-cell three-user/cell network for different values of $\gamma$. }
\begin{center}
\normalsize
\begin{tabular}{|c|c|c|c|c|c|c}
  \cline{1-6}
  \multirow{2}{*}{$\gamma$} & \multicolumn{4}{|c|}{Set of beamformers}  & \multirow{2}{2.5cm}{\centering  DoF/cell  (No. of signal-vectors per cell)\scriptsize} & \\
  \cline{2-5}
  & Packing ratio & No. of sets & Packing ratio & No. of sets  &  \\
  \cline{1-6}
  $0 < \gamma < \frac{1}{6}$ & $1\!:\!1$ & $3M$ & -- & -- & $3M$ &  \\
  \cline{1-6}
  $\frac{1}{6}\leq \gamma \leq \frac{1}{3}$ & $1\!:\!1$ & $\frac{N}{2}$ & -- & -- & $\frac{N}{2}$ & \\
  \cline{1-6}
  $\frac{1}{3}< \gamma < \frac{2}{5}$ & $3\!:\!2$ & $3M-N$  & $1\!:\!1$ & $\frac{6N-15M}{2}$ & $\frac{3M}{2}$ & \\
  \cline{1-6}
  $\frac{2}{5}\leq \gamma \leq \frac{1}{2}$ & $3\!:\!2$ & $\frac{N}{5}$ & -- & -- & $\frac{3N}{5}$ & \\
  \cline{1-6}
  $\frac{1}{2}< \gamma < \frac{5}{9}$ & $2\!:\!1$ & $3(2M-N)$ & $3\!:\!2$ & $\frac{10N-18M}{5}$ & $\frac{6M}{5}$ &  \\
  \cline{1-6}
  $\frac{5}{9}\leq \gamma \leq \frac{2}{3}$ & $2\!:\!1$ & $\frac{N}{3}$  & -- & -- & $\frac{2N}{3}$ & \\
  \cline{1-6}
  $\frac{2}{3}< \gamma < \frac{3}{4}$ & $3\!:\!1$ & $3M-2N$ & $2\!:\!1$ & $3N-4M$ & $M$ &\\
  \cline{1-6}
  $\frac{3}{4}\leq \gamma \leq 1$ & $3\!:\!1$ & $\frac{N}{4}$ & -- & -- & $\frac{3N}{4}$  & \\
  \cline{1-6}
  $1 < \gamma < \frac{4}{3}$ & $1\!:\!0$ & $3(M-N)$ & $3\!:\!1$ & $N-\frac{3M}{4}$ & $\frac{3M}{4}$ & \\
  \cline{1-6}
  $\frac{4}{3}\leq \gamma$ & $1\!:\!0$ & $N$ & -- & -- & $N$ &\\
  \cline{1-6}
  \end{tabular}
\end{center}
\label{table23}
\end{table*}
As another example, consider the two-cell, three-users/cell network with $3/4 \leq \gamma \leq 1$. When $3/4 \leq \gamma \leq 1$,  all three users of a cell can access a $3M-2N$ dimensional space at the interfering BS, thus $3M-2N$ sets of three beamformers having a packing ratio of $3\!:\!1$ are possible. Note that $3\!:\!1$ is still the highest possible packing ratio. If users in both cells were to use such beamformers, signal and interference from such beamformers can occupy at most $4(3M-2N)>N$ dimensions at any BS. Thus, when $3/4 \leq \gamma < 1$, we have sufficient sets of beamformers with packing ratio $3\!:\!1$ to use all available dimensions at the BSs. Choosing $N/4$ such sets provides us with $3N/4$ DoF/cell. 




Such an approach to designing the linear beamformers provides insight on why the optimal sDoF alternates between $M$ and $N$. When $\gamma$ is such that there are sufficient sets of beamformers having the highest possible packing ratio, it is the number of dimensions at the BSs that proves to be a bottleneck and the DoF bound becomes dependent on $N$. On the other hand, when there are not enough sets of beamformers having the highest possible packing ratio, we are forced to design beamformers with a lower packing ratio so as to use all available dimensions at the two BSs. Since for a fixed $N$, the number of sets of beamformers having the highest packing ratio is a function of $M$, the bottleneck now shifts to $M$. We thus see that for a large but fixed $N$, as we gradually increase $M$, we cycle through two stages---the first stage where beamformers with a higher packing ratio become feasible but are limited to a small number, then gradually, the second stage where there are sufficiently many such beamformers. As $M$ is increased even further, we go back to the scenario where the next higher packing ratio becomes feasible however with only limited set of beamformers, and so on.

The design strategy described for the case $2/3< \gamma \leq 1$ is also applicable to other intervals of $\gamma$, as well as the two-cell two-users/cell network. For the two-cell three-user/cell network, when $1/3 < \gamma \leq 1/2$, we design as many sets of beamformers having packing ratio $3:2$ as possible, then use beamformers having a packing ratio of $1:1$ (random beamforming) to fill any unused dimensions at the two BSs. When $1/2 < \gamma  \leq 2/3$ we first design as many sets of beamformers having  packing ratio $2:1$ as possible and then use beamformers having a packing ratio of $3:2$. When $\gamma \leq 1/3$, it is easy to see that interference alignment is not feasible and that a random beamforming strategy suffices. Finally, when $\gamma \geq 1$, we first design beamformers that zero-force the interfering BS (packing ratio $1:0$), then use beamformers having a packing ratio of $3:1$ to fill any remaining dimensions at each BS.

For the two-cell two-user/cell network we define the set $\mathcal{P}_{22}=\{1\!:\!0,\ 2\!:\!1,\ 2\!:\!1,\ 1\!:\!1 \}$ to be the set of fundamental packing ratios. When $\gamma > 1$, we first design beamformers that zero-force the interfering BS (packing ratio $1:0$), then if necessary, use beamformers having a packing ratio of $2:1$ to fill any remaining dimensions at each BS. When $1/2 < \gamma \leq 1$, the highest possible packing ratio is $2:1$, hence we first design beamformers having packing ratio $2:1$ to occupy as many dimensions as possible at the two BSs, then if there are unused dimensions at the two BSs, we use random beamformers (packing ratio $1:1$) to occupy the remaining dimensions. When $\gamma \leq 1/2$, interference alignment is not feasible and simple random beamforming achieves the optimal DoF.

In Tables \ref{table22} and \ref{table23}, we summarize the strategies used for different intervals of $\gamma$, and list the number of sets of beamformers of a certain packing ratio required to achieve the optimal DoF along with the DoF achieved per cell. Note that fractional number of sets can always be made into integers as we allow for spatial extensions. We discuss finer details on constructing beamformers using packing ratios in Appendix~\ref{finerdetails}.

\subsection{Extending packing ratios to larger networks}
    It is possible to extend the notion of packing ratios to certain larger networks. For example, the following theorem establishes the optimal sDoF of two-cell networks with more than three users per cell for certain values of $\gamma$.
    \begin{theorem}
     The optimal sDoF/user of a three-cell, $K$-user/cell MIMO cellular network with $M$ antennas per user and $N$ antennas per BS when $\gamma=\tfrac{M}{N} \in (0, \tfrac{1}{K-1}]$ is given by
     \begin{equation}
      \text{DoF/user}\leq \min \big(M, \max \big (\tfrac{N}{2K},\tfrac{M}{2} \big ), \tfrac{N}{2K-1}  \big ), \nonumber
     \end{equation}
     and when $\gamma =\tfrac{M}{N} \geq \tfrac{K}{K+1}$, the optimal sDoF/user are given by
     \begin{equation}
      \text{DoF/user}\leq \min \big(\max \big (\tfrac{N}{K+1},\tfrac{M}{K+1} \big ), \tfrac{N}{K}  \big ). \nonumber
     \end{equation}
     \label{optimalDoF_theorem_generalK}
    \end{theorem}
    The proof of this theorem follows directly from the outer bounds established in Section \ref{outerbounds} and designing beamformers using the notion of packing ratios. The optimal sDoF in the interval $(0, \tfrac{1}{K-1}]$ consists of four piecewise-linear regions and a combination of random beamforming in the uplink and beamformers having a packing ratio of $K:(K-1)$ achieves the optimal sDoF. When $\gamma \geq \tfrac{K}{K+1}$, the optimal DoF consists of three piecewise-linear regions achieved using a combination of zero-forcing beamformers and beamformers having packing ratio $K:1$.

    Extending the notion of packing ratios to any general cellular network and for all values of $\gamma$ requires us to first identify the set of fundamental packing ratios that play a crucial role in identifying the best set of beamformers that can be designed for any given system. Identifying these fundamental packing ratios requires an understanding of how multiple subspaces in a large network network interact. In the absence of a coherent theory characterizing such interactions, this is a major bottleneck in extending packing ratios to general cellular networks. Different from the approach taken here, the notion of subspace alignment chains of  \cite{chenweiwang} proves useful in establishing the optimal-DoF of the three-user interference channel, while \cite{liu-yang-arxiv} proposes a notion called irresolvable subspace chains to construct DoF-optimal beamformers for general cellular networks.

%% file: structureagnosticapproach.tex
\section{Linear Beamforming Design: Unstructured Design}
\label{USAP}
In contrast to the structured approach presented previously, we develop an alternative approach to designing linear beamformers by relying on random linear vector equations to satisfy (\ref{altIAreq1}). Since this approach does not require us to explicitly infer the underlying structure of interference alignment, it bypasses the need for counting arguments and is applicable to a wide class of cellular networks. We call this the unstructured approach (USAP) to designing linear beamformers for interference alignment and discuss the scope and limitations of such an approach. 

Our main observation is the following. For any $(G,K,M,N)$ network, in the regime where the proper-improper boundary lies above the decomposition based inner bound, i.e., $\big (\frac{MN}{KM+N}< \frac{M+N}{GK+1} \big )$, an unstructured approach appears to be able to achieve the optimal sDoF. The sDoF obtained numerically from this unstructured approach matches the optimal sDoF characterized in a parallel and independent work \cite{liu-yang-arxiv} using a structured approach. The key advantage of the unstructured approach advocated in this paper is that it is conceptually much simpler. Further, it is also achieves a significant portion of the DoF in the regime where decomposition based inner bound lies above the proper-improper boundary. The broad applicability of the unstructured approach with minimal dependence on network parameters provides a single unified technique for linear beamforming design in MIMO cellular networks. This approach along with the asymptotic scheme of \cite{cadambejafar} form the two main techniques needed to establish the optimal DoF of MIMO cellular networks. The remainder of this section describes the unstructured approach and presents the results of numerical experiments that identify the scope and limitations of this approach. 

\subsection{The Unstructured Approach}
Consider a $(G,K,M,N)$ cellular network with the goal of achieving $d$ DoF/user without any symbol extensions. In the uplink, note that each BS observes $GKd$ streams of transmission of which $(G-1)Kd$ streams constitute interference. Setting aside $Kd$ dimensions at each BS for the received signals from the in-cell users, to satisfy (\ref{altIAreq1}) the $(G-1)Kd$ interfering data streams must occupy no more than $N-Kd$ dimensions at each BS. Assuming $(G-1)Kd>N-Kd$ (no interference alignment is necessary otherwise), we require the $(G-1)Kd$ transmit beamformers of the interfering signals to satisfy $GKd-N$ $(=L)$ distinct linear equations. In other words, for the $i$th BS, we require

\begin{align}
 \sum_{ l=1,l\neq i}^G\sum_{m=1}^K\sum_{n=1}^d \alpha_{lmn,i}^p \bt H_{(lm,i)} \bt v_{lmn}=\bt 0,
\end{align}
\begin{figure*}
\begin{align}
\begin{bmatrix}
\bt 0_{4\times 3} & \bt 0_{4\times 3} & \alpha_{211,1}^1\bt H_{21,1} & \alpha_{221,1}^1 \bt H_{22,1} & \alpha_{311,1}^1 \bt H_{31,1} & \alpha_{321,1}^1 \bt H_{32,1} \\
\bt 0_{4\times 3} & \bt 0_{4\times 3} & \alpha_{211,1}^2 \bt H_{21,1} & \alpha_{221,1}^2 \bt H_{22,1} & \alpha_{311,1}^2 \bt H_{31,1} & \alpha_{321,1}^2 \bt H_{32,1} \\
\alpha_{111,2}^1\bt H_{11,2} & \alpha_{121,2}^1\bt H_{12,2} & \bt 0_{4\times 3} & \bt 0_{4\times 3} & \alpha_{311,2}^1\bt H_{31,2} & \alpha_{321,2}^1\bt H_{32,2} \\
\alpha_{111,2}^2\bt H_{11,2} & \alpha_{121,2}^2\bt H_{12,2} & \bt 0_{4\times 3} & \bt 0_{4\times 3} & \alpha_{311,2}^2\bt H_{31,2} & \alpha_{321,2}^2\bt H_{32,2} \\
\alpha_{111,3}^1\bt H_{11,3} & \alpha_{121,3}^1\bt H_{12,3} & \alpha_{211,3}^1\bt H_{21,3} & \alpha_{221,3}^1\bt H_{22,3} & \bt 0_{4\times 3} & \bt 0_{4\times 3} \\
\alpha_{111,3}^2\bt H_{11,3} & \alpha_{121,3}^2\bt H_{12,3} & \alpha_{211,3}^2\bt H_{21,3} & \alpha_{221,3}^2\bt H_{22,3} & \bt 0_{4\times 3} & \bt 0_{4\times 3} 
\end{bmatrix}
\begin{bmatrix}
\bt v_{111} \\ \bt v_{121} \\ \bt v_{211} \\ \bt v_{221} \\ \bt v_{311} \\ \bt v_{321} 
\end{bmatrix}
=\bt 0_{24\times 1}.
\label{long_eq}
\end{align}
\end{figure*}
where $\alpha_{lmn,i}^p$ refers to the coefficient associated with the interfering transmit beamformer $\bt v_{lmn}$ in the $p$th linear equation corresponding to the $i$th BS. Thus, we have $GL$ linear vector equations, each involving a set of $(G-1)Kd$ transmit beamforming vectors. Concatenating the transmit beamforming vectors $\bt v_{lmn}$ into a single vector $\bt v=[\bt v_{111}, \bt v_{112}, \hdots , \bt v_{11d},\hdots , \bt v_{GKd}]$ and by appropriately defining the matrix $\bt M$, the $GL$ linear vector equations can be expressed as the matrix equation $\bt M\bt v=\bt 0$. Note that $\bt M$ is a $GLN \times GKMd$ matrix. 

As an example, for the $(3,2,3,4)$ network with $d=1$, the linear matrix equation $\bt M\bt v=\bt 0$ is given by (\ref{long_eq}).

It is known that for the above example, interference alignment is feasible. In other words, it is known that there exists a set of coefficients $\{\alpha_{lmn,i}^p \}$ such that the system of equations in (\ref{long_eq}) has a non-trivial solution. Note that the matrix $\bt M$ in this case is a $24 \times 18$ matrix (system of 24 equations with 18 unknowns), and that a random choice of coefficients $\{\alpha_{lmn,i}^p \}$ results in a matrix $\bt M$ having full column rank, rendering the system of equations infeasible. Determining the right set of coefficients is non-trivial and highlights a particular difficulty in finding aligned beamformers using the set of equations characterized by $\bt M\bt v=0$.\footnote{A classic example in this context is the three-user interference channel with two antennas at each node, where it is known that 1 DoF per receiver is achievable \cite{cadambejafar}. The matrix $\bt M$ in this case is a $6\times 6$ matrix with no non-trivial solutions to $\bt M\bt v=\bt 0$ unless the coefficients are chosen carefully. The set of aligned transmit beamformers in this case are the eigen vectors of an effective channel matrix, with the coefficients being related to the eigen values of this effective channel matrix.}

Now, suppose we append an additional antenna to each BS, thereby creating a $(3,2,3,5)$ network and then consider designing transmit beamformers to achieve 1 DoF/user, it can be shown that the transmit beamformers now need to satisfy a system of equations of the form $\bt M\bt v=\bt 0$, where $\bt M$ is a $12\times 18$ matrix. It is easy to see that even a random choice of coefficients permits non-trivial solutions to this system of equations. The ability to choose a random set of coefficients is quite significant as instead of solving a set of bilinear polynomial equations for interference alignment, we now only need to solve a set of linear equations. We thus have two networks, namely, the $(3,2,3,4)$ network and the $(3,2,3,5)$ network that significantly differ in how aligned beamformers can be computed. This points to a much broader divide among MIMO cellular networks.

While aligned beamformers satisfy the system of equations $\bt M \bt v=\bt 0$ for a set of coefficients, not all solutions to $\bt M \bt v=\bt 0$ with a fixed set of coefficients form aligned beamformers. A vector $\hat{\bt v}$
satisfying $\bt M \hat{\bt v}=0$,  can be considered to constitute a set of aligned beamformers provided (a) the set of beamformers corresponding to a user are linearly independent, i.e., $\bt V_{ij}$ is full rank $\forall i,j$; (b) the signal received from a user at the intended BS is full rank i.e., $\bt H_{ij,i}\bt V_{ij}$ is full rank; and (c) signal and interference are separable at each BS. Since we assume generic channel coefficients and since direct channels are not used in forming the matrix $\bt M$, (c) is satisfied almost surely, while (b) is true under the assumption of generic channel coefficients provided (a) is true. While the idea of satisfying conditions for interference alignment through random linear equations is also discussed in \cite{tingtingliuWCNC}, the presentation in \cite{tingtingliuWCNC} is limited to achieving 1 DoF/user, thereby avoiding the necessity to check for linear independence of the transmit beamformers.
 
Since $\bt M$ is a $GLN \times GKMd$ matrix, whenever $LN<KMd$ the system of equations $\bt M\bt v=\bt 0$ permits a non-trivial solution for any random choice of coefficients. When $LN<KMd$, a solution to the equation $\bt M\bt v= \bt 0$ can be expressed as $ \hat{\bt v}=\det(\bt M\bt M^H)(\bt I-\bt M^H(\bt M\bt M^H)^{-1} \bt M)\bt r$, where $\bt r$ is a $GKMd \times 1$ vector with randomly chosen entries. For $\hat{\bt v}$ to qualify as a solution for interference alignment, we need to ensure that condition (a) is satisfied, i.e.,  we need to ensure that the set of transmit beamformers $\hat{\bt v}_{ij1}$, $\hat{\bt v}_{ij2} \hdots \hat{\bt v}_{ijd}$ obtained from $\hat{\bt v}$ are linearly independent for any $i \in  \{ 1,2,\hdots, G \}$, $ j\in \{ 1,2,\hdots , K \}$. Letting $\hat{\bt  V}_{ij}$ be the $M \times d$ matrix formed using $\hat{\bt v}_{ij1}$, $\hat{\bt v}_{ij2} \hdots \hat{\bt v}_{ijd}$, checking for linear independence is equivalent to checking if the determinant of the matrix $[\hat{\bt V}_{ij} \, {\bt R}_{ij}]$, where ${\bt R}_{ij}$ is a $(M-d)\times d$ matrix of random entries, is non-zero or not.  

Since the determinant of $[\hat{\bt V}_{ij} \, {\bt R}_{ij}]$ is a polynomial in the variables ${\bt R}_{ij}$, $\bt r$, the coefficients $\{\alpha_{lmn,i}^p \}$, and the channel matrices $\{\bt H_{(lm,i)} \}$, checking for linear independence of the transmit beamformers is equivalent to checking if this polynomial is the zero-polynomial or not. This problem is known as polynomial identity testing (PIT) and is well studied in complexity theory \cite{saxenapit}. While a general deterministic algorithm to solve this problem is not known, a randomized algorithm based on the Schwartz-Zippel lemma \cite{schwartz,zippel} is available and it involves evaluating this polynomial at a random instance of ${\bt R}_{ij}$, $\bt r$, $\{ \alpha_{lmn,i}^p \}$, and $\{\bt H_{lm,i} \}$. If the value of the polynomial at this point is non-zero, then this polynomial is determined to be not identical to the zero-polynomial. Further, it can be concluded that this polynomial evaluates to a non-zero value for almost all values of ${\bt R}_{ij}$, $\bt r$, $\{ \alpha_{lmn,i}^p \}$, and $\{\bt H_{lm,i} \}$. If on the other hand, the polynomial evaluates to the zero, the polynomial is declared to be identical to the zero-polynomial and this statement is true with a very high probability as a consequence of the Schwartz-Zippel lemma.

Thus, whenever $LN<KMd$, we propose a two step approach to designing aligned beamformers. We first pick a set of random coefficients, form the linear equations to be satisfied by the transmit beamformers and compute a set of transmit beamformers by solving the system of linear equations. We then perform the numerical test outlined above to ensure that the transmit beamformers are indeed linearly independent. If the transmit beamformers pass the numerical test then they can be considered to be a set of aligned transmit beamformers. Further, if such a procedure works for a $(G,K,M,N)$ network with $d$ DoF/user for a particular generic channel realization, then it works almost surely for all generic channel realizations of this network. This observation allows us to construct a numerical experiment to verify the limits of using such an approach.

\subsection{Numerical Experiment}

The numerical experiment we perform is outlined as follows. We consider a network with $G$ cells and $K$ users/cell. For this network, we consider all possible pairs of $M$ and $N$ such that $M\leq M_{max}$ and $N\leq N_{max}$, where $M_{max}$ and $N_{max}$ are some fixed positive integers. For a fixed $M$ and $N$, we then consider the feasibility of constructing aligned beamformers using the method described above in order to achieve $d$ DoF/user where $d$ is such that $L>0$\footnote{When $L\leq 0$, random transmit beamforming in the uplink achieves the necessary DoF.}, $LN<KMd$, $d\leq M$, $Kd\leq N$, $M<GKd$ \footnote{When $M\geq GKd$, random transmit beamforming in the downlink achieves the necessary DoF.}, $\text{gcd}(M,N,d)=1$\footnote{Spatial scale invariance states that if $d$ DoF/user are feasible for a $(G,K,M,N)$ network, then $sd$ DoF/user are feasible in a $(G,K,sM,sN)$ network where $s \in \mathcal{Z}^+$ denotes the scale factor. While no proof of such a statement is available, no contradictions 
to this statement exist to the best of our knowledge.} and $(G,K,M,N,d)$ form a proper system. For such a set of $M$, $N$, and $d$, we generate an instance of generic channel matrices and proceed to carry out the two step procedure outlined earlier. Such a procedure is said to be successful if the polynomial test returns a non-zero value and unsuccessful otherwise. If successful, we conclude that such a procedure can be reliably used to design transmit beamformers for almost all channel instances of the $(G,K,M,N,d)$ network under consideration. When unsuccessful, we conclude that with a very high probability such a procedure does not yield a set of aligned transmit beamformers for almost all channel instances.

While we considered designing transmit beamformers in the uplink (USAP-uplink) using random linear vector equations, we can alternately consider designing transmit beamformers in the downlink (USAP-downlink) using the same process. For the $(G,K,M,N,d)$ network, it can be shown that $GK(GKd-M)M<GKdN$ is a necessary condition for the linear system of equations obtained in USAP-downlink to have a non-trivial solution. While there are no significant differences between USAP-uplink and USAP-downlink for the interference channel ($K=1$), a major difference emerges for cellular networks where $K>1$. For cellular networks, when designing transmit beamformers in the downlink, direct channels get involved in the linear system of equations and as a result, a solution to the linear system is no longer guaranteed to satisfy conditions (b) and (c) even when channel coefficients are generic. In this respect, USAP-uplink has a significant advantage over USAP-downlink for cellular networks. In addition, for cellular networks, the 
necessary condition $GK(GKd-M)M<GKdN$ places further restrictions on the applicability of USAP-downlink in the context of achieving the optimal DoF.

We discuss the scope and limitations of USAP-uplink and USAP-downlink in the next section. For clarity, we present our observations for the interference channel $(K=1)$ and the cellular network separately $(K>1)$.

\begin{figure*}[htbp]
      \begin{center}
	\hspace{-0.5cm}
	\psfrag{0y}[Bc][Bc][1]{$0$}
	\psfrag{1y/4}[Bc][Bc][1]{$\frac{1}{G}$}
	\psfrag{ystar}[Br][Br][1]{$\frac{\gamma_l}{\gamma_l+1}$}
	\psfrag{0.5y}[Bc][Bc][1]{$\frac{1}{2}$}
	\psfrag{1x/4}[tc][tc][1]{$\frac{1}{G}$}
	\psfrag{1x/3}[tc][tc][1]{$\frac{1}{3}$}
	\psfrag{1x}[tc][tc][1]{$1$}
\psfrag{0x}[tc][tc][1]{$0$}
	\psfrag{xstar}[tc][tc][1]{$\gamma_l$}
\psfrag{t1}[Bl][Bl][1]{$\frac{\gamma}{\gamma+1}$} 	
\psfrag{t2}[Bl][Bl][1]{$\frac{\gamma+1}{G+1}$}
 	\psfrag{t3}[Bl][Bl][1]{}
\psfrag{t4}[Bl][Bl][1]{$\gamma$}
\psfrag{t5}[Bl][Bl][1]{$\frac{1}{G-\gamma}$}
\psfrag{t6}[Bl][Bl][1]{$\frac{\gamma^2}{G\gamma-1}$}
 	\psfrag{r}[Bl][Bl][1]{I}
\psfrag{s}[Bl][Bl][1]{II}
\psfrag{xlabel}[][][1]{$\gamma \ (M/N)$}
 	\psfrag{ylabel}[Bc][Bc][0.95][180]{Normalized DoF/user}
 	\psfrag{uSAAP applicable}[Bl][Bl][0.9]{USAP-uplink applicable}
 	\psfrag{MAC/BC DoF bound}[Bl][Bl][0.9]{MAC/BC DoF bound}
 	\psfrag{uSAAP prerequisite}[Bl][Bl][0.9]{USAP-uplink necessary condition}
 	\psfrag{proper-improper boundary}[Bl][Bl][0.9]{Proper-improper boundary }
 	\psfrag{random beamforming in uplink}[Bl][Bl][0.9]{Random beamforming in uplink}
 	\psfrag{decomposition inner bound}[Bl][Bl][0.9]{Decomposition based inner bound}
 	\psfrag{dSAAP prerequisite}[Bl][Bl][0.88]{USAP-downlink necessary condition}
\psfrag{piecewise-linear optimal sDoF}[Bl][Bl][0.9]{Piecewise-linear optimal sDoF}
	\includegraphics[width=5in]{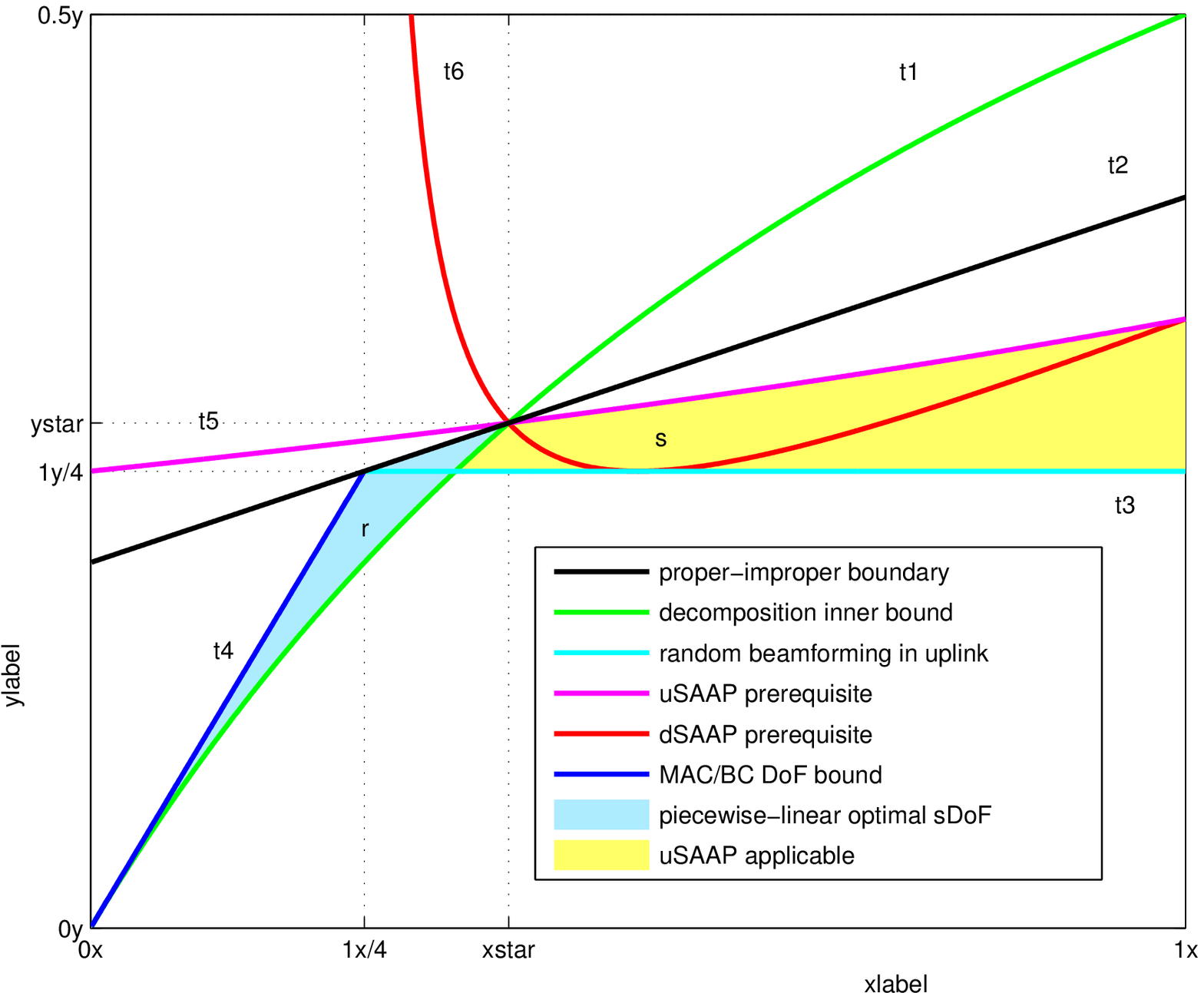}
	\caption{Inner and outer bounds on the DoF of the $G$-user interference channel. The optimal DoF consists of infinitely many piecewise-linear components when $\gamma < \gamma_l$, while the decomposition based approach determines the optimal DoF when $\gamma \geq \gamma_l$.  }
	\label{fig_DoFIC}
      \end{center}

      \begin{center}
	\hspace{-0.5cm}
	\psfrag{3y/10}[Br][Br][1]{$\frac{3}{10}$}
	\psfrag{1y/3}[Br][Br][1]{$\frac{1}{3}$}
	\psfrag{2y/5}[Br][Br][1]{$\frac{2}{5}$}
	\psfrag{3y/7}[Br][Br][1]{$\frac{3}{7}$}
	\psfrag{4y/9}[Br][Br][1]{$\frac{4}{9}$}
	\psfrag{5y/11}[Bc][Bc][1]{}
	\psfrag{1y/2}[Br][Br][1]{$\frac{1}{2}$}
	\psfrag{1/5}[tc][tc][1]{$\frac{1}{5}$}
	\psfrag{1/3}[tc][tc][1]{$\frac{1}{3}$}
	\psfrag{1/2}[tc][tc][1]{$\frac{1}{2}$}
	\psfrag{3/5}[tc][tc][1]{$\frac{3}{5}$}
	\psfrag{2/3}[tc][tc][1]{$\frac{2}{3}$}
	\psfrag{5/7}[tc][tc][1]{$\frac{5}{7}$}
	\psfrag{2/3}[tc][tc][1]{$\frac{2}{3}$}
	\psfrag{3/4}[tc][tc][1]{$\frac{3}{4}$}
	\psfrag{7/9}[tc][tc][1]{$\frac{1}{2}$}
	\psfrag{4/5}[tc][tc][1]{$\frac{4}{5}$}
	\psfrag{9/11}[tc][tc][1]{$\frac{9}{11}$}
	\psfrag{5/6}[tc][tc][1]{$\frac{5}{6}$}
	\psfrag{2/3}[tc][tc][1]{$\frac{2}{3}$}
	\psfrag{3/4}[tc][tc][1]{$\frac{3}{4}$}
 	\psfrag{1/1}[tc][tc][1]{$1$}
 	\psfrag{xlabel}[][][1]{$\gamma \ (M/N)$}
 	\psfrag{ylabel}[Bc][Bc][0.95][180]{Normalized DoF/user}
 	\psfrag{uSAAP successful}[Bl][Bl][0.9]{USAP-uplink successful}
 	\psfrag{uSAAP unsuccessful}[Bl][Bl][0.9]{USAP-uplink unsuccessful}
 	\psfrag{uSAAP prerequisite}[Bl][Bl][0.9]{USAP-uplink necessary condition}
 	\psfrag{proper-improper boundary}[Bl][Bl][0.9]{Proper-improper boundary }
 	\psfrag{random beamforming in uplink}[Bl][Bl][0.9]{Random beamforming in uplink}
 	\psfrag{decomposition inner bound}[Bl][Bl][0.9]{Decomposition based inner bound}
 	\psfrag{BSs=3, Users=1}[Bl][Bl][0.9]{}
	\includegraphics[width=7in]{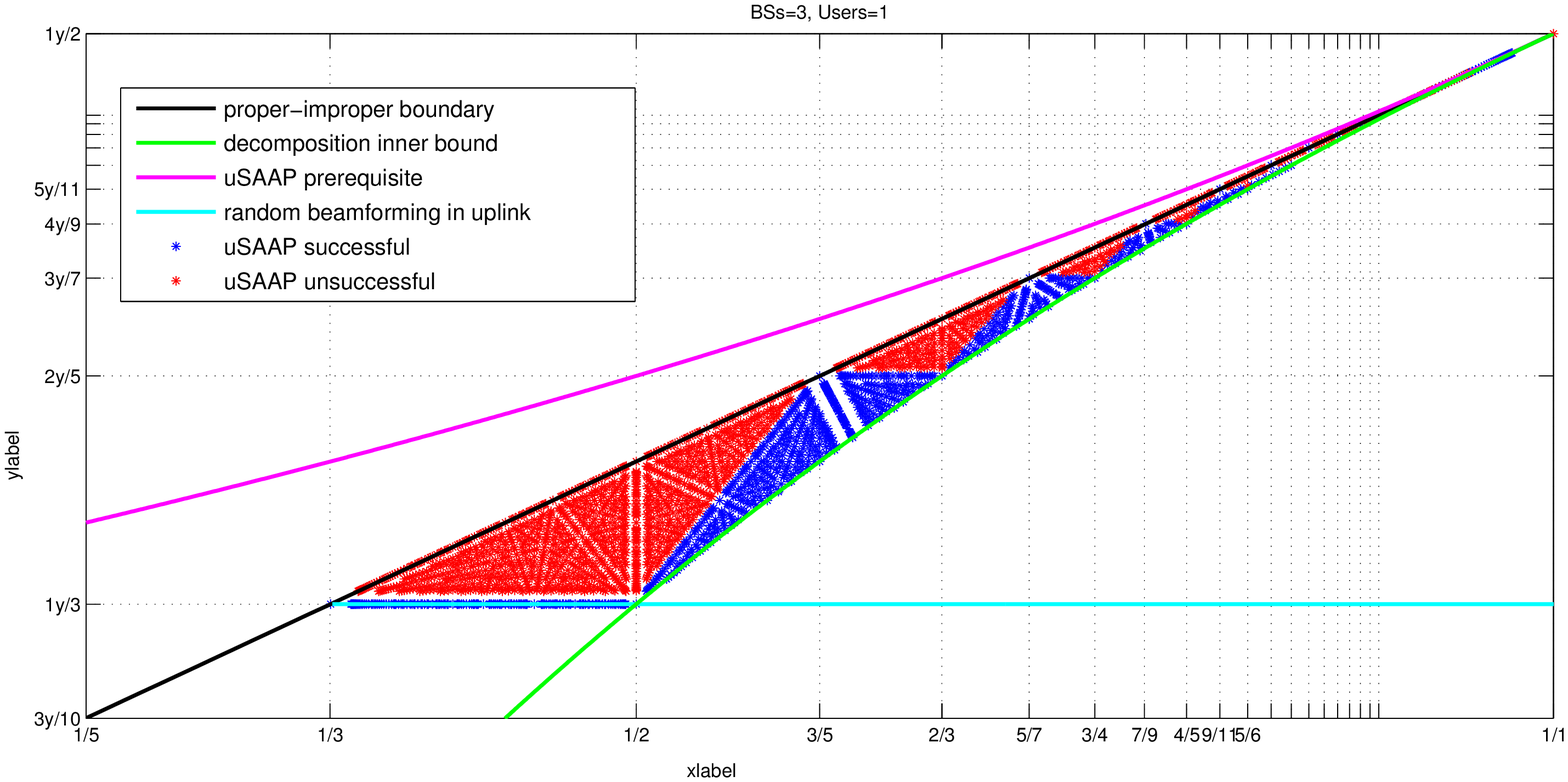}
	\caption{Results of the numerical experiment for the three-user interference channel. Observe that a clear piecewise-linear boundary emerges between the successful and unsuccessful trials of the proposed method. The observed boundary matches with the characterization of the optimal DoF in \cite{chenweiwang}.}
	\label{fig_B3U1}
      \end{center}
    \end{figure*}

\subsection{Unstructured Approach for MIMO Interference Channel}
\label{USAP-IC}

In Fig. \ref{fig_DoFIC} we sketch some well known bounds on the normalized sDoF/user (sDoF/user/$N$) as a function of $\gamma \in (0,\ 1]$ for the $G$-user $(G>3)$ interference channel. By symmetry, it suffices to only consider $\gamma \leq 1$. Except for the three-user interference channel, the proper-improper boundary and decomposition based inner bound intersect at a point $\gamma_l<1$ and this point splits the optimal sDoF characterization into a piecewise-linear region and a smooth region characterized by the decomposition based inner bound \cite{wangjafarisit2012, liu-yang-arxiv}. A simple DoF bound obtained by letting all the BSs or users\footnote{To be consistent with the previous sections, we refer to nodes with $N$ antennas as BSs and nodes with $M$ antennas as users and use the usual notions of uplink and downlink.} cooperate (denoted as MAC/BC DoF bound) is also plotted along with the maximum achievable sDoF using random transmit beamforming in the 
uplink. We also plot the curves characterizing the necessary conditions for USAP-uplink and USAP-downlink to be applicable. It can be shown that these two conditions, the proper-improper boundary and decomposition inner bound all intersect at $\gamma_l=\tfrac{(G-1)-\sqrt{(G-1)^2-4}}{2}$. 

We first narrow our focus to region I (shaded blue) in Fig. \ref{fig_DoFIC}, where the optimal sDoF exhibits a piecewise-linear behavior. For the 3-user interference channel, the point of intersection $\gamma_l$ is equal 1, and a complete characterization of this piecewise-linear behavior for all $\gamma \in (0,\ 1 ]$ is provided in \cite{chenweiwang}. Since region I lies below the necessary condition for USAP-uplink/USAP-downlink, USAP-uplink/USAP-downlink is applicable for any $(M,N,d)$ such that $(M/N,d/N)$ falls in this region. Since the optimal sDoF of the three-user interference channel are known for all $\gamma$, we test the scope of USAP-uplink for this channel.

\begin{figure*}[htbp]

      \begin{center}
	\vspace{0.5cm}
	\hspace{-0.5cm}
	\psfrag{ymin}[Br][Br][1]{}
	\psfrag{1y/4}[Br][Br][0.9]{$1/4$}
	\psfrag{3y/11}[Br][Br][0.9]{$3/11$}
	\psfrag{8y/29}[Br][Br][0.9]{\raisebox{-0.7mm}{$8/29$}}
	\psfrag{ypint1}[Br][Br][0.9]{$\frac{3-\sqrt{5}}{5-\sqrt{5}} \cdots \cdot \cdot$}
	\psfrag{opt1}[tc][tc][1]{$\frac{1}{4}$}
	\psfrag{1x/3}[tc][tc][1]{$\frac{1}{3}$}
	\psfrag{4x/11}[tc][tc][1]{$\frac{4}{11}$}
	\psfrag{3x/8}[tc][tc][1]{$\frac{3}{8}$}
	\psfrag{11x/29}[tc][tc][1]{$\frac{11}{29}$}
	\psfrag{pint1}[cc][Bc][1]{\shifttext{0.4ex}{\raisebox{-7ex}{$\begin{matrix}\vdots \\ \frac{3-\sqrt{5}}{2}\end{matrix}$}}}
	\psfrag{xmin}[tc][tc][1]{$\frac{2}{3}$}
	\psfrag{xlabel}[][][1]{$\gamma \ (M/N)$}
 	\psfrag{ylabel}[Bc][Bc][0.95][180]{Normalized DoF/user}
 	\psfrag{USAP-uplink successful}[Bl][Bl][0.9]{USAP-uplink successful}
 	\psfrag{USAP-uplink unsuccessful}[Bl][Bl][0.9]{USAP-uplink unsuccessful}
 	\psfrag{USAP-uplink prerequisite}[Bl][Bl][0.9]{USAP-uplink necessary condition}
	\psfrag{USAP-downlink prerequisite}[Bl][Bl][0.9]{USAP-downlink necessary condition}
 	\psfrag{proper-improper boundary}[Bl][Bl][0.9]{Proper-improper boundary }
 	\psfrag{random beamforming in uplink/downlink}[Bl][Bl][0.9]{Random beamforming in uplink}
 	\psfrag{decomposition based inner bound}[Bl][Bl][0.9]{Decomposition based inner bound}
	\psfrag{MAC/BC DoF bound}[Bl][Bl][0.9]{MAC/BC DoF bound}
 	\psfrag{?}[Bl][Bl][0.9]{}
	\includegraphics[width=6.8in]{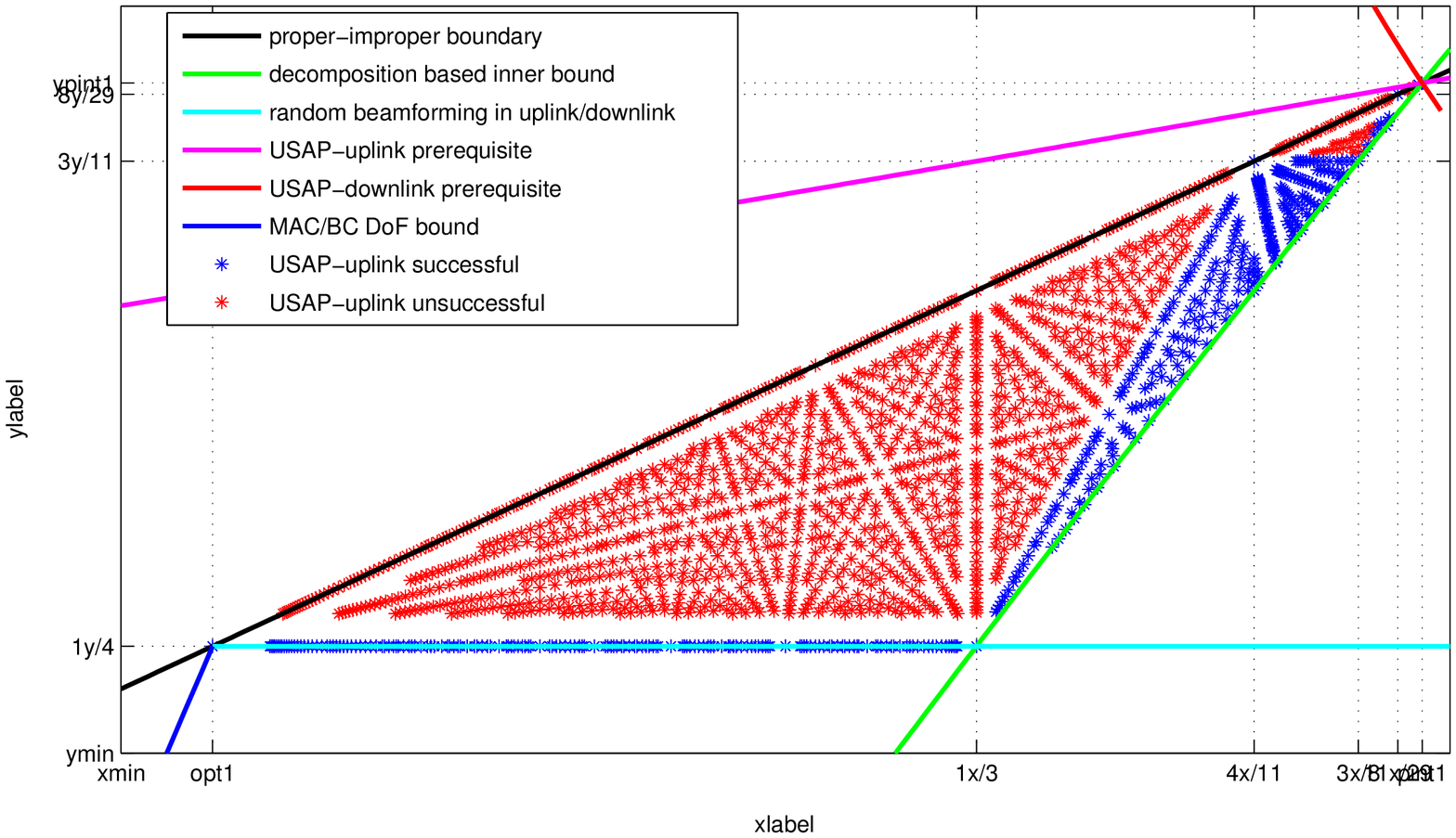}
	\caption{Results of the numerical experiment in region I of the four-user interference channel. Observe that a clear piecewise-linear boundary emerges between the successful and unsuccessful trials of the proposed method. The observed boundary matches with the optimal DoF as characterized in \cite{liu-yang-arxiv}.}
	\label{fig_B4U1b}
      \end{center}

      \begin{center}
	\hspace{-0.5cm}
	\psfrag{?}[Bl][Bl][0.9]{}

	\psfrag{ypint1}[Br][Br][1]{$\frac{3-\sqrt{5}}{5-\sqrt{5}}$}

	\psfrag{2y/5}[Br][Br][1]{$\frac{2}{5}$}
	\psfrag{yopt1}[Br][Br][1]{$\frac{1}{4}$}
	\psfrag{opt1}[tc][tc][1]{$\frac{1}{4}$}
	\psfrag{pint1}[tc][tc][1]{$\frac{3-\sqrt{5}}{2}$}
	
\psfrag{xmin}[tc][tc][1]{$0$}
\psfrag{opt4}[tc][tc][1]{$\frac{1}{3}$}
\psfrag{ymin}[Br][Br][1]{$\frac{3}{20}$}
\psfrag{ymax}[Br][Br][1]{$\frac{2}{5}$}
\psfrag{xmax}[tc][tc][1]{$1$}
\psfrag{t1}[tc][tc][1]{$\frac{\gamma^2}{4\gamma-1}$}
\psfrag{t2}[tc][tc][1]{$\frac{\gamma}{\gamma+1}$}
\psfrag{t3}[tc][tc][1]{$\frac{\gamma+1}{5}$}
\psfrag{t4}[tc][tc][1]{$\frac{1}{4-\gamma}$}
\psfrag{t5}[tc][tc][1]{$\gamma$}

 	\psfrag{xlabel}[][][1]{$\gamma \ (M/N)$}
 	\psfrag{ylabel}[Bc][Bc][0.95][180]{Normalized DoF/user}
 	\psfrag{USAP-uplink successful}[Bl][Bl][0.9]{USAP-uplink successful}
 	\psfrag{USAP-uplink unsuccessful}[Bl][Bl][0.9]{USAP-uplink unsuccessful}
 	\psfrag{USAP-uplink prerequisite}[Bl][Bl][0.9]{USAP-uplink necessary condition}
	\psfrag{USAP-downlink prerequisite}[Bl][Bl][0.9]{USAP-downlink necessary condition}
 	\psfrag{proper-improper boundary}[Bl][Bl][0.9]{Proper-improper boundary }
 	\psfrag{random beamforming in uplink}[Bl][Bl][0.9]{Random beamforming in uplink}
 	\psfrag{decomposition based inner bound}[Bl][Bl][0.9]{Decomposition based inner bound}
	\psfrag{MAC/BC DoF bound}[Bl][Bl][0.9]{MAC/BC DoF bound}
	\includegraphics[width=7in]{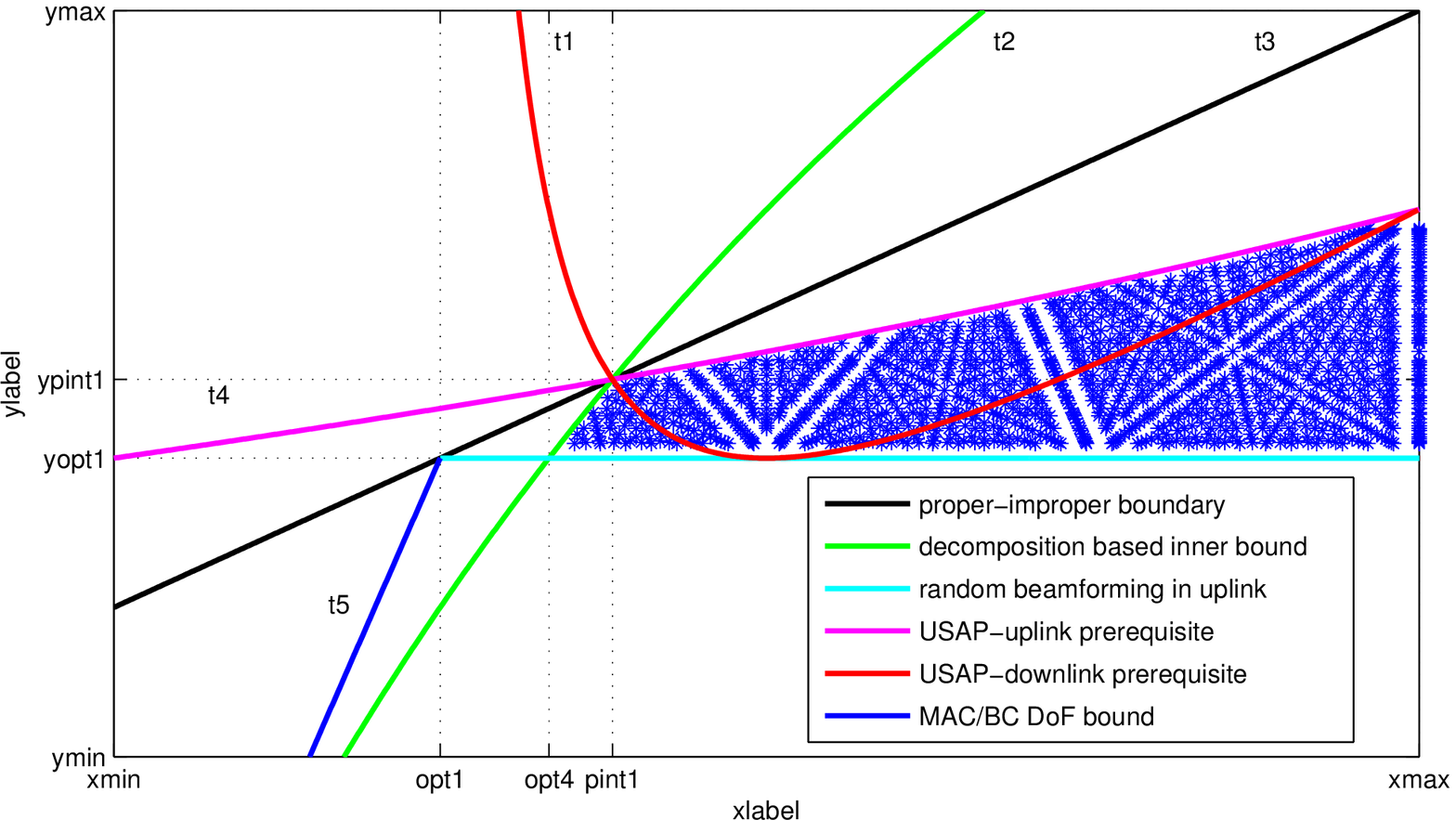}
	\caption{Results of the numerical experiment in region II of the four-user interference channel. Observe that the necessary condition for USAP-uplink completely determines the success of failure of the proposed approach, making the polynomial identity test redundant.}
	\label{fig_B4U1a}
      \end{center}

    \end{figure*}

\label{USAP-uplink-CN}

We carry out the numerical experiment described earlier for the three-user interference channel with values of $M$, $N$, and $d$ such that $(M/N,d/N)$ falls in region I, with $N_{max}=N_{max}=75$. The results of this experiment are shown in Fig. \ref{fig_B3U1}, where we observe that a clear piecewise-linear boundary emerges between the successful and unsuccessful trials on the polynomial identity test. This boundary exactly matches with the piecewise-linear optimal sDoF as detailed in \cite{chenweiwang}, suggesting that such an approach is capable of achieving the optimal sDoF of the three-user interference channel. We also observe that the boundary characterizing the necessary conditions for USAP-uplink has no particular significance and the success or failure of the proposed method is completely determined by the polynomial identity test.

A similar piecewise linear boundary also emerges in the case of the four-user interference channel as seen in Fig \ref{fig_B4U1b} for $\gamma \in ( 0,\ \gamma_l)$. These results are in-line with the results on the optimal sDoF of this network as established in \cite{liu-yang-arxiv}. Further, in contradiction to the conjecture in \cite{wangjafarisit2012}, which states that when $\gamma\geq 3/8$, the decomposition based approach achieves the optimal DoF, we see from Fig. \ref{fig_B4U1b} that the piecewise-linear behavior extends further, all the way up to $\gamma_l$. As an example, numerical experiments show that the ($4,1,11,29$) network has 8 DoF/user, and it is easy to see that this system lies strictly above the decomposition based inner bound. In fact, this is a feasible system lying right on the proper-improper boundary.

These observations lead us to conjecture that for any $G$-user interference channel, whenever $\gamma \in  (0, \gamma_l )$, the optimal sDoF exhibits a piecewise-linear behavior and the optimal sDoF in this regime can be achieved by constructing linear beamformers using the proposed method.   

Shifting focus to region II (shaded yellow) in Fig. \ref{fig_DoFIC}, note that this region lies entirely below the decomposition based inner bound and does not impact the characterization of the optimal sDoF. Also note that this region lies below the proper-improper boundary and the necessary condition for USAP-uplink, thus making USAP-uplink applicable in this region. This region is bounded below by the maximum DoF that can be trivially achieved using random transmit beamforming in the uplink. In order to verify the applicability of USAP-uplink in this region, we carry out the numerical experiment outlined earlier on the four-user interference channel for values of $(M,N,d)$ such that the $(M/N,d/N)$ falls in region II, with $N_{max}=N_{max}=75$. The results are presented in Fig. \ref{fig_B4U1a}, where it is seen that the necessary condition for USAP-uplink, $LN < KMd$, completely determines the success of the proposed method, with the 
subsequent numerical test proving to be redundant. It is also significant to note that these results bring to light a computational boundary that divides systems for which computing transmit beamformers for interference alignment is \textit{easy} (requires solving a system of linear equations; no worse than $O((GKMd)^3))$ in complexity) and systems that require techniques of higher complexity such as iterative algorithms \cite{gomadam,peters2009,dimakis,gokulICASSP} to design such transmit beamformers. 

\begin{figure*}[ht]
      \begin{center}
	\hspace{-0.5cm}
	\psfrag{0y}[Bc][Bc][1]{$0$}
	\psfrag{y1}[Br][Br][1]{\shifttext{1mm}{\raisebox{-0.2ex}{$\frac{1}{GK}$}}}
	\psfrag{y2}[Br][Br][1]{\raisebox{0ex}{$\frac{\gamma_l}{K\gamma_l+1}\cdots\cdot$}}
	\psfrag{y3}[Br][Br][1]{$\frac{\gamma_r}{K\gamma_r+1}$}
	\psfrag{y4}[Br][Br][1]{$\frac{1}{K}$}
	\psfrag{0x}[tc][tc][1]{$0$}
	\psfrag{opt1}[tc][tc][1]{\shifttext{-1ex}{$\frac{1}{GK}$}}
	\psfrag{pint1}[tc][tc][1]{\shifttext{1ex}{\raisebox{-2.5mm}{$\gamma_l$}}}
	\psfrag{1x}[tc][tc][1]{\raisebox{-3mm}{$1$}}
	\psfrag{pint2}[tc][tc][1]{\raisebox{-2.5mm}{$\gamma_r$}}
	\psfrag{opt2}[tc][tc][1]{$(G-1)+\frac{1}{K}$}
	\psfrag{opt3}[tc][tc][1]{\raisebox{-3mm}{$G$}}
	\psfrag{opt5}[tc][tc][1]{$\frac{GK-1}{K}$}
	\psfrag{t4}[Bl][Bl][1]{$\frac{\gamma}{K\gamma+1}$} 	
	\psfrag{t3}[Bl][Bl][1]{\shifttext{-1.5cm}{$\frac{\gamma+1}{GK+1}$}}
		
	\psfrag{t5}[Bl][Bl][1]{\shifttext{-6mm}{$\frac{\gamma}{GK}$}}
	\psfrag{t2}[Bl][Bl][1]{\shifttext{-1.9cm}{$\frac{1}{K(G-\gamma)}$}}
	\psfrag{t1}[Bl][Bl][1]{\shifttext{-0.5cm}{$\frac{\gamma^2}{GK\gamma-1}$}}
	\psfrag{t6}[Bl][Bl][1]{\shifttext{-0.3cm}{$\frac{\gamma}{GK}$}}
	\psfrag{l}[Bl][Bl][1]{\shifttext{-1mm}{I}} 	
	\psfrag{m}[Bl][Bl][1]{II} 	
	\psfrag{n}[Bl][Bl][1]{III} 	
	\psfrag{xlabel}[][][1]{\shifttext{-5mm}{$\gamma \ (M/N)$}}
 	\psfrag{ylabel}[Bc][Bc][0.95][180]{\shifttext{1cm}{Normalized DoF/user}}
 	\psfrag{USAP-uplink applicable}[Bl][Bl][0.9]{USAP-uplink applicable}
 	\psfrag{MAC/BC DoF bound}[Bl][Bl][0.9]{MAC/BC DoF bound}
 	\psfrag{USAP-uplink prerequisite}[Bl][Bl][0.9]{USAP-uplink necessary condtion}
 	\psfrag{proper-improper boundary}[Bl][Bl][0.9]{Proper-improper boundary }
 	\psfrag{random beamforming in uplink/downlink}[Bl][Bl][0.9]{Random beamforming in uplink/downlink}
 	\psfrag{decomposition based inner bound}[Bl][Bl][0.9]{Decomposition based inner bound}
 	\psfrag{USAP-downlink prerequisite}[Bl][Bl][0.9]{USAP-downlink necessary condition}
\psfrag{piecewise-linear optimal sDoF}[Bl][Bl][0.9]{Piecewise-linear optimal sDoF}
	\includegraphics[width=6.5in]{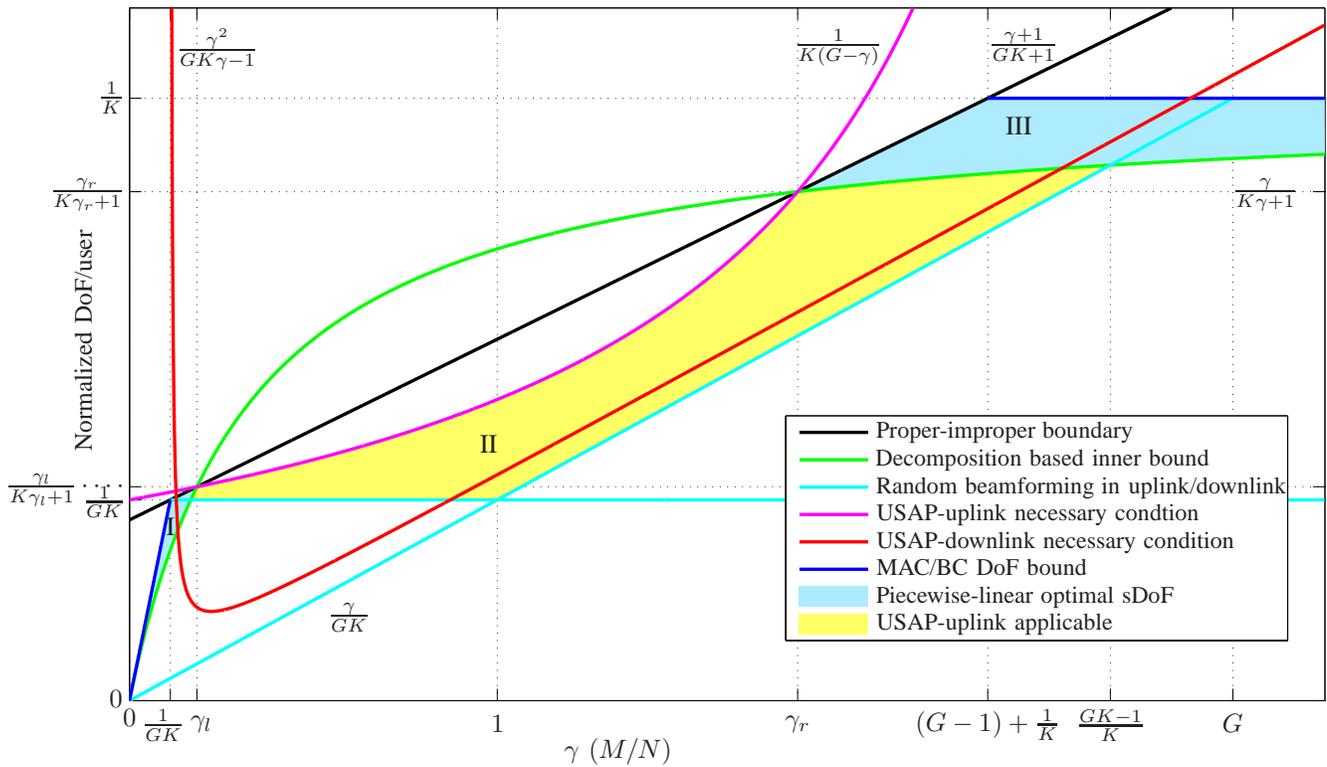}
	\caption{Inner and outer bounds on the DoF of the $G$-cell, $K$-user/cell network. The optimal DoF consists of infinitely many piecewise-linear components for $\gamma < \gamma_l$ and $\gamma > \gamma_r$, while the decomposition based approach determines the optimal DoF when $ \gamma_l \leq \gamma \leq \gamma_r$.  }
	\label{fig_DoFCN}
      \end{center}
    \end{figure*}

So far, except for networks where the underlying structure for interference alignment is known (the three-user interference channel etc.), solving for aligned beamformers of a given network meant solving a system of bilinear equations through computationally intensive iterative algorithms that can sometimes take several thousand iterations to converge \cite{schmidtTSP}. Our observations suggest that except when the DoF demand $d$ placed on a $(G,1,M,N)$ network is such that $\gamma > \gamma_l$ and $(\gamma,d/N)$ is sandwiched between the necessary condition for USAP-uplink and the proper-improper boundary, iterative algorithms are not necessary and that the aligned beamformers can be computed by simply solving a system of linear equations.

It can be shown that USAP-downlink also exhibits a similar piecewise linear behavior whenever $\gamma <\gamma_l$. When $\gamma \geq \gamma_l$, since the necessary condition for USAP-uplink lies above the necessary condition for USAP-downlink, the set of systems that can take advantage of the proposed method remains unchanged.

\subsection{USAP-uplink for MIMO Cellular Networks}

\begin{figure*}[hbtp]
      \begin{center}
	\hspace{-0.5cm}
	\psfrag{0}[Br][Br][1]{0}
	\psfrag{1y/8}[Br][Br][0.75]{$1/8$}
	\psfrag{1y/7}[Br][Br][0.75]{$1/7$}
	\psfrag{3y/20}[Br][Br][0.75]{$3/20$}
	\psfrag{2y/13}[Br][Br][0.75]{$2/13$}
	\psfrag{5y/32}[Br][Br][1]{}
	\psfrag{3y/19}[Br][Br][1]{}
\psfrag{7y/44}[Br][Br][1]{}
\psfrag{4y/25}[Br][Br][1]{}
\psfrag{9y/56}[Br][Br][1]{}
\psfrag{5y/31}[Br][Br][1]{}
\psfrag{11y/68}[Br][Br][1]{}
\psfrag{1y/6}[Br][Br][0.75]{$1/6$}

\psfrag{6y/35}[Br][Br][1]{}
\psfrag{11y/64}[Br][Br][1]{}
\psfrag{5y/29}[Br][Br][1]{}
\psfrag{9y/52}[Br][Br][1]{}
\psfrag{4y/23}[Br][Br][1]{}
\psfrag{7y/40}[Br][Br][1]{}
\psfrag{3y/17}[Br][Br][1]{}
\psfrag{5y/28}[Br][Br][0.75]{}
\psfrag{2y/11}[Br][Br][0.75]{\shifttext{2mm}{$2/11$}}
\psfrag{3y/16}[Br][Br][0.75]{$3/16$}
\psfrag{1y/5}[Br][Br][0.75]{$1/5$} 
\psfrag{1y/4}[Br][Br][0.75]{$1/4$} 
\psfrag{1x/8}[tc][tc][1]{$\frac{1}{8}$}
	\psfrag{1x/4}[tc][tc][1]{$\frac{1}{4}$}
	\psfrag{2x/7}[tc][tc][1]{$\frac{2}{7}$}
	\psfrag{1x/3}[tc][tc][1]{$\frac{1}{3}$}
	\psfrag{7x/20}[tc][tc][1]{$\phantom{.}\frac{7}{20}$}
	\psfrag{3x/8}[tc][tc][1]{$\frac{3}{8}$}
	\psfrag{5x/13}[tc][tc][1]{}
	\psfrag{2x/5}[tc][tc][1]{}
	\psfrag{13x/32}[tc][tc][1]{}
	\psfrag{5x/12}[tc][tc][1]{}
	\psfrag{8x/19}[tc][tc][1]{}
	\psfrag{3x/7}[tc][tc][1]{}
	\psfrag{19x/44}[tc][tc][1]{}

	\psfrag{7x/16}[tc][tc][1]{}
\psfrag{11x/25}[tc][tc][1]{}
\psfrag{4x/9}[tc][tc][1]{}
\psfrag{25x/56}[tc][tc][1]{}
\psfrag{9x/20}[tc][tc][1]{}
\psfrag{14x/31}[tc][tc][1]{}
\psfrag{5x/11}[tc][tc][1]{}
\psfrag{1x/2}[tc][tc][1]{$\frac{1}{2}$}
\psfrag{35x/64}[tc][tc][1]{}
\psfrag{11x/20}[tc][tc][1]{}
\psfrag{16x/29}[tc][tc][1]{}
\psfrag{5x/9}[tc][tc][1]{}
\psfrag{29x/52}[tc][tc][1]{}
\psfrag{9x/16}[tc][tc][1]{}
\psfrag{13x/23}[tc][tc][1]{}
\psfrag{4x/7}[tc][tc][1]{}
\psfrag{23x/40}[tc][tc][1]{}
\psfrag{7x/12}[tc][tc][1]{}
\psfrag{10x/17}[tc][tc][1]{}
\psfrag{3x/5}[tc][tc][1]{}
\psfrag{17x/28}[tc][tc][1]{}
\psfrag{5x/8}[tc][tc][1]{}
\psfrag{7x/11}[tc][tc][1]{$\frac{7}{11}$}
\psfrag{2x/3}[tc][tc][1]{$\frac{2}{3}$}
\psfrag{11x/16}[tc][tc][1]{$\frac{11}{16}$}
\psfrag{3x/4}[tc][tc][1]{$\frac{3}{4}$}
\psfrag{4x/5}[tc][tc][1]{$\frac{4}{5}$}
\psfrag{1x}[tc][tc][1]{$1x$}
\psfrag{5x/4}[tc][tc][1]{$\frac{5}{4}$}

	\psfrag{1x}[tc][tc][1]{$1$}
	\psfrag{0x}[tc][tc][1]{$0$}
 	\psfrag{xlabel}[][][1]{$\gamma \ (M/N)$}
 	\psfrag{ylabel}[Bc][Bc][0.95][180]{Normalized DoF/user}
 	\psfrag{uSAAP successful}[Bl][Bl][0.9]{USAP-uplink successful}
 	\psfrag{uSAAP unsuccessful}[Bl][Bl][0.9]{USAP-uplink unsuccessful}
 	\psfrag{uSAAP prerequisite}[Bl][Bl][0.9]{USAP-uplink necessary condition}
	\psfrag{dSAAP prerequisite}[Bl][Bl][0.9]{USAP-downlink necessary condition}
 	\psfrag{proper-improper boundary}[Bl][Bl][0.9]{Proper-improper boundary }
 	\psfrag{random beamforming in uplink}[Bl][Bl][0.9]{Random beamforming in uplink}
 	\psfrag{decomposition inner bound}[Bl][Bl][0.9]{Decomposition based inner bound}	
	\psfrag{BSs=2, Users=4}[Bl][Bl][0.9]{}
	\includegraphics[width=7in]{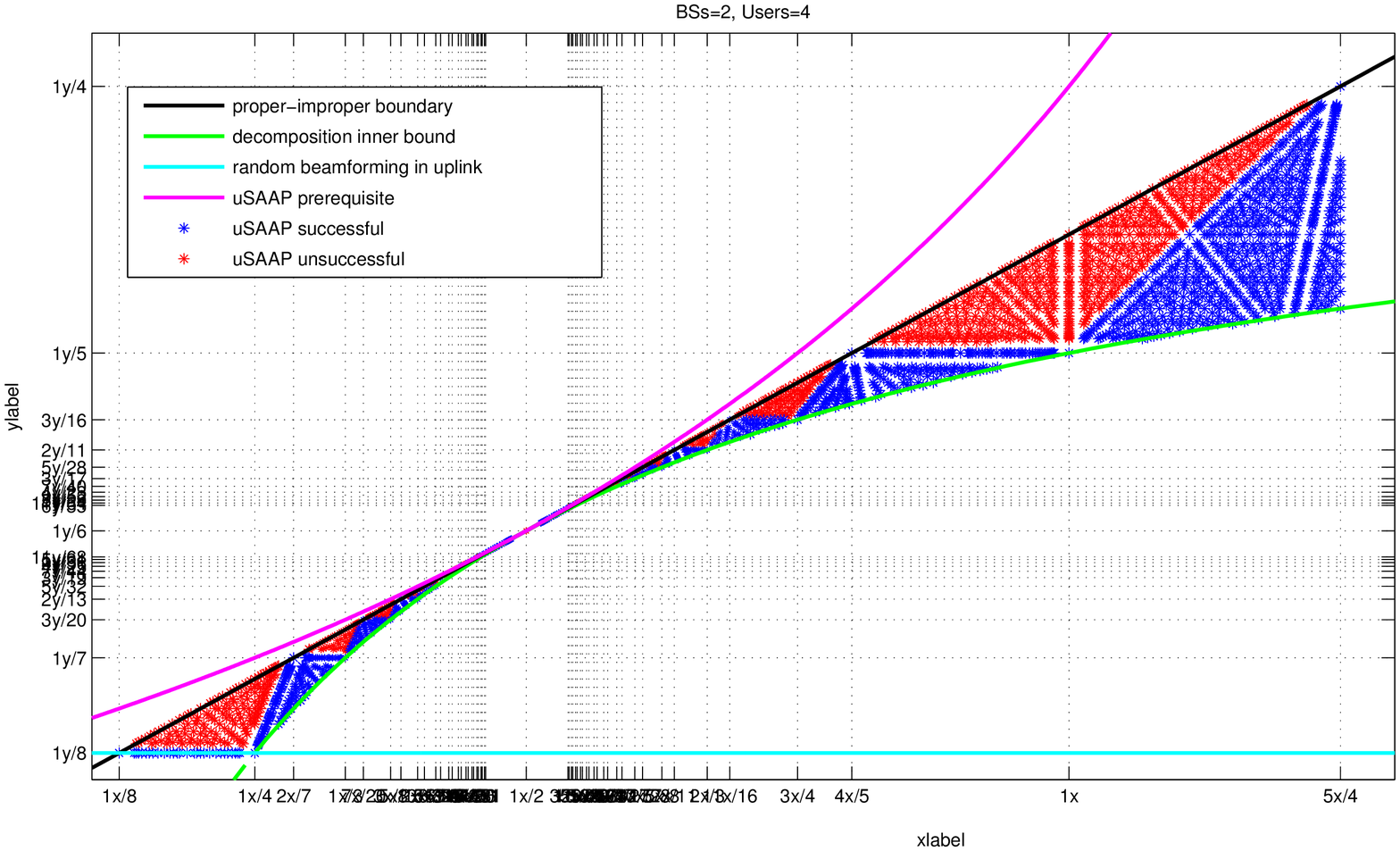}
	\caption{Results of the numerical experiment for the two-cell, four-user/cell network. Note the clear piecewise-linear boundary that emerges between the successful and unsuccessful trials of the proposed method. The observed boundary matches with the result in \cite{liu-yang-arxiv}.}
	\label{fig_B2U4}
      \end{center}

      \begin{center}
	\hspace{-0.5cm}
\psfrag{0y}[Br][Bc][1]{$0$}
\psfrag{0x}[tc][cc][1]{$0$}
	\psfrag{y4}[Br][Br][1]{$\frac{1}{2}$}
	\psfrag{y3}[Br][Br][1]{$\frac{2\sqrt{2}+1}{7\sqrt{2}}\cdot \cdot$}
	\psfrag{y2}[Br][Br][1]{$\frac{2\sqrt{2}-1}{7\sqrt{2}}\cdot \cdot$}
	\psfrag{0.34y}[Br][Br][1]{}
	\psfrag{y1}[Br][Br][1]{$\frac{1}{6}$}
	\psfrag{1x/8}[tc][tc][1]{$\frac{1}{8}$}
	\psfrag{1x/4}[tc][tc][1]{$\frac{1}{4}$}
	\psfrag{pint1}[tc][tc][1]{$\frac{\sqrt{2}-1}{\sqrt{2}}$}
	\psfrag{pint2}[tc][tc][1]{$\frac{\sqrt{2}+1}{\sqrt{2}}$}
	\psfrag{opt1}[tc][tc][1]{$\frac{1}{6}$}
	\psfrag{opt2}[tc][tc][1]{$\frac{5}{2}$}
	\psfrag{opt3}[tc][tc][1]{$3$}
	\psfrag{1x}[tc][tc][1]{$1$}
	
 	\psfrag{t1}[Bl][Bl][1]{$\frac{\gamma^2}{6\gamma-1}$}

 	\psfrag{t3}[Bl][Bl][1]{$\frac{\gamma+1}{7}$}
 	\psfrag{t2}[Bl][Bl][1]{$\frac{1}{2(3-\gamma)}$}
 	\psfrag{t4}[Bl][Bl][1]{$\frac{\gamma}{2\gamma+1}$}
 	\psfrag{t5}[Bl][Bl][1]{$\frac{1}{3-\gamma}$}
\psfrag{t5}[Bl][Bl][1]{$\frac{\gamma}{6}$}
 	\psfrag{xlabel}[][][1]{\raisebox{-2ex}{$\gamma \ (M/N)$}}
 	\psfrag{ylabel}[Bc][Bc][0.95][180]{Normalized DoF/user}
 	\psfrag{MAC/BC DoF bound}[Bl][Bl][0.9]{MAC/BC DoF bound}
 	\psfrag{USAP-uplink prerequisite}[Bl][Bl][0.9]{USAP-uplink necessary condition}
 	\psfrag{proper-improper boundary}[Bl][Bl][0.9]{Proper-improper boundary }
 	\psfrag{random beamforming in uplink/downlink}[Bl][Bl][0.9]{Random beamforming in uplink/downlink}
 	\psfrag{decomposition based inner bound}[Bl][Bl][0.9]{Decomposition based inner bound}
 	\psfrag{USAP-downlink prerequisite}[Bl][Bl][0.9]{USAP-downlink necessary condition}

 	\psfrag{USAP-uplink successful}[Bl][Bl][0.9]{USAP-uplink successful}
	\psfrag{piecewise-linear optimal sDoF}[Bl][Bl][0.9]{Piecewise-linear optimal sDoF}	
	\psfrag{BSs=3, Users=2}[Bl][Bl][0.9]{}
	\includegraphics[width=7in]{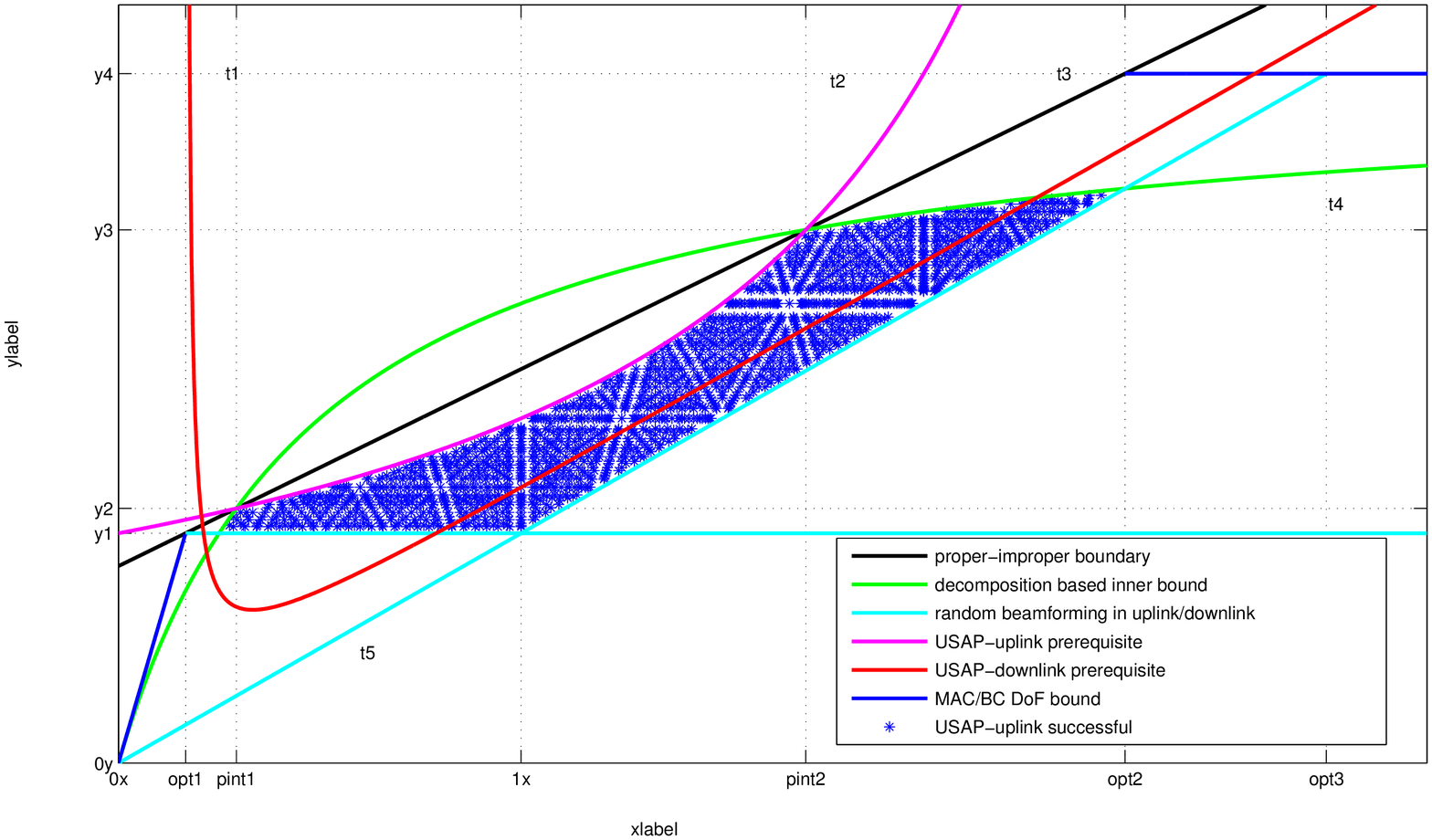}
	\caption{Results of the numerical experiment in region II of the three-cell, two-user/cell network. Observe that the necessary condition for USAP-uplink completely determines the success of failure of the proposed approach.}
	\label{fig_B3U2}
      \end{center}
    \end{figure*}

Fig. \ref{fig_DoFCN} is a sketch analogous to Fig. \ref{fig_DoFIC} and applies to any MIMO cellular network, with the exception of the two-cell, two-user/cell and the two-cell, three-user/cell networks. Note that $\gamma$ is no longer restricted to $(0,1 ]$. While the necessary condition for USAP-uplink, the proper-improper boundary and the decomposition based inner bound all intersect at the same two points $\gamma_l$ and $\gamma_r$, the same is not true for the necessary condition of USAP-downlink. The points of intersection $\gamma_l$ and $\gamma_r$ can be computed to be the points $\tfrac{K(G-1)\pm \sqrt{K^2(G-1)^2-4K}}{2K}$. The optimal sDoF of a general cellular network is recently investigated in \cite{liu-yang-arxiv}. The optimal sDoF as characterized in \cite{liu-yang-arxiv} has a piecewise-linear behavior in regions I ($\gamma < \gamma_l$) and III ($\gamma> \gamma_r$) (see Fig. \ref{fig_DoFCN}). Based on the results in \cite{wangjafarisit2012} for the MIMO interference channel, the decomposition based inner bound is likely to characterize the optimal DoF whenever $\gamma_l \leq\gamma \leq \gamma_r$.

Focusing on regions I and III, we note that USAP-uplink is applicable to all points in these two regions. To gain insight on the scope of this technique for cellular networks, we perform the numerical experiment outlined earlier for the 2-cell 4-user/cell network. For this network, the proper-improper boundary and the decomposition based inner bound touch each other at $\gamma=1/2$, i.e., $\gamma_l=\gamma_r=1/2$, with the decomposition based inner bound lying entirely below the proper-improper boundary. The results of the numerical experiment are plotted in Fig. \ref{fig_B2U4} and it is easy to see that a clear piecewise linear boundary emerges between the successful and unsuccessful trials, with the successful or failure of the proposed method completely determined by the polynomial identity test.

Remarkably, the boundary of the achievable sDoF determined by our unstructured approach matches with the optimal sDoF claimed in \cite{liu-yang-arxiv}. This leads us to conjecture that for any $G$-cell $K$-user/cell cellular network with $(G,K) \notin \{(2,2),\ (2,3) \}$, when $\gamma \in  (0, \gamma_l ) \cup (\gamma_r \infty)$ the optimal sDoF can be achieved by constructing linear beamformers using the proposed method. Further, the optimal sDoF in this regime exhibits a piecewise linear behavior as also observed in \cite{liu-yang-arxiv}, where a structured approach to linear beamforming based in irresolvable subspace chains is used to establish these results, unlike the approach discussed here.

Observations on the applicability of USAP-uplink in region II\footnote{Note that for cellular networks with $G>4$, the inner bound obtained through random transmit beamforming in the downlink $(GKd\leq M)$ and the USAP-uplink's necessary condition $(LN < KMd)$ intersect at two points, thereby splitting region II into two separate parts. This does not alter any of the observations made in this section.} are similar to observations made in the context of the interference channel. By running the numerical experiment on the 3-cell, two-user/cell network for $(M,N,d)$ such that $(M/N,d/N)$ lies in region II, we note from Fig.~\ref{fig_B3U2} that the necessary condition $LN < KMd$ also ensures the success of the polynomial identity test. It is thus seen that even in the regime where $\gamma_l \leq \gamma \leq \gamma_r $, a significant portion of the achievable sDoF can be achieved using the unstructured approach.

A major difference between interference channels and cellular networks arises with respect to the scope and limitations of USAP-downlink. It is clear from Fig. \ref{fig_DoFCN} that due to the nature of the necessary condition associated with USAP-downlink, USAP-downlink cannot be used to establish the same piecewise linear behavior in regions I and III, as observed with USAP-uplink. Further, as stated earlier, since direct channels get involved in the linear system generated by USAP-downlink, verifying that a solution to the linear system also satisfies conditions for interference alignment involves further checks such as ensuring the separability of signal and interference. Due to these reasons, the utility of USAP-downlink for cellular networks is quite limited and offers no particular advantages over USAP-uplink.

%% file: genieouterboundB.tex
\label{2by3bound_section}
In this section we show that for the two-cell three-users/cell MIMO cellular network whenever $\frac{5}{9}  \leq \gamma  \leq \frac{3}{4}$, no more than $ \max \left ( \tfrac{2N}{9}, \tfrac{M}{3} \right )$ DoF/user are possible. Since there is no duality associated with the information theoretic proof presented here, we need to establish this result separately for uplink and downlink. Similar to \cite{chenweiwang}, we first perform an invertible linear transformation at the users and the base-stations. The linear transformation involves multiplication by a full rank matrix at each user and BS. Let the $M\times M$ transformation matrix at user $(i,j)$ be denoted as $\bt T_{ij}$ and the $N\times N$ transformation matrix at BS $\bar{i}$ be denoted as $\bt R_{\bar{i}}$. Using these transformations the effective channel between user $(i,j)$ and BS $\bar{i}$ is given by $\bt R_{\bar{i}} \bt H_{(ij,\bar{i})}\bt T_{ij}$. Subsequent to this transformation, we first consider the uplink scenario and identify genie signals that enable the BSs to decode all the messages in the network and set up a bound on the sum-rate of the network. Using the same transformation, we then identify genie signals to establish the bound in the downlink. We start by considering the case when $5/9 \leq \gamma \leq 2/3$.

Throughout this section we use the relative indices $i$ and $\bar{i}$ when referring to the two cells and use the notation $ij$ to denote the $j$th user in $i$th cell. The vector random variables corresponding to the transmit signal $\bt x$, received signal $\bt y$ and additive noise $\bt z$  are denoted as $\bt X$, $\bt Y$ and $\bt Z$, respectively. $W$ denotes a uniform discrete random variable associated with the transmitted message at a transmitter.

\subsubsection{DoF Outer Bound When $5/9 \leq \gamma \leq 2/3$}

We divide the set of $N$ antennas at BS $\bar{i}$ into three groups and denote them as $\bar{i}a$, $\bar{i}b$ and $\bar{i}c$. The sets $\bar{i}a$ and $\bar{i}c$ contain the first and last $N-M$ antennas each while set $\bar{i}b$ has the remaining $2M-N$ antennas. Let the $M$ antennas at user $ij$ be denoted as $ijk$ where $k \in \{1,2,\cdots , M \}$. Using a similar notation for BS antennas, let  $\bt H_{(ij,\bar{i}p:\bar{i}q)}$ represent the channel from user $ij$ to the subset of BS antennas from the $p^{th}$ antenna to the $q^{th}$ antenna.

We first focus on the $N \times M$ channel from user $i1$ to BS $\bar{i}$. We set the first $N-M$ rows of $\bt R_{\bar{i}}$ to be orthogonal to the columns of $\bt H_{ij, \bar{i}}$. Since $\bt H_{(ij, \bar{i})}$ spans only $M$ of the $N$ dimensions at BS $\bar{i}$, it is possible to choose such a set of vectors. Similarly, the next $2M-N$ and $N-M$ rows of $\bt R_{\bar{i}}$ are chosen to be orthogonal to user $i2$ and user $i3$ respectively. Since all channels are assumed to be generic, matrix $\bt R_{\bar{i}}$ is guaranteed to be full rank almost surely.

On the user side, user $i1$ inverts the channel to the last $M$ antennas of BS $\bar{i}$, i.e., $\bt T_{i1}=(\bt H_{(i1,\bar{i}N-M+1:\bar{i}N)})^{-1}$, while user $i3$ inverts the channel to the first $M$ antennas of BS $\bar{i}$, i.e., $\bt T_{i3}=(\bt H_{(i1,\bar{i}1:\bar{i}M)})^{-1}$. We let $\bt T_{i2}=\bt I$. The signal structure resulting from such a transformation is shown in Fig. \ref{2by3modelD}.

\begin{figure*}[t]
      \begin{center}
	\psfrag{x1}{$\bar{i}a$}
	\psfrag{x2}{$\bar{i}b$}
	\psfrag{x3}{$\bar{i}c$}
	\psfrag{s1}[Bl][Bl][0.8]{$N-M$}
	\psfrag{s2}[Bl][Bl][0.8]{$2M-N$}
	\psfrag{s3}[Bl][Bl][0.8]{$N-M$}
	\psfrag{x4}[Bl][Bl][0.8]{$i1$}
	\psfrag{x5}[Bl][Bl][0.8]{$i2$}
	\psfrag{x6}[Bl][Bl][0.8]{$i3$}
	\psfrag{s4}[Bl][Bl][0.8]{$M$}
	\psfrag{s5}[Bl][Bl][0.8]{$M$}
	\psfrag{s6}[Bl][Bl][0.8]{$M$}
	\psfrag{s7}[Bl][Bl][0.8]{$g_{i1}(\bt x_{\bar{i}b},\bt x_{\bar{i}c})$}
	\psfrag{s8}[Bl][Bl][0.8]{$g_{i2}(\bt x_{\bar{i}a},\bt x_{\bar{i}c})$}
	\psfrag{s9}[Bl][Bl][0.8]{$g_{i3}(\bt x_{\bar{i}b},\bt x_{\bar{i}c})$}
	\psfrag{t1}[Bl][Bl][0.8]{$g_{\bar{i}a}(\bt x_{i2},\bt x_{i31:i3(N-M)})$}
	\psfrag{t2}[Bl][Bl][0.8]{$g_{\bar{i}b}(\bt x_{i11:i1(2M-N)},\bt x_{i3(N-M+1):i3M})$}
	\psfrag{t3}[Bl][Bl][0.8]{$g_{\bar{i}c}(\bt x_{i1(2M-N+1):i1M}, \bt x_{i2})$}
	\psfrag{u1}[Bc][Bc][0.8]{Users in Cell $i$}
	\psfrag{u2}[][][0.8]{Base-station $\bar{i}$}
	\includegraphics[scale=0.48]{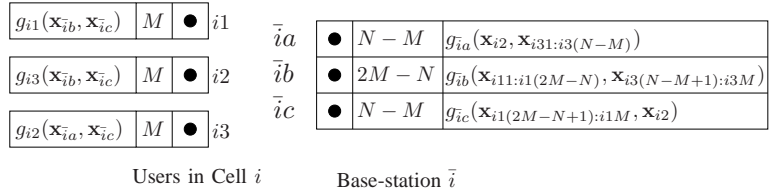}
	\caption{The signal structure obtained after linear transformation for the case when $\gamma \leq 2/3$. Note that the figure does not include signals from the same cell.}
	\label{2by3modelD}
      \end{center}
    \end{figure*}

\paragraph{DoF Bound in the Uplink}

Let $w_{ij}$ be the message from user $ij$ to BS $\bar{i}$. This message is mapped to a $Mn\times 1 $ codeword $ \bt x^n_{ij}$, where $n$ is the length of the code. = We use the notation $\bt x^n_{ijp}$ to denote the transmitted signal on the $k$th antenna over the $n$ time slots and the notation $\bt x_{ijp:ijq}$ to denote the signal transmitted by user $ij$ using antennas $p,p+1,\dots , q$. We denote the rate to user $ij$ as $R_{ij}$, the total sum-rate of the network as $R_{sum}$ and the collection of all messages in the network as $\{ w_{ij}\}$.

Now, consider providing the set of signals $\mathcal{S}_1= \{ \tilde{\bt x}^n_{i2},\tilde{\bt x}^n_{i11:i1(2M-N)}\}$ to BS $\bar{i}$. We use $\tilde{\bt x}^n$ to denote ${\bt x}^n+{\bt z}^n$ where ${\bt z}^n$ is circular symmetric Gaussian noise that is artificially added to the transmitted signal ${\bt x}^n$. Since we seek to establish a converse, we assume that BS $\bar{i}$ can decode all the messages from its users. After decoding and subtracting these signals from the received signal, the  resulting signals at the three antenna sets are given in Fig. \ref{2by3modelD} where $g_{\bar{i}*}(\cdot)$ represents a noisy linear combination of its arguments. Given $\mathcal{S}_1$, we can subtract $\bt x^n_{i2}$ from $g_{\bar{i}c}(x_{i1(2M-N+1):i1M},x_{i2})$ and along with $\tilde{\bt x}^n_{i11:i1(2M-N)}$ from $\mathcal{S}_1$, we can decode $w_{i1}$ subject to noise distortion. After decoding $w_{i1}$, and subtracting $\bt x^n_{i1}$ and $\bt x^n_{i2}$ from the received signal, $w_{i3}$ can also be decoded subject to noise distortion. Since BS $\bar{i}$ can recover all the messages in the network given $\bt y^n_{\bar{i}}$ and $\mathcal{S}_1$ subject to noise distortion, we have
\begin{align}
&nR_{sum} \nonumber \\
&\stackrel{a}{\leq}  I\left ( \{ W_{ij}  \} ; \bt Y^n_{\bar{i}}, \mathcal{S}_1 \right )+n\ms o(\log \rho)+o(n)\nonumber \\
&\stackrel{b}{\leq}  Nn\log \rho +h(\tilde{\bt X}^n_{i2},\tilde{\bt X}^n_{i11:i1(2M-N)}|\bt Y^n_{\bar{i}})+n\ms o(\log \rho)+o(n) \nonumber \\
& \stackrel{c}{\leq}  Nn\log \rho +nR_{i2}+h(\tilde{\bt X}^n_{i11:i1(2M-N)})+n\ms o(\log \rho) +o(n) 
\label{L1}
\end{align}
where (a) follows from Fano's inequality, (b) follows from Lemma 3 in \cite{chenweiwang} and (c) follows from the fact that conditioning reduces entropy. 

Next, consider providing the set of signals $\mathcal{S}_2= \{\tilde{\bt x}^n_{i3},\tilde{\bt x}^n_{i1(2M-N+1):i1M} \}$ to BS $\bar{i}$. After subtracting $\tilde{\bt x}^n_{i3}$ from the received signal, the BS can recover $w_{i2}$ from observations at antenna sets $\bar{i}a$ and $\bar{i}c$ subject to noise distortion. Subsequently, BS $\bar{i}$ can also recover $w_{i1}$ subject to noise distortion. Since BS $\bar{i}$ can recover all messages when provided with the genie signal $\mathcal{S}_2$, using similar steps as before, we obtain 

\begin{align}
& nR_{sum} \nonumber \\
& \stackrel{}{\leq}  I\left ( \{ W_{ij}  \} ; \bt Y^n_{\bar{i}}, \mathcal{S}_2 \right )+n\ms o(\log \rho)+o(n)\nonumber \\
& \stackrel{}{\leq}  Nn\log \rho +h(\tilde{\bt X}^n_{i3},\tilde{\bt X}^n_{i1(2M-N+1):i1M}|\bt Y^n_{\bar{i}})\nonumber \\
& \phantom{\leq}+n\ms o(\log \rho)+o(n) \nonumber \\
&  \stackrel{}{\leq}  Nn\log \rho +nR_{i3}+h(\tilde{\bt X}^n_{i1(2M-N+1):i1M}|\hat{\bt X}^n_{i11:i1(2M-N)})\nonumber \\
& \phantom{\leq}+n\ms o(\log \rho) +o(n) \nonumber \\
&\stackrel{}{\leq}  Nn\log \rho +nR_{i3}+nR_{i1}-h(\hat{\bt X}^n_{i11:i1(2M-N)}) \nonumber \\
& \phantom{\leq}+n\ms o(\log \rho) +o(n) \nonumber \\, 
\label{L2}
\end{align}
where $\hat{\bt X}_{i1}^n$ denotes  $\bt X_{i1}^n$ corrupted by channel noise.

Adding (\ref{L1}) and (\ref{L2}) we get,
\begin{align}
2nR_{sum} \stackrel{}{\leq} & 2nN\log \rho +\hspace{-3mm}\sum_{j=1,2,3}nR_{ij}+n\ms o(\log \rho) +o(n) 
\label{L3}
\end{align}
Using a similar inequality for BS $i$, we can write
\begin{align}
3nR_{sum} \stackrel{}{\leq} & 4nN\log \rho+n\ms o(\log \rho) +o(n) 
\label{L4}
\end{align}
Letting $n \rightarrow \infty$ and $\rho \rightarrow \infty$, we see that DoF/user $\leq \tfrac{2N}{9}$.

\paragraph{DoF Outer Bound in the Downlink}

Using same notation as before, consider providing user $i1$ with the genie signal $\mathcal{S}_1=(w_{i2},w_{i3},\bt x_{\tilde{i}a}^n)$. Since we are interested in establishing an outer bound, we assume all the users in the network can decode their own messages. Since user $i1$ can decode $w_{i1}$, using $\mathcal{S}_1$, user $i1$ can reconstruct $\bt x_{ia}^n$, $\bt x_{ib}^n$ and $\bt x_{ic}^n$, and subtract them from the received signal $\bt y_{i1}^n$. Using the signal obtained after subtracting $\bt x_{ia}^n$, $\bt x_{ib}^n$ and $\bt x_{ic}^n$ from $\bt y_{i1}^n$ and using $\bt x_{\tilde{i}a})$ from $\mathcal{S}_1$, user $i1$ can now decode messages $w_{\bar{i}1}$, $w_{\bar{i}1}$ and $w_{\bar{i}1}$ subject to noise distortion. Since user $i1$ can decode all messages in the network given $\bt y_{i1}^n$ and $\mathcal{S}_1$, we have

\begin{align}
nR_{sum} & \leq I\left (\{W_{ij}\}; \bt Y_{i1}^n,\mathcal{S}_1 \right )+n\ms o(\log \rho)+o(n) \nonumber \\
& \leq nM\log\rho+nR_{i2}+nR_{i3}+h(\tilde{\bt X}_{\bar{i}a}^n|\bt Y_{i1}^n,W_{i2},W_{i3})\nonumber \\
& \phantom{\leq}+n\ms o(\log \rho)+o(n) \nonumber \\
& \leq nM\log\rho+ nR_{i2}+nR_{i3}+h(\tilde{\bt X}_{\bar{i}a}^n|\hat{\bt X}_{\bar{i}b}^n,\hat{\bt X}_{\bar{i}c}^n)\nonumber \\
& \phantom{\leq}+n\ms o(\log \rho)+o(n).
\label{D1}
\end{align}

Next, consider providing user $i3$ with the genie signal $\mathcal{S}_3=(w_{i1},w_{i2},\bt x_{\tilde{i}c}^n)$. Following the exact same steps as before, we get
\begin{align}
nR_{sum} & \leq nM\log\rho+nR_{i2}+nR_{i3}+h(\tilde{\bt X}_{\bar{i}a}^n|\hat{\bt X}_{\bar{i}b}^n,\hat{\bt X}_{\bar{i}c}^n)\nonumber \\
& \phantom{leq} +n\ms o(\log \rho)+o(n).
\label{D3}
\end{align}

Now consider providing user $i2$ with the genie signal $\mathcal{S}_2=(w_{i1},w_{i3}, \tilde{\bt x}_{\bar{i}b}^n,\tilde{\bt x}_{\bar{i}(M+1):\bar{i}(2N-2M)}^n)$. Note that $\bt x_{\bar{i}(M+1):\bar{i}(2N-2M)}^n$ forms a part of the signal $\bt x_{\bar{i}c}^n$. After subtracting the transmitted signals from BS $i$, user $i2$ has $2N-2M$ noisy linear combinations of the signals  $\bt x_{ia}^n$ and $\bt x_{ic}^n$, which along with $\tilde{\bt x}_{\bar{i}b}^n$ from $\mathcal{S}_2$ can be used to decode all the messages from BS $\bar{i}$ subject to noise distortion. As before, we can write

\begin{align}
& nR_{sum} \nonumber \\
& \leq I\left (\{W_{ij}\}; \bt Y_{i1}^n,\mathcal{S}_2 \right )+n\ms o(\log \rho)+o(n) \nonumber \\
& \leq nM\log\rho+nR_{i1}+nR_{i3} \nonumber \\
& \phantom{\leq}+h(\tilde{\bt X}_{\bar{i}b}^n,\tilde{\bt X}_{\bar{i}(M+1):\bar{i}(2N-2M)}^n|\bt Y_{i1}^n,W_{i1},W_{i3})+n\ms o(\log \rho) \nonumber \\
& \phantom{\leq}+o(n) \nonumber \\
& \leq nM\log\rho+ nR_{i1}+nR_{i3}+h(\tilde{\bt X}_{\bar{i}(M+1):\bar{i}(2N-2M)}^n) \nonumber \\
& \phantom{\leq}+h(\tilde{\bt X}_{\bar{i}b}^n|\hat{\bt X}_{\bar{i}a}^n,\hat{\bt X}_{\bar{i}c}^n)+n\ms o(\log \rho)+o(n) \nonumber \\
& \leq nM\log\rho+ nR_{i1}+nR_{i3}+n(2N-3M)\log\rho \nonumber \\
& \phantom{\leq}+h(\tilde{\bt X}_{\bar{i}b}^n|\hat{\bt X}_{\bar{i}a}^n,\hat{\bt X}_{\bar{i}c}^n)+n\ms o(\log \rho)+o(n) \nonumber \\
& \leq n(2N-2M)\log\rho+ nR_{i1}+nR_{i3}+h(\tilde{\bt X}_{\bar{i}b}^n|\hat{\bt X}_{\bar{i}a}^n,\hat{\bt X}_{\bar{i}c}^n) \nonumber \\
& \phantom{\leq}+n\ms o(\log \rho)+o(n) 
\label{D2}
\end{align}

Adding (\ref{D1}), (\ref{D3}) and (\ref{D2}), we get

\begin{align}
& n3R_{sum} \nonumber \\
& \leq n2N\log\rho+ \sum_{j=1}^3 n2R_{ij}+h(\tilde{\bt X}_{\bar{i}a}^n|\hat{\bt X}_{\bar{i}b}^n,\hat{\bt X}_{\bar{i}c}^n) \nonumber \\
& \phantom{\leq}+h(\tilde{\bt X}_{\bar{i}b}^n|\hat{\bt X}_{\bar{i}c}^n,\hat{\bt X}_{\bar{i}a}^n)+h(\tilde{\bt X}_{\bar{i}c}^n|\hat{\bt X}_{\bar{i}b}^n,\hat{\bt X}_{\bar{i}a}^n)+n\ms o(\log \rho)+o(n) \nonumber \\
& \leq n2N\log\rho+ \sum_{j=1}^3 n2R_{ij}+h(\tilde{\bt X}_{\bar{i}a}^n)+h(\tilde{\bt X}_{\bar{i}b}^n|\hat{\bt X}_{\bar{i}a}^n) \nonumber \\
&\phantom{\leq}+h(\tilde{\bt X}_{\bar{i}c}^n|\hat{\bt X}_{\bar{i}b}^n,\hat{\bt X}_{\bar{i}a}^n)+n\ms o(\log \rho)+o(n) \nonumber \\
& \leq n2N\log\rho+ \sum_{j=1}^3 n2R_{ij}+\sum_{j=1}^3 nR_{\bar{i}j}+n\ms o(\log \rho)+o(n) \nonumber \\
\label{D3}
\end{align}

Using a similar inequality for users in cell $\bar{i}$, we can write
\begin{align}
& n6R_{sum} \leq n4N\log\rho+n3R_{sum} +n\ms o(\log \rho)+o(n) \nonumber \\
\label{D3}
\end{align}
Letting $n \rightarrow \infty$ and $\rho \rightarrow \infty$, we see that DoF/user $\leq \tfrac{2N}{9}$.

\subsubsection{DoF Outer Bound when $ 2/3 \leq \gamma \leq 3/4$}

In this case, we again group the antennas at BS $\bar{i}$ into three groups exactly as before. The $M$ antennas at each user are also grouped into three sets as shown in Fig. \ref{2by3modelE}. The linear transformation at BS $\bar{i}$ is also same as before, i.e., each group of antennas zero-forces one of three users.

On the user side, $\bt T_{i1}$ for user $i1$ is chosen such that $i1a$ zero-forces $\bar{i}b$ while $i1b$ and $i1c$ both zero-force $\bar{i}c$. Similarly, $\bt T_{i3}$ is chosen so that $i3c$ zero-forces $\bar{i}b$, while $i3b$ and $i3c$ both zero-force $\bar{i}a$ and finally $\bt T_{i2}$ is chosen such that $i2a$ zero-forces $\bar{i}a$, while $i2b$ and $i2c$ both zero-force $\bar{i}c$. The resulting signal structure at BS $\bar{i}$ after removing signals from Cell $\bar{i}$ is given in Fig. \ref{2by3modelE}.

\begin{figure*}[t]
      \begin{center}
	\psfrag{x0}{$\bar{i}a$}
	\psfrag{x1}{$\bar{i}b$}
	\psfrag{x2}{$\bar{i}c$}
	\psfrag{s0}[Bl][Bl][0.8]{$N-M$}
	\psfrag{s1}[Bl][Bl][0.8]{$2M-N$}
	\psfrag{s2}[Bl][Bl][0.8]{$N-M$}
	\psfrag{x4}[Bl][Bl][1]{$i1a$}
	\psfrag{x5}[Bl][Bl][1]{$i1b$}
	\psfrag{x6}[Bl][Bl][1]{$i1c$}
	\psfrag{x7}[Bl][Bl][1]{$i2a$}
	\psfrag{x8}[Bl][Bl][1]{$i2b$}
	\psfrag{x9}[Bl][Bl][1]{$i2c$}
	\psfrag{x10}[Bl][Bl][1]{$i3a$}
	\psfrag{x11}[Bl][Bl][1]{$i3b$}
	\psfrag{x12}[Bl][Bl][1]{$i3c$}
	\psfrag{s4}[Bl][Bl][0.8]{$N-M$}
	\psfrag{s5}[Bl][Bl][0.76]{$3M-2N$}
	\psfrag{s6}[Bl][Bl][0.8]{$N-M$}
\psfrag{a4}[Bl][Bl][0.8]{$g_{i1a}(\bt x_{\bar{i}c})$}
	\psfrag{a5}[Bl][Bl][0.76]{$g_{i1b}(\bt x_{\bar{i}b})$}
	\psfrag{a6}[Bl][Bl][0.8]{$g_{i1c}(\bt x_{\bar{i}b})$}
	
\psfrag{s7}[Bl][Bl][0.8]{$N-M$}
	\psfrag{s8}[Bl][Bl][0.76]{$3M-2N$}
	\psfrag{s9}[Bl][Bl][0.8]{$N-M$}
\psfrag{a7}[Bl][Bl][0.8]{$g_{i2a}(\bt x_{\bar{i}c})$}
	\psfrag{a8}[Bl][Bl][0.76]{$g_{i2b}(\bt x_{\bar{i}a})$}
	\psfrag{a9}[Bl][Bl][0.8]{$g_{i2c}(\bt x_{\bar{i}a})$}

	\psfrag{s10}[Bl][Bl][0.8]{$N-M$}
	\psfrag{s11}[Bl][Bl][0.76]{$3M-2N$}
	\psfrag{s12}[Bl][Bl][0.8]{$N-M$}
\psfrag{a10}[Bl][Bl][0.8]{$g_{i3a}(\bt x_{\bar{i}b})$}
	\psfrag{a11}[Bl][Bl][0.76]{$g_{i3b}(\bt x_{\bar{i}b})$}
	\psfrag{a12}[Bl][Bl][0.8]{$g_{i3c}(\bt x_{\bar{i}a})$}

	\psfrag{t0}[Bl][Bl][0.8]{$g_{\bar{i}a}(\bt x_{i2b},\bt  x_{i2c},\bt  x_{i3c})$}
	\psfrag{t1}[Bl][Bl][0.8]{$g_{\bar{i}b}(\bt  x_{i1b},\bt  x_{i1c},\bt  x_{i3a},\bt  x_{i3b})$}
	\psfrag{t2}[Bl][Bl][0.8]{$g_{\bar{i}c}(\bt  x_{i1a},\bt  x_{i2a})$}
	\psfrag{u1}[Bc][Bc][0.8]{Users in Cell $i$}
	\psfrag{u2}[][][0.8]{Base-station $\bar{i}$}
	\includegraphics[scale=0.48]{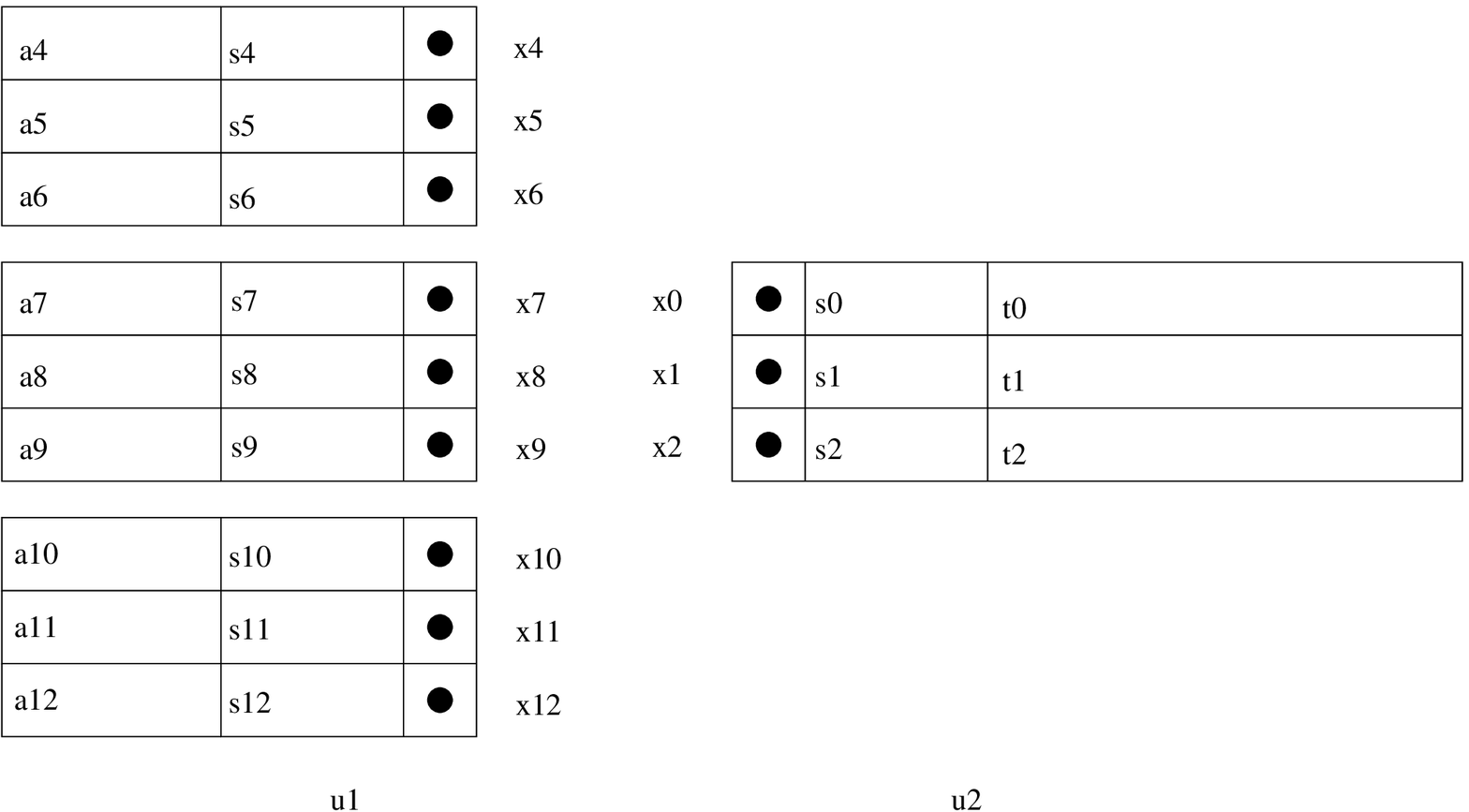}
	\caption{The signal structure obtained after linear transformation when $\gamma \geq 2/3$. The figure does not include signals from the same cell.}
	\label{2by3modelE}
      \end{center}
    \end{figure*}

\paragraph{DoF Outer Bound in the Uplink}
Consider providing the set of signals $\mathcal{S}_1= \{ \tilde{\bt x}^n_{i1}, \tilde{\bt x}^n_{i2b},\tilde{\bt x}^n_{i2c} \}$ to BS $\bar{i}$. After decoding the messages from users in Cell $\bar{i}$, we see that using $\mathcal{S}_3$, we can first decode $w_{i2}$ followed by $w_{i3}$, subject to noise distortion. Since BS $\bar{i}$ can recover all the messages in the network given $\bt y^n_{\bar{i}}$ and $\mathcal{S}_1$, subject to noise distortion, we have
\begin{align}
& nR_{sum} \nonumber \\
& \stackrel{}{\leq}  I\left ( \{ W_{ij}  \} ; \bt Y^n_{\bar{i}}, \mathcal{S}_1 \right )+n\ms o(\log \rho)+o(n)\nonumber \\
& \stackrel{}{\leq}  Nn\log \rho +h(\tilde{\bt X}^n_{i1}, \tilde{\bt X}^n_{i2b},\tilde{\bt X}^n_{i2c} |\bt Y^n_{\bar{i}})+n\ms o(\log \rho)+o(n) \nonumber \\
& \stackrel{}{\leq}  Nn\log \rho +nR_{i1}+h(\tilde{\bt X}^n_{i2b},\tilde{\bt X}^n_{i2c}|\hat{\bt X}^n_{i2a})+n\ms o(\log \rho) +o(n) \nonumber \\
& \stackrel{}{\leq}  Nn\log \rho +\hspace{-1mm}nR_{i1}+\hspace{-1mm}nR_{i2}-\hspace{-1mm}h(\hat{\bt X}^n_{i2a})+n\ms o(\log \rho) +o(n),
\label{U1}
\end{align}
where $\hat{\bt X}^n_{i2a}$ denotes $\bt X^n_{i2a}$ corrupted by channel noise. 

Next, we consider the genie signal $\mathcal{S}_2= \{ \tilde{\bt x}^n_{i3}, \tilde{\bt x}^n_{i2a},\tilde{\bt x}^n_{i2b} \}$. It can once again be shown that BS $\bar{i}$ can recover all the messages in the network given $\bt y^n_{\bar{i}}$ and $\mathcal{S}_2$. Going through similar steps as before, it can be shown that 
\begin{align}
& nR_{sum} \stackrel{}{\leq}  (3M-N)n\log \rho+nR_{i3}+h(\tilde{\bt X}^n_{i2a}) \nonumber \\
& \hspace{2cm}+n\ms o(\log \rho) +o(n) .
\label{U2}
\end{align}
Adding (\ref{U1}) and (\ref{U2}), we get
\begin{align}
2nR_{sum} \stackrel{}{\leq} & 3Mn\log \rho +\hspace{-3mm}\sum_{j=1,2,3}nR_{ij}+n\ms o(\log \rho) +o(n) .
\label{U3}
\end{align}
By symmetry we must also have an analogous inequality involving the rates $R_{\bar{i}j}$, and adding these two inequalities, we get
\begin{align}
3nR_{sum} \stackrel{}{\leq} & 6Mn\log \rho+n\ms o(\log \rho) +o(n) 
\label{U3}
\end{align}
Letting $n \rightarrow \infty$ and $\rho \rightarrow \infty$, we see that DoF/user $\leq \tfrac{M}{3}$.

\paragraph{DoF Outer Bound in the Downlink}
Consider providing the genie signal $\mathcal{S}_1=\{w_{i2},\ms w_{i3},\ms \tilde{\bt{x}}_{ia}\}$ to user $i1$. It can be shown that user i1 can decode all the messages in the network using the received signal and the genie signal subject to noise distortion. Hence, using similar steps as before, we can write

\begin{align}
 nR_{sum} 
& \stackrel{}{\leq}  I\left ( \{ W_{ij}  \} ; \bt Y^n_{i1}, \mathcal{S}_1 \right )+n\ms o(\log \rho)+o(n)\nonumber \\
& \stackrel{}{\leq} nM\log \rho +R_{i2}+R_{i3}+h(\tilde{\bt{X}}_{\bar{i}a}|\hat{\bt{X}}_{\bar{i}a},\hat{\bt{X}}_{\bar{i}a})\nonumber \\
&+n\ms o(\log \rho)+o(n)\nonumber \\
\label{B1}
\end{align}

Using identical genie signals $\mathcal{S}_2=\{w_{i1},\ms w_{i3},\ms \tilde{\bt{x}}_{ib}\}$ and $\mathcal{S}_3=\{w_{i1},\ms w_{i2},\ms \tilde{\bt{x}}_{ic}\}$ for users $i2$ and $i3$ respectively, we obtain the following two inequalities:

\begin{align}
 nR_{sum} 
& \stackrel{}{\leq} nM\log \rho +R_{i1}+R_{i3}+h(\tilde{\bt{X}}_{\bar{i}b}|\hat{\bt{X}}_{\bar{i}a},\hat{\bt{X}}_{\bar{i}c})\nonumber \\
&+n\ms o(\log \rho)+o(n), \label{B2}\\
nR_{sum} 
& \stackrel{}{\leq} nM\log \rho +R_{i1}+R_{i2}+h(\tilde{\bt{X}}_{\bar{i}c}|\hat{\bt{X}}_{\bar{i}a},\hat{\bt{X}}_{\bar{i}b})\nonumber \\
&+n\ms o(\log \rho)+o(n).
\label{B3} 
\end{align}

Adding the inequalities in (\ref{B1}), (\ref{B2}) and (\ref{B3}), we get 
\begin{align}
3nR_{sum} \leq &  3nM\log\rho +\sum_{j=1}^{3}2nR_{ij}+\sum_{j=1}^{3}nR_{\bar{i}j} \nonumber\\
&+n\ms o(\log \rho)+o(n).
\end{align}

Using a similar set of genie signals for users in cell $\bar{i}$, we can establish a corresponding inequality on the sum-rate. Adding these two inequalities gives us
\begin{align}
6nR_{sum} \stackrel{}{\leq} & 6Mn\log \rho+3nR_{sum}+n\ms o(\log \rho) +o(n).
\label{U3}
\end{align}
Letting $n \rightarrow \infty$ and $\rho \rightarrow \infty$, we see that DoF/user $\leq \tfrac{M}{3}$.

%% file: finerdetails.tex
\label{finerdetails}

In this section we provide further details on the linear beamforming strategy used to achieve the optimal sDoF for the two-cell two-users or three-users per cell MIMO cellular networks. We consider designing transmit beamformers in the uplink. By duality of linear interference alignment, the same strategy also holds in downlink.
\setcounter{subsubsection}{0}
\subsubsection{Linear Beamforming Strategy for the Two-Cell, Two-Users/Cell Network}
 We divide the discussion in this section into six cases, each corresponding to one of the six distinct piece-wise linear regions in Fig. \ref{fig_22}. Since we assume generic channel coefficients, we do not need to explicitly check to make sure that (a) interference and signal are separable at each BS and (b) signal received from a user at the intended BS occupies sufficient dimensions to ensure all data streams from that user are separable (i.e., $\bt H_{(ij,i)}\bt V_{ij}$ is full rank for all $i$ and $j$). We however need to ensure that the set of beamformers designed for a user are linearly independent. 

    \textit{Case i: $ 0 < \gamma \leq 1/4$}: Each user here requires $M$ DoF. It is easy to observe that since $N\geq4M$, random uplink transmit beamforming in the uplink suffices. The BSs have enough antennas to resolve signal from interference. Note that no spatial extensions are required here.

    \textit{Case ii: $ 1/4 \leq \gamma \leq 1/2$}: The goal here is to achieve $N/4$ DoF/user. If $N/4$ is not an integer, we consider a space-extension factor of four, in which case we have $4M$ antennas at the users and $4N$ antennas at the transmitter. Since we need $N$ DoF/user and the BSs now have $4N$ antennas, we once again see that random uplink transmit beamforming suffices.

    \textit{Case iii: $ 1/2 < \gamma < 2/3$}: Since each user requires $M/2$ DoF/user, we consider a space-extension factor of two so that there are $2M$ antennas at each user and $2N$ antennas at each BS. The two users in the second cell each have access to a $2M$ dimensional subspace at the first BS. These two subspaces overlap in $4M-2N$ dimensions. Note that since $\gamma > 1/2$, $4M> 2N$, such an overlap almost surely exists. The two users in cell 2 pick $4M-2N$ linear transmit beamformers so as to span this space and align their interference. Specifically, the transmit beamformers $\bt v_{21j}$ and $\bt v_{22j}$ for $j=1,\hdots,(4M-2N)$ are chosen such that 
    \begin{align}
    \bt H_{(21,1)}\bt v_{21j}&=\bt H_{(22,1)}\bt v_{22j} \nonumber \\
    \Rightarrow \begin{bmatrix} \bt H_{(21,1)} & -\bt H_{(22,1)} \end{bmatrix}
    \begin{bmatrix} \bt v_{21j} \\ \bt v_{22j} \end{bmatrix}&=0.
    \label{align_eq}
    \end{align}
    The $4M-2N$ sets of solutions to (\ref{align_eq}) can be generated using the expression $ (\bt I-\bt A^H(\bt A\bt A^H)^{-1}\bt A)\bt r$ where $\bt A=\begin{bmatrix} \bt H_{(21,1)} & -\bt H_{(22,1)} \end{bmatrix}$ and $\bt r$ is a random vector. Adopting the same strategy for cell 1 users, we see that at both BSs interference occupies $4M-2N$ dimensions while signal occupies $8M-4N$ dimensions, with $8N-12M$ unused dimensions. Note that since $\gamma\leq 2/3$, $8N-12M\geq 0$. Letting each user pick $2N-3M$ random beamformers, the remaining $8N-12M$ dimensions are equally split between interference and signal at each of the BSs. We have thus designed $M$ transmit beamformers for each user while ensuring that at each BS, interference occupies no more than $(4M-2N)+2(2N-3M)=2N-2M $ dimensions, resulting in $M/2$ sDoF/user.

    \textit{Case iv: $2/3 \leq \gamma \leq 1$}: We need to achieve $N/3$ DoF/user. We consider a space-extension factor of three, so that each user has $3M$ antennas and each BS has $3N$ antennas; and we need to design $N$ transmit beamformers per user. The two users in the second cell each have access to a $3M$ dimensional subspace at the first BS. These two subspaces overlap in $6M-3N$ dimensions. Since $\gamma > 2/3$, we note that $6M-3N > N$, allowing us to pick a set of $N$ transmit beamformers such that interference is aligned at BS 1. Using the same strategy for users in cell 1, interference and signal together span $3N$ dimensions. The transmit beamformers can be computed by solving the same set of equations as given in (\ref{align_eq}).

    \textit{Case v: $1 < \gamma < 3/2$}: In order to achieve $M/3$ DoF/user, we consider a space-extension factor of three and design $M$ beamformers per user. Since we now have more transmit antennas than receive antennas, transmit zero-forcing becomes possible. Each user in cell 2 picks $3M-3N$ linearly independent transmit beamformers so as to zero-force BS 1, i.e., the beamformers are chosen from the null space of the channel $\bt H_{(2i,1)}$ and satisfy
    \begin{align}
    \bt H_{(2i,1)}\bt v_{2ij}=0 \ \forall \ i \in \{1,2\},\ j \in \{1,2, \hdots (3M-3N) \}.
    \end{align}
    We let users in cell 1 use the same strategy. Now, in order to achieve $M$ DoF/user, we still need to design $3N-2M$ transmit beamformers per user. So far, both BSs do not see any interference and have $6M-6N$ dimensions occupied by signals from their own users. The remaining $9N-6M$ dimensions at each BS need to be split in a $2:1$ ratio between signal and interference to achieve $M$ DoF/user. To meet this goal, we choose the remaining $3N-2M$ transmit beamformers for users in cell 2 such that the interference from these users aligns at BS 1. This is accomplished by solving for the transmit beamformers using (\ref{align_eq}) for users in cell 2, and using a similar strategy for users in cell 1, resulting in $(3M-3N)+(3N-2M) = M$ DoF/user over a space-extension factor of three.

    \textit{Case vi: $3/2 \leq \gamma $}: Assuming a space-extension factor of two, each user needs $N$ transmit beamformers. The null space of the channel from a user in cell 2 to BS 1 spans $2M-2N$ dimensions and since $\gamma>3/2$, $2M-2N>N$. Choosing $N$ transmit beamformers from such a null space and using the same strategy for users in cell 1, we see that each BS sees no interference and hence is able to completely recover signals from both of its users. 

\subsubsection{Linear Beamforming Strategy for the Two-cell, Three-Users/Cell Network }
We divide the discussion in this section into ten cases, each corresponding to one of the ten distinct piecewise-linear regions in Fig. \ref{fig_22}. The cases $\gamma<1/6$ and $1/6\leq \gamma \leq 1/3$ 
and $\gamma \geq 4/3$ are identical to cases $(i)$, $(ii)$ and $(vi)$ in the previous section, where either random transmit beamforming or zero-forcing achieve the optimal DoF. We omit the discussion of these three cases here.

\textit{Case iii: $\frac{1}{3} < \gamma < \frac{2}{5}$:} We consider a space extension factor of two and prove that $M$ DoF/user are achievable. Since $4M<2N$, a many-to-one type of alignment between multiple interfering vectors is not possible. However, since $6M>2N$, it is possible to design a set of three beamformers, one for each user in a cell, such that the beamformers occupy only two dimensions at the interfering BS. In particular, to design beamformers for the three users in cell 2, we solve the following system of equations 
\begin{align}
 \left [ \bt H_{(21,1)} \ \bt  H_{(22,1)} \ \bt H_{(23,1) } \right ] 
\begin{bmatrix} \bt v_{21j} \\ \bt v_{22j} \\ \bt v_{23j} \end{bmatrix}= \bt 0.
\end{align}
Note that this is a system of $2N$ equations in $6M$ unknowns, and there can be at most $6M-2N$ linearly independent solutions. These solutions yield $6M-2N$ sets of three beamformers, with each set having a packing ratio of $3:2$. While the  $6M-2N$ solutions to the system of equations are linearly independent, we need to prove that the $6M-2N$ beamformers designed for each user are also linearly independent. In other words, linear independence of the set of solutions $\{[ \hat{\bt v}_{21j}^T \  \hat{\bt v}_{22j}^T \  \hat{\bt v}_{23j}^T ] \}_{j=1}^{6M-2N}$ does not immediately imply the linear independence of the set $\{ \hat{\bt v}_{2ij} \}_{j=1}^{6M-2N}$ for all $i \in \{1,2,3 \}$. We prove through a contradiction that this is indeed true. Suppose that the set $\{[ \hat{\bt v}_{21j}^T \  \hat{\bt v}_{22j}^T \  \hat{\bt v}_{23j}^T ] \}_{j=1}^{6M-2N}$ is linearly independent, but the set   $\{ \hat{\bt v}_{2ij} \}_{j=1}^{6M-2N}$ is not, for some $i$. Without loss of generality, let $i=1$. Then, there exist a set of coefficients $\{\beta_j\}$ such that 

\begin{align}
 \sum_{j=1}^{6M-2N}\beta_j \hat{\bt v}_{21j}=0.
\end{align}

Let $\hat{\bt w}$ denote the vector $\sum_{j=1}^{6M-2N}\beta_j [\hat{\bt v}_{21j}^T \ \hat{\bt v}_{22j}^T \ \hat{\bt v}_{23j}^T]^T$. Then, 

\begin{align}
 \left [ \bt H_{(21,1)} \ \bt H_{(22,1)} \ \bt  H_{(23,1) } \right ] 
\hat{\bt w}=& \bt 0, \\
\Rightarrow \left [ \bt H_{(22,1)} \ \bt H_{(23,1) } \right ] 
\hat{\bt w}(M+1:3M)=& \bt 0.
\label{eq_lindep}
\end{align}
Equation (\ref{eq_lindep}) is a system of $N$ equations and $2M$ unknowns, and since 
$2M<N$, (\ref{eq_lindep}) is satisfied only if  $\hat{\bt w}(M+1:3M)=\bt 0$ $\Rightarrow$ $\hat{\bt w}=0$ $\Rightarrow$ the set $\{[ \hat{\bt v}_{21j}^T \  \hat{\bt v}_{22j}^T \  \hat{\bt v}_{23j}^T ] \}_{j=1}^{6M-2N}$ is linearly dependent, which is a contradiction.

Using the $6M-2N$ sets of beamformers obtained in this manner, we note that at each BS, we have $18M-6N$ dimensions occupied by signal, $12M-4N$ dimensions occupied by interference with $12N-30M$ unoccupied dimensions. We now pick $2N-5M$  random beamformers for each user so as to use all available dimensions at both the BSs. Since the second set of beamformers are chosen randomly, they are linearly independent from the first set of $6M-2N$ beamformers almost surely. We have thus ensured each user achieves $M$ DoF using a space extension factor of two.

\textit{Case iv: $\frac{2}{5} \leq \gamma \leq \frac{1}{2}$:} In order to achieve $N/5$ DoF/user, we consider a space extension factor of five and consider designing $N$ transmit beamformers per user. Once again, $3:2$ is the highest possible packing ratio and there are $15M-5N$ sets of three beamformers (one for each of three user in a cell) having this packing ratio. If we are to use all such beamformers, we can at most cover $5(15M-5N)$ dimensions at each BS. Since $5(15M-5N)\geq 5N$, we have sufficient number of such sets to use all available dimensions at the two BSs. Choosing $N$ such sets of beamformers achieves $N$ DoF/user over five space extensions.

\textit{Case v: $\frac{1}{2} < \gamma < \frac{5}{9}$: } The goal here is to achieve $2M$ DoF/user using a space extension factor of five. To keep the presentation simple, we assume $M$ and $N$ are divisible by five and achieve 2M/5 DoF/user. Since $2M>N$, many-to-one type of interference alignment becomes feasible and in fact, $2:1$ is the highest possible packing ratio. There are three ways to choose a pair of users from a cell, and for each pair there exist $2M-N$ sets of beamformers having a packing ratio of $2:1$. For users in cell 2, these beamformers can be formed by solving equations of the form 

\begin{align}
[\bt H_{(2i,1)} \bt H_{(2k,1)}]
\begin{bmatrix} \bt v_{2ij} \\ \bt v_{2kj} \end{bmatrix}= \bt 0,
\label{eq_21pratio}
\end{align}
where $i,k\in \{1,2,3\},\ i\neq k$. We thus have $2(2M-N)$ beamformers per user. Since we assume channels to be generic and since $2(2M-N)<M$, the set of $2(2M-N)$ beamformers are almost surely linearly independent. When these $6(2M-N)$ beamformers are used by users in each cell, each BS has $4N-6M$ unused dimensions. We fill the unused dimensions using beamformers having the next best packing ratio---$3:2$. In order to ensure the linear independence of this new set of beamformers from the set of beamformers already designed, we multiply each channel matrix $\bt H_{(lm,n)}$ with a matrix $\bt W_{lm}$ on the right,  where $\bt W_{lm}$ is a $M \times (2N-3M)$ matrix whose columns are orthogonal to the $4M-2N$ beamformers that have already been designed for user $lm$. Let the effective channel matrix ${\bt H}_{(lm,n)}\bt W_{lm}$ be denoted by $\tilde{\bt H}_{(lm,n)}$. Note that $\tilde{\bt H}_{(lm,n)}$ is a $N \times 2N-3M$ matrix and since 
$3(2N-3M)>N$, there exist beamformers having packing ratio $3:2$. Similar to Case iv, we design $2N-\frac{18M}{5}$ sets of such beamformers, ensuring that all dimensions at the two BSs are used while achieving 
$(2N-\frac{18M}{5})+2(2M-N)=2M/5$ DoF/user.

\textit{Case vi: $\frac{5}{9} \leq \gamma \leq \frac{2}{3}$: } We need to achieve $2N$ DoF/user over 9 spatial extensions. To keep the presentation simple, we simply assume that $N$ is divisible by nine and present the arguments without any spatial extensions. Since $2M>N$, beamformers having packing ratios $2:1$ exist. We have $3(2M-N)$ sets of such beamformers per cell, and using any $N/3$ (note that $(N/3)<3(2M-N)$) of them ensures that all dimensions at both the BSs are occupied by either interference or signal.

\textit{Case vii: $\frac{2}{3} < \gamma < \frac{3}{4}$:} This case is discussed in detail in Section \ref{SAP} and we only mention the exact equations and transformations necessary to design the required beamformers. For users in cell 2, the $3M-2N$ sets of beamformers having packing ratio $3:1$ are designed by solving the system of equations given by

\begin{align}
\begin{bmatrix}
\bt H_{(21,1)} & \bt H_{(22,1)} & \bt 0 \\
\bt 0 & \bt H_{(22,1)} & \bt H_{(23,1)} 
\end{bmatrix}
\begin{bmatrix}
\bt v_{21j} \\
\bt v_{22j} \\
\bt v_{23j} 
\end{bmatrix}
=\bt 0.
\label{eq_31pratio}
\end{align}
We use an analogous set of equations for users in cell 1 and denote the set of beamformers designed in this manner using the set $\{ \hat{\bt v}_{ikj} \}_{j=1}^{3M-2N}$ for all $i\in \{1,2\}$ and $k\in \{1,2,3 \}$. We then multiply each channel matrix $\bt H_{ik,l}$ on the right by a matrix $\bt W_{ik}$, where $W_{ik}$ is a $M\times (2N-2M)$ matrix whose columns are orthogonal to the set $\{ \hat{\bt v}_{ikj} \}_{j=1}^{3M-2N}$. Letting the effective channel matrix be denoted by $\tilde{\bt H}_{ik,l}$, we see that we now have $2N-2M$ effective antennas at each user and the best possible packing ratio is $2:1$. There exist $3(3N-4M)$ pairs of beamformers having a packing ratio of $2:1$, and solving for any $3N-4M$ pairs using equation (\ref{eq_21pratio}) allows us to achieve the requisite number of DoF/user.

\textit{Case viii: $\frac{3}{4} \leq \gamma \leq 1$}. Our goal is to achieve $N/4$ DoF/user. We assume $N$ to be divisible by four and present the arguments without any explicit reference to spatial extensions. Since $3M>N$, packing ratio of $3:1$ is possible and there exists a total of $3M-2N$ such sets of beamformers. Designing any $N/4$ such sets through (\ref{eq_31pratio}) gives us the requisite number of DoF/user.

\textit{Case ix: $1 < \gamma < 4/3 $} We need to design $M/4$ DoF/user, and we assume that $M$ is a multiple of four. Note that since $M>N$, the users can now zero-force the interfering BS. Each user can design $M-N$ transmit beamformers such that the interfering BS sees no interference. As before, we then multiply the channel matrices $\bt H_{ik,l}$ by a $M\times 2N-M$ matrix $\bt W_{ik}$ that is orthogonal to the $M-N$ transmit beamformers obtained from zero-forcing. We now have $2N-M$ effective antennas at each user and it is easy to see that there exist $4N-3M$ sets of transmit beamformers having packing ratio of $3:1$ for such a system. Designing any $N-\frac{3M}{4}$ sets of such beamformers through (\ref{eq_31pratio}) lets us achieve $M/4$ DoF/user.